\documentclass[journal]{IEEEtran}

\usepackage[blocks]{authblk}
\usepackage{amsmath}
\usepackage{amssymb}
\usepackage[colorlinks=true,linkcolor=blue,breaklinks]{hyperref}
\usepackage{fullpage}
\usepackage{graphicx}
\usepackage{amsthm}
\usepackage{dsfont}
\usepackage{array}
\usepackage{multicol}
\usepackage[utf8]{inputenc}
\usepackage[T1]{fontenc}
\usepackage{MnSymbol}
\usepackage{multicol}
\usepackage{pdfpages}
\usepackage{tikz}
\usepackage{bbm}
\usepackage{cleveref}
\usepackage[all]{xy}
\usepackage{braket}
\usepackage{cite}
\usepackage{setspace}
\usepackage{verbatim}
\usepackage{thm-restate}
\usepackage{enumerate}
\usepackage{todonotes}
\usepackage[top=2cm, bottom=1.8cm,left=2.5cm,right=2.5cm,head=2cm,headsep=0.5cm]{geometry}
\usepackage{fancyhdr} %
\usetikzlibrary{arrows}
\usetikzlibrary{decorations.pathmorphing}

\usepackage{tikz-cd}
\newcommand{\cb}{\operatorname{cb}}
\newcommand{\outerp}[2]{\ket{#1}\!\bra{#2}}
\newcommand{\cA}{\mathcal{A}}
\newcommand{\ZZ}{\mathbb{Z}}
\newcommand{\RR}{\mathbb{R}}
\newcommand{\CC}{\mathbb{C}}
\newcommand{\cK}{\mathcal{K}}
\newcommand{\cF}{\mathcal{F}}
\newcommand{\cB}{\mathcal{B}}
\newcommand{\LL}{\mathcal{L}}
\newcommand{\cL}{\mathcal{L}}
\newcommand{\ntr}{\operatorname{tr}}
\newcommand{\e}{\operatorname{e}}
\newcommand{\tr}{{\operatorname{Tr}}}
\newcommand{\cH}{\mathcal{H}}
\newcommand{\cD}{\mathcal{D}}
\newcommand{\cP}{\mathcal{P}}
\theoremstyle{plain}
\newtheorem{theorem}{Theorem}[section]
\newtheorem{lemma}[theorem]{Lemma}
\newtheorem{proposition}[theorem]{Proposition}
\newtheorem{corollary}[theorem]{Corollary}

\theoremstyle{definition}
\newtheorem{definition}[theorem]{Definition}

\newtheorem{hypothesis}[theorem]{Hypothesis}
\theoremstyle{remark}
\newtheorem{remark}[theorem]{Remark}

\newcommand{\Hcal}{\mathcal{H}}

\newcommand{\Lcal}{\mathcal{L}}

\newcommand{\Dcal}{\mathcal{D}}

\newcommand{\Bcal}{\mathcal{B}}

\newcommand{\cE}{\mathcal{E}}

\newcommand{\cT}{\mathcal{T}}

\newcommand{\Fcal}{\mathcal{F}}

\newcommand{\Rbb}{\mathbb{R}}
\newcommand{\Cbb}{\mathbb{C}}

\newcommand{\Ebb}{\mathbb{E}}

\newcommand{\NN}{\mathbb{N}}

\newcommand{\HC}{\operatorname{HC}}

\newcommand{\LSI}{\operatorname{LSI}}

\newcommand{\Nfix}{{N_{fix}}}

\newcommand{\norm}[1]{\left\| #1 \right\|}

\newcommand{\Id}{\mathds{1}}
\newcommand{\eps}{\varepsilon}

\newcommand{\Ind}{\mathds{1}}

\newcommand{\EE}{\mathbb{E}}
\newcommand{\Tr}{\text{Tr}\,}

\newcommand{\id}{\operatorname{id}}

\newcommand{\cM}{\mathcal{M}}

\newcommand{\cU}{\mathcal{U}}

\newcommand{\la}{\lambda}
\newcommand{\si}{\sigma}
\newcommand{\al}{\alpha}
\newcommand{\ten}{\otimes}
 
\newcommand{\ssubset} {\!\!\subset\! \!}

\newcommand{\kl}{\pl \le \pl}
\newcommand{\gl}{\pl \ge \pl}

\newcommand{\lel}{\pl = \pl}

\renewcommand{\L}{\mathcal{L}}

\newcommand{\rz}{{\mathbb R}}

\newcommand{\Mz}{{\mathbb M}}

\newcommand{\cz}{{\mathbb C}}

\newcommand{\pl}{\hspace{.1cm}}

\newcommand{\qd}{\end{proof}\vspace{0.5ex}}

\newcommand{\om}{\omega}

\newcommand{\Si}{\Sigma}

\newcommand{\C}{{\mathcal C}}

\newcommand{\pf}{\begin{proof}}

\newcommand{\xspace}{\hbox{\kern-2.5pt}}
\newcommand{\xyspace}{\hbox{\kern-1.1pt}}

\newcommand{\lb}{\langle \langle }
\newcommand{\rb}{\rangle \rangle }

\ifCLASSINFOpdf
\else
\fi
\begin{document}
\title{Group transference techniques for\\ the estimation of the decoherence times \\and capacities of quantum Markov semigroups.}

\author{
    \IEEEauthorblockN{Ivan Bardet\IEEEauthorrefmark{1}, Marius Junge\IEEEauthorrefmark{2}, Nichcolas Laracuente\IEEEauthorrefmark{2}, Cambyse Rouz\'{e}\IEEEauthorrefmark{3}, Daniel Stilck Fran\c{c}a\IEEEauthorrefmark{4}\\}
    \IEEEauthorblockA{\IEEEauthorrefmark{1}Institut National de Recherche en Informatique et en Automatique, Paris, France
    \\}
    \IEEEauthorblockA{\IEEEauthorrefmark{2}Department of Mathematics
University of Illinois at Urbana-Champaign
Champaign, IL, USA\\
}
    \IEEEauthorblockA{\IEEEauthorrefmark{3}Faculty of Mathematics,
Technische Universit\"at M\"unchen,
Munich, Germany\\
}

\IEEEauthorblockA{\IEEEauthorrefmark{4}Centre for the Mathematics of Quantum Theory,
University of Copenhagen, Copenhagen, Denmark\\
}
}

\maketitle

\begin{abstract}
Capacities of quantum channels and decoherence times both quantify the extent to which quantum information can withstand degradation by interactions with its environment. However, calculating capacities directly is known to be intractable in general. Much recent work has focused on upper bounding certain capacities in terms of more tractable quantities such as specific norms from operator theory. In the meantime, there has also been substantial recent progress on estimating decoherence times with techniques from analysis and geometry, even though many hard questions remain open. In this article, we introduce a class of continuous-time quantum channels that we called \emph{transferred channels}, which are built through representation theory from a classical Markov kernel defined on a compact group. In particular, we study two subclasses of such kernels: H\"ormander systems on compact Lie-groups and Markov chains on finite groups. Examples of transferred channels include the depolarizing channel, the dephasing channel, and collective decoherence channels acting on $d$ qubits. Some of the estimates presented are new, such as those for channels that randomly swap subsystems. We then extend tools developed in earlier work by Gao, Junge and LaRacuente to transfer estimates of the classical Markov kernel to the transferred channels and study in this way different non-commutative functional inequalities. The main contribution of this article is the application of this transference principle to the estimation of decoherence time, of private and quantum capacities, of entanglement-assisted classical capacities as well as estimation of entanglement breaking times, defined as the first time for which the channel becomes entanglement breaking. Moreover, our estimates hold for non-ergodic channels such as the collective decoherence channels, an important scenario that has been overlooked so far because of a lack of techniques.
\end{abstract}

\IEEEpeerreviewmaketitle

\section{Introduction}
\IEEEPARstart{I}{n} quantum mechanics, the evolution of a global, closed system leads to a unitary evolution of states. However, any realistic quantum system undergoes dissipative dynamics, due to its unavoidable interaction with its surrounding environment. Understanding how this noise limits the usefulness  of these systems for various information processing tasks is of central importance to the development of quantum technologies.

The dynamics of open systems are modeled by completely positive trace preserving maps. In the Markovian approximation, continuous time evolutions are then modeled by quantum Markov semigroups $(T_t)_{t\ge 0}$ of such maps. Given a concrete quantum Markov semigroup, it is then important to identify short time versus long time behaviour of the evolution. For example, it is important to know how long entanglement can be preserved. This remains a challenging problem. Even for classical systems, precise decoherence time estimates are very delicate, see \cite{saloff1994precise,DSC93group}. The aim of this paper is to obtain some `concrete' estimates on the decoherence time of such dissipative evolutions, as well as to derive bounds on various capacities using classical and quantum functional inequalities.

In the classical setting, the connection between functional inequalities and decoherence times is very well-established, (see e.g.~\cite{saloff1994precise,Diaconis1996a,Ledoux}), and many works started to establish the connections in the quantum case in recent years. But many of the techniques only work for semigroups with a unique invariant state, and actually proving such inequalities for quantum systems remains challenging.
In this paper, we start mending this gap by showing how to obtain various quantum functional inequalities starting from a classical one.
To make the connection between classical and quantum Markov semigroups, we consider mixed unitary quantum channels in which the unitaries form a representation of a group, and their weights come from a classical Markovian process on the group. We call such semigroups transferred semigroups, and they include widely studied error models, such as depolarizing or dephasing quantum channels.
We are particularly interested in non-ergodic semigroups of channels, that is, semigroups admitting more than one invariant state. These semigroups have been known to play an important role in various quantum error prevention schemes \cite{lidar2013quantum}. Unfortunately, even the literature for non-ergodic semigroups in the commutative community is relatively sparse. However, we will show how to directly translate classical results to the quantum setting if the underlying classical dynamics is ergodic, even if the quantum semigroup is not.

Besides the obvious application of  decoherence times estimate, we use this transference principle to obtain entropic inequalities for these semigroups.  By showing some tensorization results and that the inequalities remain valid when we tensorize the underlying semigroup with the identity channel, we use them to estimate several different capacities of these semigroups, such as the (two-way) private and quantum capacity, classical capacity and the entanglement assisted classical capacity. Moreover, these bounds are in the strong converse sense and have the right asymptotics.

The geometry of the underlying space (i.e. of the compact Lie group or the finite group) is crucial for these a priori estimates, in particular for concrete estimates. However, it should be noted that these inequalities are not sharp in general, as they do not depend on the representation at hand. Furthermore, as we will observe later, one may obtain the same quantum channel from transfering Markov kernels from different groups, which can lead to more or less pertinent estimates.

To exemplify the power of these methods, consider collective channels. These are quantum channels in which the same error occurs in different registers. The simplest examples for collective channels are derived from the standard Pauli matrices, and correspond to the same Pauli error occurring on different qubits at the same time. These quantum channels clearly do not have a unique invariant state and thus, it is difficult to quantify how fast they mix using current techniques in the literature. Nevertheless, using group transference, we will be able to estimate decoherence times of these channels \emph{independently} from the number of qubits, and similar estimates hold for their capacities. Moreover, we will derive some estimates that are new even in the classical literature. This will be the case for quantum channels that randomly swap subsystems.

Inspired by the techniques of~\cite{gao2018fisher}, we also show a smorgasbord of functional inequalities for these semigroups, such as spectral gap, hypercontractivity, logarithmic Sobolev inequalities and ultracontractive estimates. The latter three are particularly important to obtain good estimates for small times, before the spectral gap kicks in and closes the deal. We also exemplify why in the non-transference case, the correct notion for applications of the decoherence-time is a \emph{complete} version, i.e. where we consider the semigroup tensorized with the identity on a matrix algebra.

Finally, let us point out that in contrast to \cite{gao2018fisher}, which focuses on Lie groups and Horm\" ander systems (where very good estimates are available from the fundamental work of Rothschild and Stein~ \cite{Rostein}), we are also interested in finite groups and jump processes, as mixed unitary channels with unitaries arising from a representation of a finite group play an important role in quantum information theory.

\paragraph{Layout of the paper:} In \Cref{sect1}, we introduce the framework of quantum Markov semigroups and explain their connection to classical diffusions and jump processes on groups via the so-called \textit{transference technique}. In \Cref{sect2}, we explain the technical tools that allow us to bound various norm estimates of a quantum Markov semigroup in terms of the kernel of an associated classical process: namely, noncommutative $L_p$ spaces and the norm transference technique. \Cref{examplessec} is devoted to some examples of transferred semigroups to illustrate the technique. We then illustrate how to use our techniques to estimate capacities and apply them to other resource theories in \Cref{capa}.
Then, in \Cref{sect3}, we show how contractivity properties (in particular, ultracontractivity) of the quantum Markov semigroups can also provide a way to estimates some capacities in the quantum case, without the use of transference.

%
%

%
%

%
%

%
%
%
%
%
%
%
%
%
%
%
%
%
%
%
%
%
%
%
%
%
%

%
%

%
%
%
%
%
%
%
%
%
%
%
%
%
%
%
%
%
%
%
%
%
%
%
%
%
%

%
%
%
%
%
%
%
%
%
%
%
%
%
%
%
%
%
%
%
%
%
%
%
%
%
%
%

%
%
%
%
%
%
%
%
%

\section{Quantum Markov semigroups via group transference}\label{sect1}

A \emph{quantum Markov semigroup} (QMS) $(T_t)_{t\geq0}$ on $\Bcal(\Hcal)$ is a uniformly continuous semigroup of completely positive maps such that $T_0=\id$ and $T_t(I_\Hcal)=I_\Hcal$ for all $t\geq0$ \cite{alicki2007quantum}. The limit $\cL=\lim_{t\to0}(\id-T_t)/t$ exists and is called the \emph{Lindblad generator}. We insist on our convention that consequently $T_t=e^{-tL}$ \emph{with a minus sign}! This is not the most often used convention in the quantum case but it is more consistent with the classical situations we will consider.\\

The QMS $(T_t)_{t\geq0}$ models the evolution of observables in the \emph{Heisenberg picture}. In the dual \emph{Schr\"odinger picture}, one is instead interested in the evolution of states of density matrices. We recall that a density matrix $\rho\in\Bcal(\Hcal)$ on $\Hcal$ is a trace-one positive semi-definite operator. We denote by $\cD(\cH)$ the set of density matrices on $\Hcal$ and by $\cD_+(\Hcal)$ the set of invertible (full-rank) density matrices.

\noindent We shall mainly (but not only) study self-adjoint (or symmetric) QMS for the Hilbert Schmidt scalar product:
\[\Tr[T_t(x^*)\,y]=\Tr[x^*\,T_t(y)]\qquad\forall x,y\in\Bcal(\Hcal)\,,\forall t\geq0\,.\]
This is equivalent to the fact that $T_t=T_t^\dagger$, where $T^\dagger$ is the adjoint of $T_t$ with respect to the Hilbert-Schmidt inner product. This also implies that the maximally mixed state is an invariant state: $T_t(\frac{I_\Hcal}{d_\Hcal})=\frac{I_\Hcal}{d_\Hcal}$.\\

Form now on, we always assume that \emph{the maximally-mixed state is an invariant state}. Then, the set of fixed-points $N_{fix}$ of the QMS becomes an algebra, as proved by Frigerio in \cite{frigerio1978}. It is defined by:
\[N_{fix}=\{x\in\Bcal(\Hcal)\,;\,T_t(x)=x\quad\forall t\geq0\}\,.\]
Let $E_{fix}$ be the orthogonal projection on $\Nfix$ for the Hilbert-Schmidt scalar product. As proved in the same article, it is a conditional expectation in the sense of operator algebra, that is $ E_{fix}(a\,x\,b)=a\,E_{fix}(x)\,b$ for all $a,b\in\Nfix$ and $x\in\Bcal(\Hcal)$. Remark also that as an orthogonal projection, it is self-adjoint: $E_{fix}=E_{fix}^\dagger$. We denote by $\Dcal(\Nfix)\equiv E[\cD(\cH)]$ the image of the density matrices for this conditional expectation: one has $\rho\in\cD(\Nfix)$ if and only if $\rho=E_{fix}[\rho]$. Selfajoint QMS are in particular \emph{ergodic}, in the sense that:
\begin{equation}\label{eq_def_ergodic}
	T_t(x)\underset{t\to+\infty}{\longrightarrow} E_{fix}(x)\qquad\forall x\in\Bcal(\Hcal)\,.
\end{equation}
This can for instance be seen by considering the spectrum of the QMS, which in this case is real with no peripheral eigenvalue (see \cite{wolftour} for instance).\\

We now proceed to the presentation of the class of QMS we shall study in this article. We start in \Cref{sect11} by introducing the general method based on group transference, which allows to build a QMS from a (classical) symmetric Markov semigroup on a group with right invariant kernel. In \Cref{sect12} and \Cref{sect13} we specialize this discussion to two classes of Markov semigroups: H\"ormander diffusions and jumps. On the other hand, given a QMS, we show in \Cref{sect14} how to find a Markov semigroup for which the QMS can be transferred. This construction can be interpreted as a finer version of the characterization of quantum convolution semigroups of \cite{K-M4}.

\subsection{General construction}\label{sect11}

The starting point is a compact group $G$, either Lie or finite, with Haar measure $\mu_G$ (we shall simply write $\mu$ when there is no ambiguity). Let $(S_t)_{t\geq0}$ be a Markov semigroup on the space $L_\infty(G)$ of bounded, measurable functions on $G$. We will always assume that $(S_t)_{t\geq0}$ admits the following kernel representation:
\begin{equation}\label{eq_conv_kernel}
	S_t(f)(g)=\int_G\,k_t(g,h)\,f(h)\,d\mu_G(h)\,.
\end{equation}
We also assume that $(S_t)_{t\ge 0}$ is right-invariant, which means that the probability to visit $h$ from $g$ only depends on $gh^{-1}$. This implies that $\mu_G$ is an invariant probability distribution and that $k_t(g,h)=k_t(gh^{-1},\e)$, where $\e$ is the neutral element of the group. We keep the same notation $k_t(g)$ for $k_t(g,e)$.

Let $g\mapsto u(g)$ be a projective representation of $G$ on some finite dimensional Hilbert space $\Hcal$. We define the following convolution QMS on $\Bcal(\Hcal)$ which we call a \emph{transferred QMS}:
\begin{equation}\label{eq_groupQMS}
	T_t(x)=\int_G\,k_t(g^{-1})\,u(g)^*\, x\, u(g)\,d\mu_G(g)\,.
\end{equation}
At the root of the transference techniques that we study in this article is a factorization property between $(S_t)_{t\geq0}$ and $(T_t)_{t\geq0}$, involving the standard co-representation
\[\pi\,:\,\Bcal(\Hcal)\to L_\infty \left(G,\Bcal(\Hcal)\right)\,,\,\,\,\pi(x)(g)=u(g)^* x u(g)\,.\]
The following lemma, which is a special case of a result from \cite{gao2018fisher}, is at the heart of the transference method. In particular, it will allow us to obtain contraction properties of a transferred quantum Markov semigroup in terms of the ones of the classical Markov semigroup from which it is transferred.
\begin{lemma}[Lemma 4.6 in \cite{gao2018fisher}]\label{lem_factorization}
	The following relation holds for all $t\geq0$:
	\begin{equation}\label{eq_fact}
		\pi\circ T_t=(S_t\otimes \id_{\Bcal(\Hcal)})\circ\pi\,.
	\end{equation}
\end{lemma}
\begin{proof}
	We recall the proof for sake of completness. We have, for any $x\in\cB(\cH)$,
	\begin{align*}
		\pi&\circ T_t(x)(g)\\
		&=u(g)^*\int_Gk_t(h^{-1})\,u(h^{-1})\,x\,u(h)\,d\mu_G(h)\,u(g)\\
		&=\int_G k_t(gg^{-1}h^{-1})\,u((hg)^{-1})\,x\,u(hg)\,d\mu_G(h)\\
		&=\int_Gk_t(gh^{-1})\,u(h^{-1})\,x\,u(h)\,d\mu_G(h)\\
		&=(S_t\otimes \id_{\cB(\cH)})(\pi(x))(g)\,.
		\qedhere \end{align*}
\end{proof}	
From the invariance of $\mu_G$, one can also easily verify that any QMS $(T_t)_{t\ge 0}$ transferred from $(S_t)_{t\ge 0}$ is \textit{doubly stochastic}: $T_{t}^\dagger(d_\cH^{-1} I_{\cH})=d_\cH^{-1}\,I_{\cH}$ for any $t\ge 0$. On the other hand, the reversibility of $(S_t)_{t\ge 0}$ is transferred to the QMS $(T_t)_{t\ge 0}$:
\begin{lemma}\label{lem_QMS_self-adjoint}
	Assume that the Markov semigroup $(S_t)_{t\ge 0}$ is reversible, or equivalently that $k_t(g)=k_t(g^{-1})$ for any $g\in G$. Then any QMS $(T_t)_{t\ge 0}$ transferred from $(S_t)_{t\ge 0}$ is self-adjoint with respect to $d_{\cH}^{-1}\,I_\cH$.
\end{lemma}	

\begin{proof}
	The result follows from the simple calculation:
	\begin{align*}
		\langle x,&\,T_t(y)\rangle_{\operatorname{HS}}\\
		&=\int_G\,k_t(g^{-1})\,\tr(x^*\,u(g)^*\,y\,u(g))\,d\mu_G(g)\\
		&=\int_Gk_t(g^{-1})\,\tr((u(g)\,x\,u(g)^*)^*\,y)\,d\mu_G(g)\\
		&=\int_Gk_t(g^{-1})\,\tr((u(g)^*\,x\,u(g))^*\,y)\,d\mu_G(g)\\
		&=\langle T_t(x),\,y\rangle_{\operatorname{HS}}\,,
	\end{align*}	
	where the third line follows from the identity $k_t(g)=k_t(g^{-1})$ for all $g\in G$.
\end{proof}
Since $d_\cH^{-1}\,I_\cH$ is an invariant state of $(T_t)_{t\ge 0}$, the set $N_{fix}$ of fixed points is an algebra (see \cite{Spohn1977,frigerio1978}, Theorem 6.12 of \cite{wolftour}), and is characterized as the commutant of the projective representation (see Theorem 6.13 of \cite{wolftour}):
\begin{equation}\label{eq_fixedpoints}
	N_{fix} = \{ \sigma \in \cB(\cH)\,|\forall g\in G: \sigma u(g)=u(g)\sigma\} \equiv u(G)' \,.
\end{equation}

By definition, it is also the algebra of fixed points of the $*$-automorphisms $x\mapsto u(g)^*\,x\,u(g)$, $g\in G$. This implies that the following commuting diagram holds:
\begin{center}
	\begin{tikzcd}
		\cB(\cH) \arrow[r, "E_{fix}"] \arrow[d, "\pi"]
		& N_{fix}\arrow[d, "\pi"]  \\
		{L}_\infty(G,\cB(\cH))  \arrow[r, "\EE_{\mu_G}"]
		& \cB(\cH)\,.
	\end{tikzcd}
\end{center}
Here $\Bcal(\Hcal)$ in the lower right corner has to be understood as the subalgebra of constant value functions on $G$ with value in $\Bcal(\Hcal)$.

In practice, we will only consider situations where the classical Markov semigroup $(S_t)_{t\geq0}$ is \emph{primitive}, that is, $\mu_G$ is the unique invariant distribution and furthermore
\[S_t f\underset{t\to+\infty}{\longrightarrow}\, \EE_{\mu_G}[f]=\int_G f(g)\,d\mu_G(g)\,.\]
This does not imply that $(T_t)_{t\geq0}$ is also primitive, however, it will always be ergodic as defined in \Cref{eq_def_ergodic}.\\

We now turn our attention to two special cases of the above construction. In both cases, we explicitly construct the Lindblad generator of the QMS.

\subsection{Diffusion}\label{sect12}

Given a Riemannian manifold $\cM$, a  \textit{H\"{o}rmander system} on $\cM$ is a set of vector fields $V=\{V_1,...,V_m\}$ such that, at each point $p\in \cM$, there exists an integer $K$ such that the iterated commutators $\big[V_{i_1}\,,\,\left[V_{i_2}\,,\,\cdots[V_{i_k}\,,\,\cdot]\,\right] \big]$, $k=1,...,K$, generate the tangent space $T_p\,\cM$. Specializing to the case of a Lie group $G$, a H\"ormander system $V=\{V_1,...,V_m\}$ can more simply be defined as a set of vectors in the Lie algebra, i.e. the tangent space at the neutral element $e$, such that for some $K\in\NN$ the iterated commutators of order at most $K$
span the whole tangent space. For fixed $j\in \{1,...,K\}$, we find a geodesic $g_j(t)$ with $g_j(0)=e$ such that for any $f\in C^1(G)$
\[ V_j(f)(h) = \left.\frac{d}{dt} f(g_j(t)\,h)\, \right|_{t=0} \,  .\]
This leads to the corresponding left invariant classical generator
\begin{align}\label{thegeneratorclassicaldiff1}
	L_V := -\sum_j V_j^2  \, .
\end{align}
The generator $L_V$ generates a Markov semigroup $P_t=e^{-t L_V}$ on ${L}_{\infty}(G)$. Whenever $V$ is a basis for the Lie algebra, then $L_V\equiv -\Delta$ is the negative of the Laplacian and $(P_t)_{t\geq0}$ is called the \emph{heat semigroup}.\\
Since the semigroup commutes with the right action of the group it is implemented by a right-invariant convolution kernel as in \Cref{eq_conv_kernel} and it is reversible with respect to the Haar measure.\\

Next, considering a projective representation $g\mapsto u(g)$ of $G$ on some finite dimensional Hilbert space $\Hcal$, we want to find the Lindblad generator of the QMS defined by \Cref{eq_groupQMS}. We first observe that, for fixed $j\in\{1,...,K\}$ and given the geodesic $g_j$ associated to the vector field $V_j$, $u(g_j(t))$ is a one-parameter family of unitaries and hence
\begin{equation}\label{eq_Hormander_generator}
	\left.\frac{d}{dt} u(g_j(t))\right|_{t=0} = i\,a_j
\end{equation}
where $a_j\in\Bcal(\Hcal)$ is self-adjoint. This implies that, for any $x\in\cB(\cH)$,
\begin{align*}
	(V_j\otimes \id_{\cB(\cH)}&\circ \pi(x)(g) \\
	&=
	\left. \frac{d}{dt}  \pi(x)(g_j(t)g) \right|_{t=0} \\
	&= \left.\frac{d}{dt} u(g)^*\,u(g_j(t))^*\,x\, u(g_j(t))\,u(g) \right|_{t=0} \\
	&= -i \pi([a_j,x])(g) \,.
\end{align*}
Therefore we get
\begin{align*}
( L_V\otimes\id_{\cB(\cH)})&\circ\pi(x)\\
&=
- \sum_j i^2 \,\pi([a_j,[a_j,x]])\\
&= \pi(\sum_j a_j^2\,x+x\, a_j^2-2a_j\,x \,a_j) \, .
\end{align*}
It means that the Lindblad generator of the transferred QMS is given by
\[  \L_V(x)=\sum_j a_j^2\,x+x\, a_j^2-2\,a_j\,x\, a_j\,.\]
Considering the more general case of H\"ormander systems instead of the Laplacian is motivated by simple examples relevant to quantum information. One such example is the Lindblad generator with Kraus operators $\sigma_x$ and $\sigma_z$. These are transferred from a H\"ormander system for the group $SU(2)$. However, the third direction is missing and hence it is not a Laplace-Beltrami Laplacian, which would involve the whole orthogonal basis of the Lie algebra.\\

Conversely, if a Lindblad generator of a QMS on $\Bcal(\Hcal)$ has the form given by the previous equation for some self-adjoint elements $a_j\in\Bcal(\Hcal)$,  then we can consider the anti-self-adjoint operators $i\,a_j$ as tangent elements of the Lie group $\mathcal{U}(\Hcal)$ at the identity $I_\Hcal$. Therefore they generate a H\"ormander system. Furthermore, we can consider the Lie-subgroup $G$ of $\mathcal U(\Hcal)$ with tangent space at identity spanned by this H\"ormander system. The corresponding generator $L_V$ is therefore the generator of a primitive Markov semigroup $(S_t)_{t\geq0}$ on $L_\infty(G)$. \\

We summarize this discussion in the next theorem.
\begin{theorem}\label{theo_QMS_diffusive}
	Let $g\mapsto u(g)$ be a projective representation of a compact Lie group $G$ on some finite dimensional Hilbert space $\Hcal$. Then the Lindblad generator of the transferred QMS $(T_t=e^{-t\Lcal_V})_{t\geq0}$ as defined by \Cref{eq_groupQMS} is given by
	\begin{equation}\label{eq_theo_QMS_diffusive}
		\L_V(\rho)=\sum_j a_j^2\,\rho+\rho\, a_j^2-2\,a_j\,\rho\, a_j\,,
	\end{equation}
	where the $a_j$ are defined by \Cref{eq_Hormander_generator}. Conversely, let $\L$ be the Lindblad generator of a QMS on $\Bcal(\Hcal)$ which takes the form \eqref{eq_theo_QMS_diffusive} for some self-adjoint elements $a_j$ in $\Bcal(\Hcal)$. Then there exists a compact Lie group $G$, a continuous projective representation $u:G\to U(\Hcal)$ and a H\"ormander system $V=\{V_1,...,\,V_m\}$ in the Lie algebra of $G$ such that $\pi:x\mapsto (g\mapsto u(g)^*\,x\,u(g))$ satisfies
	\[ \pi(\L(x))\lel (L_V\otimes \id_{\cB(\cH)})\circ \pi(x)\qquad\forall x\in\Bcal(\Hcal)\,.\]
\end{theorem}

\subsection{Jumps}\label{sect13}

Let now $G$ be a finite group and let $(k_t(g,h))_{g,h\in G}$ be a right-invariant density kernel on $G$. Let us denote by $(g_t)_{t\geq0}$ the stochastic process on $G$ induced by this kernel. The corresponding Markov semigroup admits a transition matrix $L$ such that $S_t=\e^{-tL}$ for all $t\ge 0$. In view of \Cref{eq_conv_kernel}, the connection between the Markov kernel and the transition matrix is therefore given by:
\[k_t(g,h)=|G|\,e^{-tL}(g,h)\,,\qquad g,h\in G\,.\]
Writing $c_h=-L(h^{-1},\e)$ for all $h\ne e$, we then have by right invariance that for all $f\in L_\infty(G)$,
\[ L(f)(g) \lel -\sum_{h\in G} c_h [f(hg)-f(g)] \pl \,.\]
Thanks to the right-invariance, we can define a family of independent Poisson processes $\left((\tilde N_t^h)_{t\geq0}\right)_{h\in G}$ with intensity $c_h$ such that for any function $f\in L_\infty(G)$:
\[f(g_{t})-f(g_{t^-})=\sum_{h\in G}\,\left(f(hg_{t^-})-f(g_{t^-})\right)\, \left(\tilde N_t^h-\tilde N_{t^-}^h\right)\,.\]
Define the compensated Poisson process with intensity $c_h$ and jumps $1/\sqrt{c_h}$, $h\ne e$:
\[N_t^h=\frac{1}{\sqrt{c_h}}\,(\tilde N_t^h-c_h t)\,.\]
Writing $df(g_t):=f(g_{t})-f(g_{t^-})$ and $d N_t^h:= N_t^h- N_{t^-}^h$, we can rewrite the previous equation as the stochastic differential equation:
\begin{align}\label{eq_SDE_jumps}
	df(g_t) &= \sum_{h\in G\backslash \{e\}}\,c_h\left(f(hg_{t^-})-f(g_{t^-})\right)\,dt\\
	&+\sum_{h\in G\backslash \{e\}}\,\sqrt{c_h}\left(f(hg_{t^-})-f(g_{t^-})\right)\,d N_t^h\nonumber \\
	&= \sum_{h\in G\backslash \{e\}} (f(hg_{t-})-f(g_{t-}))d\tilde{N}_{t}^h \, .  \nonumber
\end{align}
We are now ready to build a QMS from this Markov chain. Let $g\mapsto u(g)$ be a projective representation of $G$ on some finite dimensional Hilbert space $\Hcal$. We want to find a stochastic differential equation for $\left(u(g_t)\right)_{t\geq0}$. To this end, take $y\in\Bcal(\Hcal)$ and define $f_y:h\in G\mapsto \Tr[y\,u(h)]$. Applying \Cref{eq_SDE_jumps} to $f_y$ we find
\begin{align*}
df_y(g_t)&=\sum_{h\in G\backslash \{e\}}\,c_h\,\Tr\left[y\,\left(u(hg_{t^-})-u(g_{t^-})\right)\right]\,dt\\
&+\sum_{h\in G\backslash \{e\}}\,\sqrt{c_h}\,\Tr\left[y\,\left(u(hg_{t^-})-U(g_{t^-})\right)\right]\,dN_t^h\,.
\end{align*}
From this we deduce
\begin{align}
	du(g_{t}) 
	& =\sum_{h\in G\backslash \{e\}}\,c_h\,\left(u(h)-I_\Hcal\right)\,u(g_{t^-})\,dt\nonumber\\
	&+\sum_{h\in G\backslash \{e\}}\,\sqrt{c_h}\,\left(u(h)-I_\Hcal\right)\,u(g_{t^-})\,dN_t^h\,\nonumber \\
	& = \sum_{h\in G\backslash \{e\}}\,\left(u(h)-I_\Hcal\right)\,u(g_{t^-})\,d\tilde N_t^h\,.\label{eq_QSDE_jumps}
\end{align}
This equation is well-known in the theory of quantum stochastic calculus, see \cite{H-P,Mey2}.

\begin{theorem}\label{theo_QMS_jumps}
	Let $(S_t=e^{-tL})_{t\geq0}$ be a Markov semigroup on a finite group $G$ with right-invariant Markov kernel. Write $c_g=-L(g^{-1},\e)$ for $g\ne e$. Then the generator of the QMS $(T_t)_{t\geq0}$ defined by \Cref{eq_groupQMS} is given for all $x\in\cB(\cH)$ by
	\begin{equation}\label{eq_theo_QMS_jumps}
		\L(x)=\sum_{g\in G\backslash \{e\}}\,c_g\,\big(x-u(g)^*\,x\,u(g)\big)\,,\qquad x\in\Bcal(\Hcal)\,.
	\end{equation}
	Furthermore, ${d_\Hcal}^{-1}I_\Hcal$ is an invariant density matrix and if $(S_t)_{t\geq0}$ is reversible, then so is $(T_t)_{t\geq0}$.\\
	Conversely, let $\L$ be a Lindblad generator on $\Bcal(\Hcal)$ of the form
	\[ \L(x) \lel \sum_{k=1}^m c_k (x -u_k^*\, x\, u_k)\,,\qquad x\in\Bcal(\Hcal), \]
	for some unitary operators $u_k\in\cU(\Hcal)$ and some positive constants $c_k$. Assume that the group $G$ generated by  $u_1,...,u_m$ is finite and define
	\[ L(f)(g) \lel -\sum_{k=1}^m c_k [f(u_kg)-f(g)] \pl .\]
	Then $L$ is the generator of a primitive Markov semigroup $(S_t=e^{-tL})_{t\geq0}$ on the oriented graph
	\[ E \lel \{(g,u_kg)\,|\, k=1,...,m\,;\,g\in G\}\,. \]
	Furthermore, the map $\pi:\Bcal(\Hcal)\to L_\infty(G,\Bcal(\Hcal))$ defined by
	\[ \pi(x)(k) \lel u_k^*\, x\, u_k \]
	extend to a $*$-representation of $\Bcal(\Hcal)$ on $L_\infty(G,\Bcal(\Hcal))$ such that $(L\otimes \id_{\cB(\cH)})\circ\pi(x)=\pi\circ\L(x)$ for all $x\in\Bcal(\Hcal)$.
\end{theorem}

\begin{proof}
	We begin by proving \Cref{eq_theo_QMS_jumps}. By definition, we have for all $x\in\Bcal(\Hcal)$
	\[\cL(x)=-\left.\frac{d}{dt} \, T_t \, \right|_{t=0}=-\left.\frac{d}{dt}\Ebb\big[\left(u(g_t)^*\,x\,u(g_t)\right)]\right|_{t=0}\,.\]
	\Cref{eq_theo_QMS_jumps} follows from an application of the \^Ito formula for compensated Poisson processes. The fact that ${d_\Hcal}^{-1}I_\Hcal$ is an invariant density matrix is straighthforward as clearly
	\[\cL^\dagger\left(\frac{I_\Hcal}{d_\Hcal}\right)=0\,,\]
	where $\cL^\dagger$ is the adjoint of $\L$ for the Hilbert-Schmidt scalar product. The case where $(S_t)_{t\geq0}$ is reversible is the content of \Cref{lem_QMS_self-adjoint}. The second part of the proof is straightforward from what preceded.
\end{proof}

\begin{remark} We should warn the reader that a selfadjoint Lindblad generator may sometimes be both transferred from a compact Lie group or a finite group. For example, the partial depolarizing generator $\cL(\rho)=\rho-\frac{X\rho X+Y\rho Y}{2}$ is both discrete and continuous in this sense.
\end{remark}

\subsection{The general situation}\label{sect14}

The two cases explored above are particular instances of convolution QMS as defined by Kossakowski in \cite{Kossakowski1972}. Such QMS were then entirely characterized by Kümmerer and Maassen in \cite{K-M4}, both in terms of their Lindblad generator and as the QMS having an essentially commutative dilation. We recall the first characterization.

\begin{theorem}[Theorem 1.1.1 in \cite{K-M4}]\label{theo_convQMS}
	Let $(T_t)_{t\geq0}$ be a QMS on $\Bcal(\Hcal)$. The two following assertions are equivalent.
	\begin{enumerate}
		\item There exists a weak$^*$-continuous convolution semigroup $(\rho_t)_{t\geq0}$ of probability measures on the group $\operatorname{Aut}\,(\Bcal(\Hcal))$ of automorphisms on $\Bcal(\Hcal)$ such that
		\[T_t(x)=\int_{\operatorname{Aut}\,(\Bcal(\Hcal))}\,\alpha(x)\,d\rho_t(\alpha)\,,\qquad x\in\Bcal(\Hcal)\,.\]
		\item The Lindblad generator $\cL$ of $T$ takes the form
		\begin{align}
			\cL(x)&=-i[h,x]+\sum_{j=1}^n\,\frac12\left(a_j^2\,x+x\,a_j^2-a_j\,x\,a_j\right)\nonumber\\
			&+\sum_{i=1}^m\,\kappa_i\left(x-u_i^*\,x\,u_i\right)\,,\label{eq_theo_convQMS}
		\end{align}
		where $h$ and the $a_j$ are self-adjoint operators in $\Bcal(\Hcal)$, where the $u_i$ are unitary operators on $\Hcal$ and where the $\kappa_i$ are positive real numbers.
	\end{enumerate}
\end{theorem}
Remark that in quantum information terms, the theorem above characterizes the generators of quantum dynamical semigroups consisting of mixed unitary channels.
The generators of the form given by \Cref{eq_theo_convQMS} are thus the sum of three parts:
\begin{itemize}
	\item The first part corresponds to a unitary evolution with generator given by $\Bcal(\Hcal)\ni x\mapsto i[h,x]$ where $h$ is self-adjoint;
	\item A diffusive part, given by
	\[\Bcal(\Hcal)\ni x\mapsto \sum_{j=1}^n\,\frac12\left(a_j^2\,x+x\,a_j^2-2a_j\,x\,a_j\right)\,,\]
	where the $a_j$ are self-adjoint operators. Any such family $\{a_j\}$ is a H\"ormander system for the sub-Lie algebra that they generate, as elements of the unitary group $ \mathcal{U}(\Hcal)$ of $\Hcal$. Consequently the result of \Cref{sect12} applies.
	\item A jump part, given by
	\[\Bcal(\Hcal)\ni x\mapsto \sum_{i=1}^m\,\kappa_i\left(x-u_i^*\,x\,u_i\right)\,,\]
	where the $u_i$ are unitary operators on $\Hcal$. Compared to previously, this class is larger than the one presented in \Cref{sect12}. Indeed, the family $\{u_i\}$ spans a subgroup of the unitary group $ \mathcal{U}(\Hcal)$, however, in general, it will not be a finite group.
\end{itemize}

\begin{remark}
	\begin{enumerate}\
		\item Starting with a Lindblad generator, there may be an ambiguity on the choice of the underlying group and classical Markov semigroup leading to it. Indeed, in the jump scenario when the QMS is self-adjoint, it is always possible to write the Lindblad generator as in the diffusive case. Then either the group is large, i.e. the commutator is $\Cbb\,I_{\cH}$, and we can treat it as an H\"ormander system, or the group is small (for us finite) and we can treat it as a Markov semigroup with jumps on the Cayley graph of the group. In both cases, estimates on the decoherence time of the corresponding QMS can be found. We shall illustrate this fact in \Cref{examplessec}.
		\item It should be clear that the construction of QMS in \Cref{sect11} is in essence different from the one of convolution QMS. In the former, we can start from \emph{any} compact group with a Markov kernel. Then we shall see that from the $*$-corepresentation $\pi$ and \Cref{lem_factorization}, we can transfer certain properties of this kernel to the induced QMS. The existence of this $*$-corepresentation, which was absent in \cite{K-M4} and only discovered in \cite{gao2018fisher}, stands at the root of this transference principle.
	\end{enumerate}
\end{remark}

\subsection{Collective decoherence}

Motivated by applications in quantum information theory, we shall study a particular class of transferred QMS. These QMS are particularly relevant in the study of fault-tolerant passive error correction as they display non-trivial decoherence-free subsystems, that is, subsystems preserved from dissipative effects. The interesting QMS are therefore non-primitive (with non-trivial fixed-point algebras). Let $G$ be a group and $u:G\to \mathcal{B}(\mathcal{H})$ a projective representation of $G$ on some finite dimensional Hilbert space $\Hcal$. For all $n\geq 1$, this representation induces a new representation on $\Hcal^{\otimes n}$ given by:
\[g\mapsto u(g)^{\otimes n}\,.\]
Let $(S_t)_{t\geq0}$, $(T_t)_{t\geq0}$ be defined as in \Cref{eq_conv_kernel,eq_groupQMS} using the representation $g\mapsto u(g)$. We write $(T^{(n)}_t)_{t\geq0}$ the corresponding QMS on $\Hcal^{\otimes n}$ for the representation $u^{\otimes n}$ and $\L_n$ its generator.

\paragraph{Diffusive case:} In the diffusive case presented in \Cref{sect12}, the generator $\cL$ of $(\cP_t)_{t\geq0}$ has the following form:
\begin{align}\label{diffusivecollectivedeco}
	\cL(x)=\,\sum_k\,a_k^2\,x+x\,a_k^2-2a_k\,x\,a_k\,,
\end{align}
where the $a_k$'s are selfajdoint operators on $\Hcal$. Then the generator $\cL_n$ takes the form
\begin{align*}
 \cL_n(x) = \sum_k \,a_k(n)^2\,x+x\, a_k(n)^2 -2\,a_k(n)\,x\, a_k(n)\,,
 \end{align*}
 with
 \begin{align*}
a_k(n) = \sum_{j=1}^n I_\Hcal^{\otimes j-1}\otimes  a_k\otimes I_{\cH}^{n-j} \,,
 \end{align*}
where in the $j$th term of the above sum, $a_k$ acts on the $j$th copy of $\Hcal$.

More generally, if $\Lcal$ is the generator of a QMS on $\Bcal(\Hcal^{\otimes n})$ of the above form, then any $i\,a_k$ belongs to the tangent space at identity of some unknown compact Lie group, hence the family satisfies the transference principle and the different results presented in this article can be applied. As a consequence, we obtain bounds independent of the number $n$ of qudits.
\paragraph{Jump case:} In the jump case presented in \Cref{sect13}, the generator $\cL$ of $(T_t)_{t\geq0}$ has the following form:
\begin{align}\label{jumpcollectivedeco}
	\cL(x)=\sum_{k=1}^m \,c_k\left(x-u_k^*\,x\,u_k\right)\,,
\end{align}
where the $u_k$'s are unitary operators on $\Hcal$. Then the generator $\cL_n$ takes the form
\[\cL_n(x)=\sum_{k=1}^m c_k\left(x-v_k^*\,x\,v_k\right)\,,~~~\text{where }~~v_k=u_k^{\otimes n}\,.\]
If the unitary operators $u_k$ generate a finite group $G$ then thanks to \Cref{theo_QMS_jumps} we can find a Markov semigroup on $G$ and all the estimates we find on this semigroup can be transferred to $(T_t^{(n)})_{t\ge 0}$ for all $n$.
\begin{remark}
	Unfortunately, it is not the decoherence time or any other interesting quantity for $\cL$ itself which transfers to all the $\cL_n$, but the underlying group which gives the corresponding estimates. Thus, the choice of the group and the classical Markov semigroup on it are particularly important.
\end{remark}
\section{Noncommutative $L_p$ spaces and norm transference}\label{sect2}

In this section we introduce the main conceptual ideas of this article that we called \emph{group transference techniques}. These ideas and subsequent mathematical results are mostly contained in \cite{gao2018fisher}. All the applications we study in this article are concerned with the properties of certain (non-commutative) functional $L_p$ spaces. When studying primitive QMS, only the usual (normalized Schatten) $L_p$ spaces are required. However, in our case we are interested in \emph{non-primitive} QMS with non-trivial fixed-point algebra. As first illustrated in \cite{BR18}, the relevant $L_p$ spaces in this case are the \emph{conditioned} or \emph{amalgamated} $L_p$ spaces. Furthermore, the transference techniques require to look at the amplification of the classical semigroup $(S_t)_{t\geq0}$ to the algebra $L_\infty(G)\otimes\Bcal(\Hcal)$ (see \Cref{lem_factorization}). This in turn makes it necessary to consider completely bounded version of the $L_p$ spaces. All these notions are introduced in \Cref{Lpnorm}. \Cref{normtransference} is dedicated to the presentation of the transference techniques. We present them in a general framework, as we believe they can also be useful in other settings (see \cite{gao2018fisher} for an other example of application in quantum information theory). Finally we specialize to QMS in \Cref{sect23}, where these transference techniques are applied to transfer estimates on the classical Markov kernel to the QMS.

\subsection{$L_p$ norms and entropies}\label{Lpnorm}
We are now going to introduce several $L_p$ norms and entropies related to von Neumann algebras. Although this may not be clear at first sight, it turns out that many of them are just the sandwiched R\'enyi entropies~\cite{muller2013quantum,strongconvrenyi} in disguise, as we will clarify.
In the following $M$ is a finite von Neumann algebra and $\tau:M\to \cz$ a normalized normal, faithful, tracial state (i.e. $\tau(I_{\cH})=1$).
Let us recall the definition of the noncommutative $L_p$ spaces via
\begin{align}\label{equ:noncommutativenorm}
	\\|x\|_{L_p(\tau)} := [\tau(|x|^p)]^{1/p} \pl.  
\end{align}
Then $L_p(M,\tau)\equiv L_p(M)$ is the completion of $M$ with respect to this norm. Indeed, for $1\le p\le \infty$ the space is a Banach space such that $L_p(M,\tau)^*=L_{\hat{p}}(M,\tau)$ holds for $\frac1p+\frac{1}{\hat p}=1$ and $1\le p<\infty$. In this article we will focus on three types of von Neumann algebras:
\begin{itemize}
	\item Our main example is $M=\Mz_m$, the space of $m\times m$ matrices over the field of complex numbers, and $\tau(x)\equiv\tau_m(x)\equiv\frac{1}{m}\tr(x)$. To keep the notations at a more abstract level, we shall most of the time refer to a finite dimensional Hilbert space $\Hcal$ and to the algebra of (bounded) linear operators $\Bcal(\Hcal)$ and we denote by $L_p(\Bcal(\Hcal))$ the corresponding non-commutative $L_p$ space.
	\item If $(E,\Fcal,\mu)$ is a probability space, where $\Fcal$ is a $\sigma$-algebra on the set $E$ and $\mu$ a probability distribution, then the set of bounded complex-valued function  $M=L_\infty(\mu)$ is a von-Neumann algebra, $\tau:f\mapsto\mathbb{E}_\mu(f)$ is a normal, faithful and tracial state and the corresponding $L_p$ spaces are the usual $L_p(\mu)$.
	\item The last key example in the transference principle is the algebra of bounded $\Mz_m$-valued function on a probability space $(E,\Fcal,\mu)$, $M=L_\infty(E,\Mz_m)$, with trace given by
	\[\tau:f\mapsto \int_E\,\tau_m(f(x))\,d\mu(x)\,.\]
\end{itemize}
For a subalgebra $N\subset M$ we define the conditioned $L_p^q(N\,\ssubset \,M)$ norm \cite{Pis93,JP10} (see also \cite{gao2017strong,BR18}) via
\begin{align}\label{Lpnorms}
&	\|x\|_{L_p^q(N\,\,\ssubset\,\, M)}\\
	& := \begin{cases}
		\inf_{x=ayb} \|a\|_{L_{2r}(\tau)}\,\|y\|_{L_q(\tau)}\,\|b\|_{L_{2r}(\tau)} & p\le q \pl,\\
		\sup_{\|a\|_{L_{2r}(\tau)}\,\|b\|_{L_{2r}(\tau)} \le 1}\, \|a\,x\,b\,\|_{L_q(\tau)} & p\gl q \pl .
	\end{cases}\nonumber
\end{align}
Here $\frac1r=\left|\frac1q-\frac1p\right|$ and $a,b$ are elements in $L_{2r}(N)$. For $N=M$, we just find another description of $L_p(M)$, i.e.~$L_p^q(N\,\ssubset \,M)=L_p(M)$. Note that for a selfadjoint element $x$, we may assume $a=a^*$ in (\ref{Lpnorms}). By H\"{o}lder's inequality, $L_p^q(N\,\subset\, M)\subset L_p(M)$. In the particular case when $M=\Mz_k(N)$ is the algebra of $k$ by $k$ matrices with coefficients in $N$, the spaces $L_p^q(N\subset M)\equiv S_p^k(L_q(N))$ coincide with Pisier's vector-valued $L_p$ spaces \cite{Pis93}. We will also be concerned with norms of linear maps between these $L_p$ spaces. A map $T:L_p(M)\to L_q(M)$ is called a $N$-bimodule map if, for any $a,b\in N$ and any $x\in M$:
\[ T(a\,x\,b) \lel a\,T(x)\,b \pl. \]
For instance, when $N=\Nfix$ is the fixed point subalgebra of a selfadjoint quantum Markov semigroup $(T_t)_{t\ge 0}$ acting on $M$, the maps $T_t$ are $N$-bimodule maps with respect to $N$. For $N$-bimodule maps  and $p\le q$, the following was proved in Lemma 3.12 of \cite{gao2018fisher}, generalizing an earlier statement for vector valued $L_p$ norms (see Lemma 1.7 of \cite{Pis93}): for any $s\ge 1$:
\begin{align} \label{bimod}
	\|T:L_{\infty}^p(N\,\ssubset\, M)&\to L_{\infty}^q(N\,\ssubset\, M) \|\\
&	\lel \|T:L_{s}^p(N\,\ssubset \,M)\to L_{s}^q(N\,\ssubset \,M) \|  \pl \,.\nonumber
\end{align}
We refer to \cite{Pis93,JP10} for motivation and further properties. We will also use the completely bounded version of these norms:
\begin{align}
& \|T:L_s^p(N\,\ssubset\, M) \to L_s^q(N\,\ssubset\, M)\|_{\cb}\nonumber\\
&\lel \sup_{m} \|\id_m\ten T:L_s^p(\Mz_m \ten N\,\ssubset \,\Mz_m\ten M)\nonumber\\
&~~~~~~~~~~~~~~~~~~~~~~~~~~\to L_s^q(\Mz_m \ten N\,\ssubset \,\Mz_m\ten M) \| \pl .\nonumber
\end{align}
which also does not depend on $s$ for $N$-bimodule maps, as we discuss in more detail in~\Cref{sect3}.

Noncommutative $L_p$ norms are closely related to the sandwiched R\'{e}nyi divergences introduced in~\cite{muller2013quantum,strongconvrenyi}:
For $p\in(1,+\infty)$, these are defined for two quantum states $\sigma,\rho\in\mathcal{D}\left(\mathcal{H}\right)$ as:
\begin{align}
&	D_p\left(\rho\|\sigma\right)=\nonumber\\
	&
	\begin{cases}
		\frac{\log\left(\tr\left[\left(\sigma^{\frac{1-p}{2p}}\,\rho\,\sigma^{\frac{1-p}{2p}}\right)^p\right]\right)}{p-1} & 
		\text{ker}\left(\sigma\right)\subseteq\text{ker}\left(\rho\right)\text{ or }p\in \left( 0,1\right)
		\\ +\infty, & \mbox{otherwise,}
	\end{cases}
\end{align}
and, for $p=\infty$, we set
\begin{equation}
	D_\infty\left(\rho\|\sigma\right)=
	\begin{cases}
		\log\left(\|\sigma^{-\frac{1}{2}}\rho\sigma^{-\frac{1}{2}}\|_\infty\right) & 
		\text{ker}\left(\sigma\right)\subseteq\text{ker}\left(\rho\right)
		\\ +\infty, & \mbox{otherwise.}
	\end{cases}
\end{equation}
One can then show that setting $x=\sigma^{-\frac{1}{2}}\rho\sigma^{-\frac{1}{2}}$ and $\tau(x)=\tr\left(\sigma x\right)$ in~\ref{equ:noncommutativenorm}, we have:
\begin{align*}
	D_p\left(\rho\|\sigma\right)=\frac{p}{p-1}\log(\|x\|_{L_p(\tau)}).
\end{align*}
Moreover, the sandwiched R\' enyi conditional entropy $H_p(A|B)$ introduced in~\cite{muller2013quantum,strongconvrenyi} can be seen as a special case of the conditional $L_p$ norms defined in~\ref{Lpnorms}. They are defined for a bipartite state $\rho\in\mathcal{D}\left(\mathcal{H}_A\otimes\mathcal{H}_B \right)$ as:
\begin{align*}
	H_p(A|B)_\rho=-\min\limits_{\sigma_B\in\mathcal{D}\left(\mathcal{H}_B \right)}D_{p}\left(\rho_{AB}||\operatorname{id}_{\cB(\cH_A)}\otimes \sigma_B\right).
\end{align*}
It is then easy to see that
\begin{align*}
	-p'\log\left(\|\rho\|_{L_1^p(\mathcal{B}(\mathcal{H}_B)\subset\mathcal{B}(\mathcal{H}_A\otimes \mathcal{H}_B)}\right)=H_p(A|B)_{\rho}.
\end{align*}
for $p'$ the H\"older conjugate. Thus, we see that the conditional norms can also be interpreted as a generalization of conditional entropies in which we condition w.r.t. to states on a subalgebra,  which naturally includes the case of a subsystem. For two algebras $N\subset M$ we will denote this conditional R\' enyi entropy by $H_p(M|N)_{\rho}$.
Moreover, a bound on the norms defined on~\eqref{bimod} can be interpreted as a strong data processing inequality for conditional R\' enyi entropies. Indeed, given that
\begin{align*}
	&\|T:L_{\infty}^p(N\,\ssubset\, M)\to L_{\infty}^q(N\,\ssubset\, M) \|\\
	&~~~~~~~~~~~~~~~~~~~~~~=\|T:L_{1}^p(N\,\ssubset\, M)\to L_{1}^q(N\,\ssubset\, M) \|\leq c\,,
\end{align*}
we have
\begin{align*}
	\log\left(\|T(\rho)\|_{L_1^q(N\,\ssubset\, M)}\right)\leq \log(c)+ \log\|\rho\|_{L_1^p(N\,\ssubset\, M)}
\end{align*}
holds for all states $\rho$. Normalizing the expressions appropriately, we see that this implies that:
\begin{align*}
	H_q(M|N)_{T(\rho)}\geq \frac{q'}{p'}H_p(M|N)_{\rho}-\log(c),
\end{align*}
which is stronger than a data-processing inequality for $c=1$ and $q\geq p$.
\subsection{Norm transference}\label{normtransference}

In \cite{gao2017strong}, the authors proved the following factorization property: given the representation $\alpha:\,g\mapsto \alpha_g(.)=u(g)\,(.)\,u(g)^*$ of a finite or compact Lie group $G$ on the algebra $\mathcal{B}(\mathcal{H})$ of linear operators on a finite dimensional Hilbert space $\cH$, and for any $t\ge 0$, define the co-representation $\pi:\,\cB(\cH)\to L_\infty(G,\,\cB(\cH))$, $x\mapsto (g\mapsto \alpha_{g^{-1}}(x))$. Then we may transfer properties of completely positive maps on $L_{\infty}(G)$ to completely positive maps on ${\cB}(\cH)$. Indeed, for every positive function $k$ on $G$, we define
\begin{equation}\label{eq_Phi}
	\Phi_{k}(\rho) := \int k(g) \,u(g)^*\rho \,u(g) \,d\mu(g) \, .
\end{equation}
Here $\mu$ is the Haar measure. Therefore, the fixed-point algebra of the map $\Phi_k$ is given by the commutant of $u(G)$:
\[ N_{fix} = \{ \sigma \in \cB(\cH)\,| \,\sigma u(g)=u(g)\sigma\} = u(G)' \,.\]
To see this, note for any element $X\in u(G)'$ we have $u(g)^*X \,u(g)=X$, which shows $u(G)'\subset N_{fix}$. For the other inclusion, note that for any $v\in G$ and $X\in N_{fix}$:
\begin{align*}
	u(v)^*Xu(v)&=u(v)^*\Phi_{k}(X)u(v) \\
	&= \int k(g) \,u(gv)^*\rho \,u(gv) \,d\mu(g) \\
	&=\int k(g) \,u(g)^*\rho \,u(g) \,d\mu(g)\,,
\end{align*}
where in the last step we used the fact that the kernel is right-invariant. Thus, $X\in u(G)' $
Note that the following natural bimodule property holds
\[ \Phi_k(\sigma_1 \rho \sigma_2) = \sigma_1 \Phi_k(\rho) \sigma_2 \, .\]
for $\sigma_1,\sigma_2\in u(G)'$.
We then have
\begin{align*}
	\pi\circ \Phi_k=(\varphi_k\otimes\operatorname{id})\circ\pi\,,
\end{align*}
where $\varphi_k:\,L_\infty(G)\to L_\infty(G)$ is defined by
\begin{align*}
	\varphi_k(f)(g)=\int k(gh^{-1})\,f(h)\,d\mu(h)\,.
\end{align*}
We will denote by
\[ E_{fix}(\rho) \equiv E_{\Nfix}(\rho)= \int u(g)^*\,\rho \,u(g)\, d\mu(g) \]
the conditional expectation onto the fixed-point algebra. The following commuting square, already mentioned in \Cref{eq_fact} in the specific case of a Markov semigroup transference, was recently found in \cite{gao2018fisher}:
\[\begin{tikzcd}
	\cB(\cH) & L_\infty(G,\,\cB(\cH)) \arrow{l}{E_{\cB(\cH)}} \\
	\Nfix \arrow[swap]{u}{\pi} &\cB(\cH)\arrow{l}{E_{fix}} \arrow[swap]{u}{\pi}
\end{tikzcd}
\]
where $E_{\cB(\cH)}$ simply denotes the usual expectation over $G$, that is, for any $f\in L_\infty(G,\cB(\cH))$,
\begin{align*}
	E_{\cB(\cH)}(f)=\int_G f(g)\,d\mu(g).
\end{align*}	
This in particular implies that the natural inclusion \[L_p^q(\Nfix\subset \cB(\cH)) \subset L_p^q(\mathcal{B}(\cH)\subset \,L_\infty(G,\mathcal{B}(\cH))) \] is completely isometric (see \cite{JP10} for more details).\\
The next theorem constitutes the basis of the estimates that we provide in \Cref{capa}. We provide the proof in full generality in \Cref{normestimates} for sake of clarity and only present a simplified version for some cases.

\begin{theorem}\label{compar} Let $\sigma \in \cD(N_{fix})$ and $k:\,G\to\RR^+$ a measurable function such that $\int k \,d\mu=1$. Then for any $\rho\in\cD(\cH) $ and $\sigma\in\cD(N_{fix})$, and any $p\in( 1,\infty)$:
	\begin{align}
&		D_p(E_{fix}(\rho)\|\sigma)\label{ineqfundam}\\
		&~~~~~\le D_p(\Phi_k(\rho)\|\sigma)
		\leq D_p(E_{fix}(\rho)\|\sigma) + D_p(k\mu\|\mu)\,,\nonumber
	\end{align}
	where $D_p(k\mu\|\mu):=\frac{1}{p-1}\log\int\,   k^pd\mu $.	For $p=1$, this translates into
	\begin{align}\label{transfer}
	&	D(E_{fix}(\rho)\|\sigma)\\
		&~~~~~\le  D(\Phi_k(\rho)||\sigma) \leq D(E_{fix}(\rho)\|\sigma) + \int k \log k \,d\mu \, ,\nonumber
	\end{align}
	and for $p=\infty$:
	\begin{align}\label{transferinfto}
		&D_\infty(E_{fix}(\rho)\|\sigma)\\
		&~~~~~\le  D_\infty(\Phi_k(\rho)||\sigma) \leq D_\infty(E_{fix}(\rho)\|\sigma) +
		\log \|k\|_{\infty} \, .\nonumber
	\end{align}
\end{theorem}
\begin{proof}
	Proof of~\eqref{transfer} for the case of a finite group and $p=1$ is simple and gives some intuition for the idea of the proof. The first observation is that in the case of a finite group $G$, we can always dilate the channel $\Phi_k$ as
	\begin{align}\label{equ:dilationsemigroup}
		\Phi_k(\rho)=\text{Tr}_{\ell_2(G)}\left( U_G\left(\rho\otimes \tau_k \right)U_G^\dagger\right),
	\end{align}
	where
	\begin{align*}
		U_G=\sum\limits_{g\in G}U_g\otimes \ket{g}\bra{g}
	\end{align*}
	and $\tau_k\in\mathcal{D}\left(\ell_2(G)\right)$ defined as
	\begin{align*}
		\tau_k=|G|\sum\limits_{g\in G}k(g)\ket{g}\bra{g}.
	\end{align*}
	Interchanging the roles of $\rho$ and $\tau_k$, we find a channel $\Psi:L_1(G)\to \cB(\cH)^*$, $\Psi(f)=\sum_{g} f(g) u_g^*\rho u_g$. Therefore the data processing inequality shows that
	\begin{align*}
	 D(T_k(\rho)\|E_{fix}(\rho))&\kl D\left(\frac{k}{|G|}\|\frac{1}{|G|}\right) \\
	 &\lel \frac{1}{|G|} \sum_{g} k(g)\ln k(g) \pl .
	 \end{align*}
	We apply this to $k(g)=k_t(g^{-1})$, the density of this instance of the semigroup. We combine this with the following decomposition which holds for the relative entropy and a conditional expectation (see \Cref{capa}):
	\begin{align}\label{equ:pythagorean}
		D(T_t(\rho)\|\sigma)=D(E_{fix}(\rho)\|\sigma)+D(T_t(\rho)\|E_{fix}(\rho))\pl,
	\end{align}
	and deduce $D(T_t(\rho)\|\sigma)\kl D(E_{fix}(\rho)\|\sigma)+D(k_t\mu\|\mu)$, where $\mu(g)=\frac{1}{|G|}$ is the  Haar measure. 
	For $p>1$ we lack the Pythagorean identity in Eq.~\eqref{equ:pythagorean} and, thus, must use interpolation to prove the inequality. The proof for finite groups is a simplification of Theorem 2.5 from \cite{gao2017strong}. Let $\eta$ be defined such that $\eta^* \eta = \rho$. Here we use the fact that for any $\rho$,
	\begin{equation*}
		\|\Psi(f)\|_p = \bigg \| \frac{1}{|G|} \sum_g \sqrt{f(g)} \ket{g} \otimes \eta U_g \bigg \|_{2p}^2,
	\end{equation*}
	where $\Psi(1) = E_{fix}$. Let $\hat{p}$ be such that $1 = 1/p + 1/\hat{p}$. Then
	\begin{equation*}
		\| \sigma^{-1/2 \hat{p} } \Psi_k(\rho) \sigma^{-1/2 \hat{p} } \|_p = \| \Psi_k(\tilde{\rho}) \|_p
	\end{equation*}
	for $\tilde{\rho} = \sigma^{-1/2 \hat{p}} \rho \sigma^{-1/2 \hat{p}}$, because $\sigma^{-1/2 \hat{p}}$ is an element of the fixed point algebra and therefore commutes with conjugation by $U_g$ for all $g$. While $\tilde{\rho}$ is not assured to be normalized, this will not pose a problem for our calculation. For $\|\Psi_1(\tilde{\rho})\|_p < 1$, let $\tilde{\eta}$ be defined such that $\tilde{\rho} = \tilde{\eta}^* \tilde{\eta}$. There exists an analytic function $\xi : S = \{z | 0 \leq Re(z) \leq 1\} \rightarrow \cB(\cH)$ such that
	\begin{align*}
		&\bigg \| \frac{1}{\sqrt{|G|}} \sum_g \ket{g} \otimes \xi(it) u_g  \bigg \|_\infty < 1,\\
		&		\bigg \| \frac{1}{\sqrt{|G|}} \sum_g \ket{g} \otimes \xi(1+it) u_g \bigg \|_2 < 1 \text{ for all } t \in \Rbb,
	\end{align*}
	and $\xi(1/p) = \tilde{\eta}$. Assume that $\|f\|_p < 1$ and $f \geq 0$, and let
	\[V_{G, f}(z) = \frac{1}{\sqrt{|G|}} \sum_g f^{p z/2}(g) \ket{g} \otimes \xi(z) u_g . \]
	We then have that
	\begin{equation*}
		\| V_{G, f}(it) \|_\infty = \bigg \| \frac{1}{\sqrt{|G|}} \sum_g f^{i p t / 2}(g) \ket{g} \otimes \xi(it) u_g \bigg \| < 1 ,
	\end{equation*}
	and
	\begin{align*}
		\begin{split}
			& \| V_{G, f}(1 + it) \|_2 \\
			&= \bigg \| \frac{1}{\sqrt{|G|}} \sum_g f^{p (1 + it)/2}(g) \ket{g} \otimes \xi(1 + it) u_g \bigg \| \\
			&=  \bigg \| \frac{1}{\sqrt{|G|}} \sum_g f^{p/2}(g) \ket{g} \otimes \xi(1 + it) \bigg \|_2\\
			&= \bigg \| \frac{1}{\sqrt{|G|}} \sum_g f^{p/2}(g) \ket{g} \bigg \|_2 \bigg \|  \xi(1 + it) \bigg \|_2 \\
			&=  \bigg \| \frac{1}{\sqrt{|G|}} \sum_g f^{p/2}(g) \ket{g} \bigg \|_2 \bigg \| \frac{1}{\sqrt{|G|}} \sum_g \xi(1 + it) u_g \bigg \|_2\\
			& < 1 .
		\end{split}
	\end{align*}
	By Stein's interpolation theorem,
	$ \| V_{G, f}(1/p) \| _{2p} \leq 1$, completing the proof by homogeneity.
\end{proof}

It is straightforward to extend the $p=1$ proof to the case of Lie groups, but $p>1$ requires a more technical use of interpolation theory. We refer to~\Cref{normestimates} for the details. We note that in the context of a QMS with kernel $k\equiv k_t$, a variation of the inequality above is straightforward to prove given a bound on
\begin{align*}
	t(\epsilon)=\inf\{t|\,\|k_t-1\|_{\infty}\leq\epsilon\}.
\end{align*}
Indeed, it is easy to see that it follows from the definition of $t(\epsilon)$ that:
\begin{align*}
	k_t(g)U_gXU_g^\dagger\leq (1+\epsilon)U_gXU_g^\dagger
\end{align*}
holds almost everywhere. By integrating the inequality above and conjugating with $\sigma^{-1/2\hat{p}}$ we have
\begin{align*}
	\sigma^{-1/2\hat{p}}T_t(\rho)\sigma^{-1/2\hat{p}}\leq (1+\epsilon)\sigma^{-1/2\hat{p}}E_{fix}(\rho)\sigma^{-1/2\hat{p}}.
\end{align*}
It is then easy to see that this gives for any $p>1$
\begin{align*}
	D_p(T_t(\rho)||\sigma)\leq \log(1+\epsilon)+D_p(E_{fix}(\rho)||\sigma)
\end{align*}
and the statement for $p=1$ immediately follows by taking the appropriate limit. Thus, we see that a variation of the transference statement is easy to obtain if we start from a $t(\epsilon)$ bound and the bound above is only advantageous if we can control $\|k\|_{L_p(\mu_G)}$ or the relative entropy of $k$. This is for instance the case when one has access to the so-called \textit{modified logarithmic Sobolev constant} for the underlying group.
\subsection{Application to quantum Markov semigroups}\label{sect23}

In this subsection, we show how the machinery developed in \Cref{normtransference} provides estimates on the norms and entropies at the output of a quantum Markov convolution semigroup of \Cref{sect1} in terms of the kernel of their associated classical semigroup. We start by recalling the notations of \Cref{sect1}: $(S_t)_{t\geq0}$, $S_t=\e^{-tL}$ is a Markov semigroup on the compact group $G$ (either Lie or finite), with right-invariant kernel $(k_t)_{t\geq0}$. The QMS $(T_t)_{t\geq0}$, $T_t=\e^{-t\cL}$ is the transferred QMS on $\Bcal(\Hcal)$ defined by \Cref{eq_groupQMS} through the projective representation $g\mapsto u(g)$ of $G$ on the finite dimensional Hilbert space $\Hcal$.\\

The next theorem regroups all the transference techniques which will be frequently used in this paper (see \cite{gao2018fisher}). We recall that the spectral gap of the symmetric Lindblad generator $\L$, denoted by $\lambda_{\min}(\L)$ (resp. of $L$, denoted by $\lambda_{\min}(L)$) is the smallest non-zero eigenvalue of $\L$ (resp. of $L$) and is a quantity that palys a central role in the theory of decoherence times .

\begin{theorem}[Transference]  \label{Transfer} Let $T_t=e^{-t\L}$, $S_t=e^{-tL}$ as above. Then
	\begin{enumerate}
		\item[i)] The spectral gap for $\L$ is bigger than the spectral gap for $L$: $\la_{\min}(\L)\gl \la_{\min}(L)$;
		\item[ii)] For any $1\leq p,q\leq \infty$, we have $\|T_t:L_s^p(\Nfix\,\ssubset\, \cB(\cH))\to L_s^q(\Nfix\,\ssubset \,\cB(\cH))\|_{\cb}\kl \|S_t:L_p(\mu_G)\to L_q(\mu_G)\|_{\cb}$;
		\item[iii)] For any $1\leq p,q\leq \infty$, we have $\|T_t-E_{fix}:L_s^p(\Nfix\subset\, \cB(\cH))\to L_s^q(\Nfix\,\ssubset\, \cB(\cH))\|_{\cb}\kl \|S_t-\Ebb_{\mu_G}:L_p(\mu_G)\to L_q(\mu_G)\|_{\cb}$;
	\end{enumerate}
\end{theorem}

We conclude this section by briefly explaining how this theorem can be applied to get estimates for the QMS. First we recall a result on the cb norm in commutative $\C^*$-algebra (see e.g. Theorem 3.9 of \cite{[P02]}).

\begin{lemma}\label{lem_CB_commutative}
	Let $S$ be an operator on the commutative $L_p(\mu)$ space. Assume that either $p=1$ or $q=+\infty$. Then
	\[ \|S:L_p(\mu_G)\to L_q(\mu_G)\|_{\cb} \lel \|S:L_p(\mu_G)\to L_q(\mu_G)\| \pl .\]
\end{lemma}

Let us furthermore mention that for a classical Markov semigroup $(S_t)_{t\geq0}$ acting on a compact group $G$ with kernel $(k_t)_{t\geq0}$, a  convexity argument yields
\begin{equation}\label{eq_norm1_kernel}
	\|S_t:L_1(\mu_G)\to L_{\infty}(\mu_G)\|
	\lel \sup_{g\in G}\, |k_{t}(g)|
\end{equation}
and similarly
\begin{align}\label{eq_norm2_kernel}
	&\|S_t-\Ebb_{\mu_G}:L_1(\mu_G)\to L_2(\mu_G)\| \\
	&~~~~~~~~~~~~~~~~~~\lel \left(\int |k_t(g)-1|^2 d\mu_G\right)^{1/2}\nonumber\\
	&~~~~~~~~~~~~~~~~~~=\norm{k_t-\Ind}_2 \pl .\nonumber
\end{align}
This gives us ``for free'' the following estimates on the norm of the transferred QMS between certain non-commutative $L_p$ spaces. We will see in the next sections how to apply these estimates to concrete situations arising in quantum information theory.

\begin{corollary}\label{coro_example_transference}
	Let $T_t$ and $S_t$ be as above, and $\la_{\min}(\L)\gl \la_{\min}(L)$ be their respective spectral gaps.  Then for all $t\geq0$
	\begin{equation}\label{eq_comparaison_mixingtime}
		\begin{aligned}
			&\|T_{t}-E_{fix}:L_1(\cB(\cH))\to L_1(\cB(\cH))\|\\
			& \kl  \|T_{t}-E_{fix}:L_1(\cB(\cH))\to L_1(\cB(\cH))\|_{\cb} \\
			& \kl  \|S_{t}-\Ebb_{\mu_G}:L_1(\mu_G)\to L_1(\mu_G)\|\,
		\end{aligned}
	\end{equation}
	and for all $s,t\geq0$
	\begin{align}
		&\|T_{t+s}-E_{fix}:L_1(\cB(\cH))\to L_1^2(\Nfix\,\ssubset \,\cB(\cH))\|\nonumber\\
		&~~~~~~~~~~~~~~~~~~~ \kl e^{-\la_{\min}(\Lcal)s}\, \norm{k_t-\Ind}_2  \label{eq_coro_example_transference1} \\
		&~~~~~~~~~~~~~~~~~~~ \kl e^{-\la_{\min}(L)s}\, \norm{k_t-\Ind}_2 \pl .\label{eq_coro_example_transference2}
	\end{align}
	
\end{corollary}

\begin{proof} \Cref{eq_comparaison_mixingtime} is a direct consequence of the definition of the cb norm and of \Cref{lem_CB_commutative}. In order to prove \Cref{eq_coro_example_transference1,eq_coro_example_transference2}, we just note that we have
	$T_t\,E_{fix}=E_{fix}\,T_t$. This implies
	\[ 
	T_{t+s}-E_{fix}\lel (T_s-E_{fix})(T_t-E_{fix}) \pl .\]
	and $(T_{s}-E_{fix})T_t\lel T_{t+s}-E_{fix}$. Therefore, using \Cref{Transfer}, 
	\begin{align*}
		&\|T_{t+s}-E_{fix}:L_1(\cB(\cH))\to L_1^2(\Nfix\,\ssubset \,\cB(\cH))\|\\
		& \leq \|T_{t}-E_{fix}:L_1(\cB(\cH))\to L_1^2(\Nfix\,\ssubset \,\cB(\cH))\|\, \\
		&~~ \times\,\|T_{s}-E_{fix}:L_1^2(\Nfix\,\ssubset \,\cB(\cH))\to L_1^2(\Nfix\,\ssubset \,\cB(\cH))\| \\
		& \leq \|S_{t}-\Ebb_{\mu_G}:L_1(\mu_G)\to L_1(\mu_G)\|\, \\
		&~~\times\,\|T_{s}-E_{fix}:L_2(\cB(\cH))\to L_2(\cB(\cH))\| \,.
	\end{align*}
	The result follows by the definition of the spectral gap as well as \Cref{eq_norm2_kernel}.
\end{proof}

Note that the norms in the statement above are just the trace and diamond distance between the channels.
For H\"ormander systems, the following kernel estimates go back to the seminal work of Stein and Rothshield \cite{Rostein}, see also \cite{LUGEEE}.
\begin{theorem}\label{Ho} Let $V=\{V_1,...,V_m\}$ be a H\"ormander system such that $K$ iterated commutators span a Lie algebra of dimension $d$. Then $L_V$ has a spectral gap and there exists a constant $C_V>0$ such that, for all $0<t\leq1$:
	\begin{align}\label{equ:boundhormander1toinfty}
		\sup_{g\in G} \,|k_t(g)| \leq C_V\, t^{-Kd/2}
	\end{align}	
\end{theorem}
Remark that for Markov kernel on graph, estimates of the form $\sup_g|k_t(g)|\leq C_L\, t^{-\alpha/2}$ with $\alpha,c>0$ also hold in general. We shall discuss several examples of such estimates in \Cref{sect3}. From a quantum information perspective, such bounds can be interpreted as saying that the maximum output min entropy of these semigroups is bounded by the logarithm of the R.H.S. of ~\eqref{equ:boundhormander1toinfty}. That is, they quantify how fast the semigroup "spreads out" all over the group.

\section{Examples}\label{examplessec}

Here, we illustrate the method developed in the previous sections by listing examples of known QMS, that can be seen as transferred from classical semigroups. We focus on the \emph{decoherence-time} of the QMS, a notion that generalizes the mixing-time to non-primitive evolutions. We recall that for a general QMS (not necessarily selfadjoint), the decoherence time is defined for any $\eps>0$ as
\begin{align}\label{mixtimeTt}
	&t_{\operatorname{deco}}(\eps)\\&~~:=\inf\{t\ge 0;\,\|T^\dagger_t(\rho)-E^\dagger_{fix}(\rho)\|_1\le \eps~~\forall\rho\in\mathcal{D}(\mathcal{H})\}\,.\nonumber
\end{align}
In all the examples below however, the QMS is selfadjoint. We also recall the definition of the \textit{mixing time} of a classical primitive Markov process $(S_t)_{t\ge 0}$
\begin{align}\label{mixtimeSt}
	t_{\operatorname{mix}}(\eps) &= \inf\{ t\geq0\,|\, \|S_t(f)-\EE_{\mu_G}(f)\|_{L_1(\mu_G)}\,\nonumber\\
	&~~~~~~~~~~~~~~~~~\le\, \eps\,\,\forall f\ge 0,\,\EE_{\mu_G}[f]=1\} \, .
\end{align}
Remark the difference in the normalization of the norms in both definitions. In the quantum case, density matrices are normalized with respect to the unnormalized trace whereas in the classical case, we look at the evolution of states normalized with respect to the probability distribution $\mu_G$.

Then, \Cref{eq_comparaison_mixingtime} in \Cref{coro_example_transference} implies that for a transferred QMS $(T_t)_{t\geq0}$ with associated classical semigroup $(S_t)_{t\geq0}$, we have for any $\eps> 0$:
\begin{align} \label{eq_tdeco_tmix}
	t_{\operatorname{deco}}(\eps)\leq t_{\operatorname{mix}}(\eps)\,.
\end{align}
This shows that the decoherence time of a QMS is controlled by the mixing time of any classical Markov semigroup from which it can be transferred. In the examples below, the classical mixing time is estimated thanks to known results on functional inequalities that we mainly introduce in \Cref{sect3} and \Cref{examplesfinitegroups}, apart for the last example in \Cref{sect_SWAP} which is directly computed in \Cref{examplesfinitegroups}.

\subsection{The depolarizing QMS}\label{dephasing}
Perhaps the simplest QMS that one can think of is the depolarizing semigroup on $\cB(\CC^n)$:
\begin{align}\label{eq_depo}
	\cL^{\operatorname{dep}}(\rho)=\rho-\frac{I_{\,\CC^n}}{n}\,,~T^{\,\operatorname{dep}}_t(\rho)=\e^{-t}\rho+(1-\e^{-t})\,\frac{I_{\,\CC^n}}{n}\,.
\end{align}
This QMS can be seen to be transferred from the uniform walk on the group $\ZZ_n\times \ZZ_n$, via the projective representation given by the discrete Weyl matrices $\{U_{i,j}\}_{i,j\in [n]}$ (see e.g.~\cite{wolftour}). Indeed, using \Cref{eq_theo_QMS_jumps} and denoting by $\Lcal$ the transferred QMS given by this representation, we find that for all $\rho\in\cD(\Cbb^d)$,
\begin{align}
	\cL(\rho)&=\frac{1}{n^2}\sum_{i,j=1}^n\,\big(\rho-U_{i,j}\,\rho\,U_{i,j}^*\big)\\
	&=\,\rho-\tr(\rho)\,\frac{I_{\,\CC^n}}{n}\\
	&=\,\cL^{\operatorname{dep}}(\rho)\,,\label{eqnormalization}
\end{align}
where we took $c_{i,j}:=\frac{1}{n^2}$ for all $i,j$. This choice implies that the uniform random walk on the complete graph with $n^2$ vertices transfers to $(T^{\,\operatorname{dep}}_t)_{t\ge 0}$. Using the logarithmic Sobolev constant for the complete graph given in \cite{Diaconis1996}, we find the following upper bound on the mixing time of $(T_t^{\operatorname{dep}})_{t\ge 0}$:
\begin{align*}
	t_{\operatorname{deco}}(\eps)&\le 	t_{\operatorname{mix}}(\eps)\\
	&\le  n^2\frac{1-\ln \eps}{n^2-1}+\frac{n^2\ln (n^2-1)}{2(n^2-2)}\ln\ln n^2\\
	&\underset{n\to\infty}{\sim} \,\ln(n)\,\ln\ln(n)\,.
\end{align*}

This can be compared with the tighter bound that one can get from the modified logarithmic Sobolev constant $\alpha_1(\cL^{\operatorname{dep}})$ (see \Cref{capmlsi}), from which we can obtain~\cite{[KT13],mmuller2016entropy}:
\begin{align*}
	\|\rho_t-n^{-1}I_{\CC^n}\|_1\le \sqrt{2  \ln n }\e^{-\frac{t}{2}}\,,
\end{align*}
so that
\[t_{\operatorname{deco}}(\eps)\le 2 \ln\frac{\sqrt{2\ln n}}{\eps}\underset{n\to\infty}{\sim}  2\ln\ln n \,.\]

\subsection{The dephasing QMS}\label{dephasingmix}
We recall that the \textit{dephasing quantum Markov semigroup} (also called \textit{decoherent QMS} in \cite{BarEID17}) on $\cB(\CC^n)$ with $n\ge 3$, is given by
\begin{align*}
&\cL^{\text{deph}}(\rho)=\rho-E_{\text{diag}}[\rho]\,,\\
& T^{\,\text{deph}}_t(\rho)=e^{-t}\,\rho+(1-e^{-t})\,E_{\text{diag}}[\rho]\,,
\end{align*}
where $E_{\text{diag}}$ denotes the projection on the space of matrices that are diagonal in some prefixed eigenbasis. Here, we show how simple representations of the discrete and continuous torus both lead to the \textit{dephasing quantum Markov semigroup}.

\paragraph{Dephasing from the discrete torus:}
Choose the uniform random walk on $\ZZ_n$ of kernel $K(j,k)=1/n$ for any $j,k\in \ZZ_n$.
A simple unitary representation of $\ZZ_n$ is given by taking $\Hcal=\CC^n$ and
\[U_j:=U^j,~~~j\in \mathbb{Z}_n,\]
where $U$ denotes the Weyl unitary operator given by $U=\operatorname{diag}\,(1,e^{\frac{2i\pi}{n}},...,e^{\frac{2i(n-1)\pi}{n}})$ on $ \cB(\CC^n)$, where the diagonal is chosen to be the one corresponding to $E_{\operatorname{diag}}$. One can easily verify from \Cref{eq_theo_QMS_jumps} that the QMS $(T^{\,\text{deph}}_t)_{t\ge 0}$ coincides with the generator of the transferred QMS corresponding to the uniform kernel on $\ZZ_n$, since by a direct calculation $E_{\text{diag}}[X]=\frac{1}{n}\,\sum_{j\in \ZZ_{n}}U^{-j}\,X\,U^j$.\\

Now, it results from the logarithmic Sobolev constant for the uniform walk on $\ZZ_n$, computed in \cite{Diaconis1996a}, that
\begin{align*}
	t_{\operatorname{deco}}(\eps)&\le n\frac{1-\ln\eps}{n-1}+\frac{n\ln (n-1)}{2(n-2)}\ln\ln n\\
	&\underset{n\to\infty}{\sim}  \frac{\ln (n)\ln\ln (n)}{2}   \,.
\end{align*}

\paragraph{Dephasing from the $n$-dimensional torus:}
Take the representation of the $n$-dimensional torus that consists of
diagonal unitary matrices:
\begin{align*}
	\mathbb{T}^n \ni	(t_1,...,t_n)\mapsto
	\begin{bmatrix}
		\e^{2i t_1\pi} &&0 \\
		& \ddots & \\
		0&& \e^{2it_n\pi}
	\end{bmatrix}\,.
\end{align*}	
The QMS associated to the heat semigroup and the above representation corresponds to $(T_{t}^{\textrm{deco}})_{t\ge 0}$. This simply follows from \Cref{eq_theo_QMS_diffusive} by taking the generators $A_j:=|j\rangle\langle j|$ of $\mathbb{T}^n$, so that the generator of the transferred QMS is equal to
\begin{align*}
	\cL(x)&=\frac{1}{2}\,\sum_{j=1}^n\,|j\rangle\langle j|\,x+x\,|j\rangle\langle j|-2\,|j\rangle\langle j|x
	|j\rangle\langle j|\\
	&= x-\,E_{\text{diag}}[x]\,.
\end{align*}
Then the estimation (\ref{eq:estimatetorus}) leads to the following bound on the decoherence time of this QMS:
\begin{align*}
	t_{\operatorname{deco}}(\eps)\le \frac{1}{2}\ln \left( \frac{1}{2}n\ln n  \right)+6-\ln\eps  \underset{n\to\infty}{\sim}  \frac{\ln\left( n\right)}{2}\,.
\end{align*}	
Hence the estimate found on the decoherence time from the continuous torus turns out to be sharper than the one found from the discrete torus. Moreover, these two bounds can be compared with the one found via decoherence-free modified logarithmic Sobolev inequality in \cite{BarEID17}, which implies that
\begin{align*}
	\|T_t^{\operatorname{deph}}(\rho)-E_{fix}^\dagger(\rho)\|_1\le \sqrt{2  \ln n }\e^{-\frac{t}{2}}\,,
\end{align*}
so that
\begin{align*}
	t_{\operatorname{deco}}(\eps)\le 2\ln\frac{\sqrt{2\ln n}}{\eps}\underset{n\to \infty}{\sim} \ln\ln n\,.
\end{align*}	
Once again, we see that the transfer method does not immediately lead the best decoherence-time.
\subsection{Collective decoherence}
The bounds provided by the transference method for the examples studied in the last two sections, namely the depolarizing and the dephasing semigroups, are worse than the already known ones derived from the modified logarithmic Sobolev inequality. In this section, on the other hand, we show that our method provides an easy way of deriving new estimates for collective decoherence on $n$-register systems, that is on $(\Cbb^2)^{\otimes n}$. The power of the method lies in the fact that the constants derived are independent of the choice of the representation. In particular, we get estimates that are independent of the number of qubits by choosing tensor product representations.

We focus on two particular examples of collective decoherence, namely the weak and the strong collective decoherences. We first recall the definition of the Pauli matrices on $\Cbb^2$:
\begin{equation*}
	\sigma_x:=\left( {\begin{array}{cc}
			0 &1 \\
			1 &0 \\
	\end{array} } \right)\,,\,
	\sigma_y:=\left( {\begin{array}{cc}
			0 &-i \\
			i &0 \\
	\end{array} } \right)\,,\,
	\sigma_z:=\left( {\begin{array}{cc}
			1 &0 \\
			0 &-1 \\
	\end{array} } \right)\,.
\end{equation*}

\paragraph{Weak collective decoherence:}
We recall the generator of the weak collective decoherence on $n$ qubits:
\begin{align}\label{WCD}
	\cL_n^{\operatorname{wcd}}(x):=\frac{1}{2}((\sigma_z^{(n)})^2x+x\,(\sigma_z^{(n)})^2)-\sigma_z^{(n)}\,x\,\sigma_z^{(n)}\,, 
\end{align}
 where 
 \begin{align*}\sigma_z^{(n)}:=\sum_{i=1}^n I_{\CC^2}^{\otimes i-1}\otimes \sigma_z\otimes I_{\CC^2}^{n-i}\,.
\end{align*}	
One can easily show that the completely mixed state $2^{-n}I_{(\CC^2)^{\otimes n}}$ is invariant, since $\cL^{\operatorname{wcd}}_{n}(I_{(\CC^2)^{\otimes n}})=0$. Moreover, since $\sigma_z^{(n)}$ is self-adjoint, $	\cL_n^{\operatorname{wcd}}$ is selfadjoint with respect to that state. \\

The spectral gap for this QMS was computed in \cite{BR18} and found to be equal to $\lambda_{\min}(\LL_n^{\operatorname{wcd}})=2$ for any $n\geq2$. From this and the universal upper bound on the logarithmic Sobolev constants found in the same article, the authors conclude that the weak collective decoherence QMS satisfies
\begin{align}\label{boundrough}
	t_{\operatorname{deco}}(\eps)=\mathcal{O}(n)\,.
\end{align}	
We will see that the transference method leads to a better estimate. Let us first consider the heat diffusion on the one dimensional torus $\mathbb{T}^1$, which we represent on $(\CC^2)^{\otimes n}$ as follows:
\begin{align*}
	\mathbb{T}^1\ni \theta\mapsto (\e^{i\theta \sigma_z})^{\otimes n}.
\end{align*}
One can easily verify that the QMS transferred via the above representation is the weak collective decoherence semigroup (up to a rescaling of the Lindblad operators by a factor of $\sqrt{2}$) as a direct consequence of \Cref{diffusivecollectivedeco}. Then, by \Cref{eq_tdeco_tmix} and the estimate of \Cref{compactliegroups}, we find for all $t\ge 0$, and any $\rho\in\cD(\cH)$:
\begin{align*}
	\| T_{t}^{\operatorname{wcd},n}(\rho-E^\dagger_{fix}(\rho))\|_1\le \sqrt{2+\sqrt{\pi/t}}\,\e^{-\frac{t}{2}}\,,
\end{align*}	
which represents a fast convergence independent of the number $n$ of qubits of the system (remark that we set the dimension of a single copy to be $2$, but the result would depend on this dimension otherwise).
\paragraph{Strong collective decoherence:}
We recall the generator of the strong collective decoherence on $n$ qubits:
\begin{align*}
	\cL^{\operatorname{scd}}_n(x):=\sum_{i\in \{x,y,z\}}\frac{1}{2}((\sigma_i^{(n)})^2x+x(\sigma_i^n)^2)-\sigma_i^{(n)}\,x\,\sigma_i^{(n)}\,, 
\end{align*}
 where 
 \begin{align*}
 \sigma_i^{(n)}:=\sum_{k=1}^n I_{\CC^2}^{\otimes k-1}\otimes \sigma_i\otimes I_{\CC^2}^{n-k}\,.
\end{align*}
The difference with \Cref{WCD} arises from the consideration of all three Pauli operators $\sigma_x^{(n)}$, $\sigma_y^{(n)}$ and $\sigma_z^{(n)}$ as Lindblad operator.
We consider the three-dimensional simple Lie group $SU(2)$ of associated generators $\sigma_x,\,\sigma_y,$ and $\sigma_z$ spanning the Lie algebra $\frak{su}(2)$, as well as the $n$-fold representation $SU(2)\ni g\mapsto U_g^{\otimes n}$, where $U$ denotes the defining spin $1/2$ representation of $SU(2)$: for any $\psi\in\CC^2$ and $ g\in SU(2)$,
\[U_g\,\psi=g\,\psi\,. \]
Just like previously, an easy use of \Cref{diffusivecollectivedeco} shows that the semigroup transferred from the heat semigroup on $SU(2)$ via the above tensor product representation coincides with the strong collective decoherence QMS (up to a rescaling of the Lindblad operators by a factor $\sqrt{2}$).
An easy application of \Cref{eq_tdeco_tmix} and the estimate of \Cref{matrixliegroups} for $n=3$ provides the following dimension-independent bound for the decoherence time of the strong collective decoherence QMS:
\begin{align}\label{eq_deco_strong1}
	t_{\operatorname{deco}}(\eps)\le \frac{34}{3}-8\ln\eps+2\ln\left( 1+\frac{3}{2}\ln\frac{3}{4}   \right)\,.
\end{align}

\subsection{Random SWAP gate}\label{sect_SWAP}

We consider two QMS that informally represent a random SWAP gate $F_{ij}$, applied at random times to two registers $i,j$ on $\cH^{\ten n}$, where $\cH$ is a $d$ dimensional Hilbert space. Both QMS are transferred from a classical semigroup on the permutation group on $n$ elements $\Sigma_n$. In both cases, the unitary representation on $\cH^{\otimes n}$ is the canonical one:
\[u_\omega:\,\,e_1\otimes\cdots\otimes e_n\mapsto e_{\omega(1)}\otimes\cdots \otimes e_{\omega(n)}\,,\qquad\omega\in\Sigma\,.\]
The first QMS is induced by so-called \emph{random transpositions} (RT). Here the   SWAP gate $F_{ij}$ is applied to any pair of registers $\cH\ten  \cH$ placed in the $(i,j)$'s registers:
\[ \L^{RT}(\rho) \lel \frac{1}{n}\sum_{ij} (\rho- F_{ij}\rho F_{ij}) \]
This QMS can be easily seen as being transferred from the classical semigroup with generator
\[L^{RT}(f)(\om)=\frac{1}{n}\sum_{ij} (f(\om)-f(\si_{ij}\om))\,\]
acting on $\ell_{\infty}(\Si_n)$, where $\sigma_{ij}$ is the transposition $(ij)$. For our second QMS we only allow  \emph{nearest neighbor} (NN) interaction on a cyclid $1D$ grid:
\[ \L^{NN}(\rho) \lel \sum_{1\le j\le n} (\rho- F_{j(j+1)}\rho F_{j(j+1)}) \pl , \]
where the $n+1$ registered is identified with the $1$. It can be seen that $\L^{NN}$ is being transferred from the generator on the permutation group
\[L^{NN}(f)(\omega)=\sum_{j=1}^n\,\left(f(\omega)-f(\sigma_{j(j+1)}\omega)\right)\,.\]
According to \cite{church}, the latter can be simulated with local gates. For the random transposition model, local means that only two registers are involved, whereas for the nearest neighbor interaction, local means neighboring gates. The different normalizations are chosen to fit with existing random transposition models in the literature, in particular \cite{Yau}.

\begin{theorem}\label{RandT} Following the notations above, we have
	\begin{align}
		t_{1,\infty}^{L^{RT}}(\eps)& \kl 4(1+\ln^2 n)-\ln \eps \,  \\
		t_{1,\infty}^{L^{NN}}(\eps) & \kl
		2n^2((1+\ln^2n)-\ln \eps) \pl ,
	\end{align}
	where the mixing time $t_{1,\infty}(\eps)$ is defined as (see also \Cref{sect3})
	\begin{align*}
&	t_{1,\infty}^L(\eps)\\
	&~~~=\inf\{t\geq0\,;\,\|e^{-tL}-\mu_G\,:\,L_1(\mu_G)\to L_\infty(\mu_G)\|\leq\eps\}\,.
	\end{align*}
	Consequently we obtain the following estimate on the decoherence time of the quantum SWAP evolutions:
	\begin{align}
		t_{\operatorname{deco}}^{\L^{RT}}(\eps)& \kl 4(1+\ln^2 n)-\ln \eps
		\\
		t_{\operatorname{deco}}^{\L^{NN}}(\eps) & \kl 4n^2((1+\ln^2n)-\ln \eps) \pl .
	\end{align}
\end{theorem}

\begin{proof}
	Remark that by ordering of the classical $L_p$ norm and by transference, it suffices to estimate $t_{1,\infty}(\eps)$ for random transposition models on the permutation group for $L^{TR}$ and $L^{NN}$. Let us start with the RT model.
	Indeed, according  to \cite{Yau} we know that
	\[ t_{1,\infty}^{L^{RT}}(\eps)\leq  c(\ln n-\ln \eps) \pl \]
	for some unknown constant $c$.
	Such an estimate is not available for nearest neighbor interaction.  Our starting point is the LSI-inequality
	\[ \LSI(L^{RT})\kl 2\ln(n) \]
	for $n\gl 2$. According to \cite{Diaconis1996} this implies
	\[ t_{1,2}(1/e)\kl (1+\frac{1}{4}\ln \ln n n!) 2\ln n \pl .\]
	Hence 
	\[ t_{1,\infty}(1/e^2)\kl 4\ln n(1+\frac{\ln n}{2})\kl 4(1+\ln n^2) \pl .\]
	Thus for arbitrary $\eps$, because of the spectral gap $1$, we find
	\[ t_{1,\infty}^{L^{RT}}(\eps)
	\kl 4(1+\ln^2n)-\ln \eps \pl .\]
	However, the  factor $(\ln n)^2$ is too large, because Diaconis-Shahshahani proved that as $n$ goes to $\infty$ $t_{1,2}(1/e)\sim \ln n$. Since the spectral gap is of order $1$, then the estimates for $t_{1,2}(\eps)$ requires an additional term $-\ln\eps$ as above. For the nearest neighbour model we first consider a graph, in our case the Caley graph of the permutation group, and compare the energy form
	\[ \mathcal{E}_E(f)  \lel
	\frac{1}{|\Si_n|} \sum_{\si \in \Si_n}
	\sum_{(ab)\in E} |f(\si_{ab}\si)-f(\si)|^2 \pl .\]
	Let $E'\subset \{1,...,n\}^2$ another generating set such that the graph is complete. For every $a,b\in E$, we can find a geodesic path $\gamma_{ab}:\{1,...,m\}\to E'$ and observe that
	\[ (f(\tau_{ab}\si)-f(\si))^2
	\kl m \sum_{j=1}^m  (f(\tau_{j}\tau_{a_jb}\si)-f(\tau_{a_jb}\si))^2 \pl .\]
	Here $\tau_j=\tau_j^{ab}$ comes from the generating set $E'$. In our case the longest possible path is $\le n$. This implies that
	\begin{align*}
	&	\mathcal{E}_E(f)\\
		&\le n \sum_{(c,d)\in E'}
		\sum_{ab\in E} \sum_j \delta_{\tau_{cd},\rho_{ab}^j}
		\frac{1}{|\Si_n|} \sum_{\si} (f(\rho_{ab}^j\tau_{a_j,b}\si)-f(\tau_{a_j,b}\si))^2 \\ 
		&= n \sum_{\si} \sum_{(c,d)\in E'}
		\sum_{ab\in E} \sum_j \delta_{\tau_{cd},\rho_{ab}^j}
		\frac{1}{|\Si|}
		\sum_{\si} (f(\tau_{cd}\si)-f(\si))^2 \\
		&= n \frac{1}{|\Si|}
		(\sum_{\si} \sum_{(c,d)\in E'} \sum_{(a,b)\in E, (c,d)\in \Gamma})
		(f(\tau_{cd}\si)-f(\si))^2 \pl .
	\end{align*}
	Since, we may take geodesic, no $(cd)$ is counted double. Thus for $E'=\{(j,j+1)|1\le j\le n\}$ and $E=\{(a,b)| a\neq b\}$, we deduce that
	\[   \mathcal{E}_E(f) \kl n^3 \mathcal{E}_{E'}(f) \pl .\]
	Thanks to the normalization factor $\frac{1}{n}$ for $L^{RT}$ is implies that
	\[ LSI(L^{NN})^{-1}\kl n^2 LSI(L^{RT})^{-1}\kl 2n^2 \ln n \pl .\]
	Note that our estimate also implies that the spectral gap $\la(L^{NN})\gl \frac{1}{2n^2}$. Thanks to \cite{Diaconis1996}, we deduce
	\begin{align*}
	t_{1,2}(1/e)^{L^{NN}} &\kl (1+\frac{1}{4}\ln\ln n!) 2 n^2 \ln n\\
&	\kl 2n^2 (1+\frac{\ln^2 n}{2}) \pl .
	\end{align*}
	By symmetry this implies
	\[ t_{1,\infty}^{L^{NN}}(1/e^2)
	\kl 4n^2 (1+\ln(n))^2 \pl .\]
	Using the spectral gap this yields
	\[  t_{1,\infty}^{L^{NN}}(\eps)\kl
	4n^2((1+\ln n)^2 -\ln \eps) \pl.\]
	Note that we have an automatic $cb$-norm estimate in this case, and hence transference, allows us to estimate  the decoherence time of the tensor swaps. \end{proof}

\begin{remark} The correct estimate $\operatorname{MLSI}(L^{RT})\sim 1$ for the modified (not complete) logarithmic Sobolev inequality has only recently been found \cite{gao2003exponential}.
	The standard inductive procedure appears not directly applicable for   $L^{NN}$, though. However, our proof shows that $\LSI^{-1}(L^{NN})\kl cn^2(1+\ln^2n)$, and hence trivially we find a bound $\operatorname{MLSI}$. We conjecture that for every  graph, we always have
	\[ \operatorname{CLSI}^{-1}(L_{E})\kl c \ln \ln |V| \operatorname{LSI}^{-1}(L_E) \pl .\]
	This would yield a bound of order $cn^2(1+\ln n)^3$ for the inverse of the CLSI constant. At any rate a better estimate of order  $cn^2$  would be highly desirable.
\end{remark}

\subsection{Compact Lie groups}\label{compactliegroups}
Here, we recall some well-known estimates for the heat semigroup defined on various compact Lie groups.

\paragraph{$1$-dimensional torus}
The following estimate was derived in example 1 of Section 3 of \cite{saloff1994precise}:
\begin{align*}
	\| h\mapsto k_t(gh^{-1})-1\|_2\le \sqrt{  2+\sqrt{\pi/2t}  }\,\e^{-t}.
\end{align*}

\paragraph{$n$-dimensional torus ($n>1$):} The logarithmic Sobolev constant associated to the heat semigroup on the n dimensional torus $\mathbb{T}^n$ is known to achieve the bound $1/c(S^{\text{Heat}})\le \lambda(S^{\text{Heat}})=1$. In Theorem 5.3 of \cite{saloff1994precise}, the following upper bound on its kernel (and in fact on the kernel of any uniformly elliptic generator) was found:
\begin{align}\label{eq:estimatetorus}
&	\sup_{g\in\mathbb{T}^n}	\|h\mapsto k_t(gh^{-1})-1\|_2\\
	&~~~~~~~\le \operatorname{exp}\left({-t+\frac{1}{2}\log\left(  \frac{1}{2}\,n\,\log n  \right) +6}\right).\nonumber
\end{align}

\paragraph{Matrix Lie groups:}
In \cite{saloff1994precise}, precise estimates on the kernel of diffusion semigroups on various Riemannian manifolds were obtained starting from a curvature dimension inequality $\operatorname{CD}(\rho,\nu)$. Applying this to the curvature dimension inequalities satisfied by semisimple Lie groups \cite{ROTHAUS1986358}, Saloff-Coste derived the following straightforward corollary (stated here as a theorem for sake of completeness):
\begin{theorem}\label{matrixliegroups}
	Let $(G,\mathfrak{g})$ be a real connected semi-simple compact Lie group of dimension $n$ endowed with the Riemannian metric induced by its Killing form. Then, the heat diffusion satisfies
	\begin{align}\label{CD}
	&	\sup_{g\in G}	\|h\mapsto k_t(gh^{-1})-1\|_2\\
		&~~~~\le \operatorname{exp}\left(   1+\lambda(\Delta)\left[ -t+\frac{16}{n}+2\log\left( 1+\frac{1}{2}\,n\log\frac{n}{4}  \right)   \right] \right),\nonumber
	\end{align}	
	where $\lambda(\Delta)$ is the spectral gap of $(G,\mathfrak{g})$. In particular, for $\eps>0$ and $t_n=2(1+\eps)\log n$, the above bound converges to $0$.
	Moreover, the following bounds the logarithmic Sobolev and Poincar\'{e} constants hold:
	\begin{align}\label{specalph}
		& \lambda(\Delta)\ge \frac{n}{8(n-1)}\geq \frac18\,, \\
		&c(\Delta)\leq \frac{4(n-1)}{n}\leq4\,.
	\end{align}	
	
\end{theorem}

\renewcommand{\lb}{\left(}
\renewcommand{\rb}{\right)}

\section{Capacity bounds, resource theories and entropic inequalities  }\label{capa}

Here we will discuss how to apply the estimates provided in \Cref{normtransference} to obtain upper bounds on
the (two-way) private, quantum, entanglement-assisted classical capacity and classical capacity of a transferred semigroup $(T_t)_{t\ge 0}$ converging to its associated conditional expectation $E_{fix}\equiv E_{N_{fix}}$.
Roughly speaking, the capacity is the ultimate rate of transmission of a certain resource through a quantum channel such that in the limit
of infinitely many uses of the channel, the success probability of the transmission  of this resource converges to $1$ after encoding and decoding operations.
We refer to~\cite[Chapter 8]{Watrous_2018} for a precise definition of the various capacities considered here and note that we will express all capacities in $e-$bits, as they are more convenient in this context.

Computing the exact value of the capacity is most often out of reach of current techniques. Instead, we are interested in upper bounding them. More particularly, we will mostly be interested in showing strong converse bounds on these capacities. A bound on a capacity is called a strong converse bound
if we have that if we exceed the transmission rate given by the bound, the probability of successful transmission
of a certain resource goes to $0$ exponentially fast
as the number of channel uses goes to infinity.
Our method relies on relating norm estimates to bounds on entropic quantities derived
from the sandwhiched R\'enyi entropies. Intuitively speaking, as the QMS $(T_t)_{t\ge 0}$ converges to $E_{fix}$, we expect that its capacity also converges to that of the conditional expectation, and wish to quantify this convergence.
For all of the results in this chapter, we will assume that we know the decomposition of the fixed point algebra $\Nfix\subset \cB(\cH):$
\begin{align}\label{fixedpointalgebra}
	\Nfix=\bigoplus\limits_{k=1}^m \mathbb{M}_{n_k}\otimes I_{\CC^{d_k}}\,.
\end{align}
It is easy to find such decompositions for transferred semigroups in terms of the decomposition into irreducible representations of the representation we use to define the semigroup. We recall our notation for states on $\Nfix$: we write $\sigma\in\Dcal(\Nfix)$ whenever $E^\dagger_{fix}(\sigma)=\sigma$ (recall also that $E^\dagger_{fix}=E_{fix}$ for unital, not necessarily self-adjoint QMS, as we assume for transferred QMS). Moreover, in~\cite{Fukuda_2007,Gao_2018tro} closed formulas are derived for the capacities of conditional expectations in terms of the decomposition of the underlying fixed point algebra. As we will see soon, roughly speaking, all these capacities will be at most the limiting capacity plus an additive error $\epsilon$ at time $$t(\epsilon):=\inf\left\{t\ge 0\,|\,\,\|k_t-1\|_{\infty}\le \epsilon   \right\}$$ for transferred semigroups. The main technical tool needed for this section is~\Cref{compar}. We will also show how to obtain capacity bounds from the modified logarithmic Sobolev inequality \cite{[KT13]} in \Cref{capmlsi}. This will in general be better than the one we obtain with control over $t(\epsilon)$, but it should be noted that finding an estimate on the logarithmic Sobolev constant of a QMS is a very hard problem in general (see \cite{beigi2018quantum,gao2018fisher,capel2018quantum}).
Another advantage of the $t(\eps)$ control is the fact that we can control all the relative entropies in the parameter range $p\in[1,\infty]$, which will be of crucial importance for our later applications. This is due to the fact that many strong converse bounds are only known in terms of R\' enyi entropies for $p>1$.\\

Before moving to the results, let us briefly explain our proof techniques: using the upper bound in \eqref{ineqfundam}, we have
\begin{align*}
&	D_p(T_{t(\epsilon)}(\rho\|\sigma)\\
	~~~~~&\leq D_p(E_{{fix}}(\rho)\|\sigma)+\frac{p}{p-1}\log(\|k_{t(\epsilon)}\|_{L_p(\mu_G)})\,.
\end{align*}
One way to bound the classical R\'{e}nyi divergence on the right-hand side of the above inequality is by invoking the following chain of inequalities: $\|k_{t(\epsilon)}\|_{L_p(\mu_G)}\leq\|k_{t(\epsilon)}\|_\infty^{\frac{p-1}{p}}$, by Riesz-Thorin theorem and the fact that $\|k_t(\eps)\|_{L_1(\mu_G)}=1$, so that
\begin{align}\label{boundeps}
	\|k_{t(\epsilon)}\|=\sup\limits_{g\in G}|k_{t(\epsilon)}(g)|\leq1+\sup\limits_{g\in G}|k_{t(\epsilon)}(g)-1|\leq1+\epsilon\,.
\end{align}
Although these estimates  based on $t(\epsilon)$ do not require the full power of transference, in the sense that it is also possible to derive the outer inequalities without resorting to interpolation as remarked after~\Cref{compar}, we still believe that these inequalities are of interest. This is due to the fact that it is usually nontrivial to derive them without noting that the underlying quantum channels arise as transferred  semigroups. Moreover, one can directly use an entropy decay estimate instead of the bound (\ref{boundeps}) from the \textit{modified logarithmic Sobolev inequality} for the underlying classical Markov kernel $(k_t)_{t\ge 0}$.

All of our bounds will be based on comparing the capacity of the semigroup at time $t$ to that of its limit as $t\to\infty$ and using transference techniques to estimate how close the channel is to its limit. Note, however, that our estimates will be tighter than working directly with continuity bounds for capacities. To the best of our knowledge, known continuity bounds for various capacities such as the ones of~\cite{Leung_2009}, allow to conclude that the capacity of two channels with output dimension $d$ differ by a factor of order $\epsilon\log(d)$ if they are $\epsilon$-close in diamond norm. However, in the case of transference we will be able to obtain that the capacities are $\epsilon$ close under similar assumptions.
\subsection{(Two way) Private and Quantum Capacity and Entanglement Breaking Times}
The private capacity quantifies the rate at which classical information that is secret to the environment can be reliably transmitted by a given
quantum channel in the asymptotic limit of many channel uses. For a quantum channel $T$, the private capacity is denoted  by $\mathcal{P}(T)$. Similarly,
the quantum capacity of a quantum channel quantifies at which rate quantum information can be reliably transmitted with a quantum channel, and is denoted by $\mathcal{Q}(T)$. Note that we always have $\mathcal{Q}(T)\leq\mathcal{P}(T)$ and thus, any upper bound
on the private capacity extends to a bound on the quantum capacity.
Moreover, we may also consider variations of these capacities in which we also allow for unlimited classical communication between
the sender and the receiver of the output of the quantum channel. These are usually called the two-way private and quantum capacities
and we will denote them by $\mathcal{P_\leftrightarrow}(T)$ and $\mathcal{Q_\leftrightarrow}(T)$, respectively. Clearly, we have
$\mathcal{P}(T)\leq\mathcal{P_\leftrightarrow}(T)$. We refer to e.g.~\cite{christandl2017relative} for a precise definition of these
quantities.

Using the techniques above, we can derive strong converses on the two-way quantum and private capacities based on the results of~\cite{christandl2017relative} and mixing time estimates.
More specifically, in~\cite{christandl2017relative} the authors show that for a quantum channel
$T: \mathcal{B}\left(\mathcal{H}\right)\to \mathcal{B}\left(\mathcal{H}\right)$ the quantity
\begin{align}\label{eq_upperbound_qC}
&	E_{\max}(T)\\
	&=\sup\limits_{\rho\in\mathcal{D}\left(\mathcal{H}\otimes\mathcal{H}\right)}\inf\limits_{\sigma\in\mathcal{D}\left(\mathcal{H}\otimes\mathcal{H}\right),\sigma\in\text{ SEP}}
	D_{\infty}\left(T\otimes\text{id}\left(\rho\right)||\sigma\right)\nonumber
\end{align}
is a strong converse upper bound on the two-way private and quantum capacities of $T$. Here $\operatorname{SEP}$ stands for the convex subset of $\Dcal(\cH\otimes\cH)$ of separable states, that is,
\begin{align*}
&\operatorname{SEP}\\
&=\{\sum_i p_i \rho^A_i\otimes\rho^B_i\,;\,p_i\geq0\,,\,\sum_i\,p_i=1\,,\,\rho^A_i,\rho_i^B\in\Dcal(\cH)\}\,.
\end{align*}
We then have:
\begin{theorem}[Bounding the two-way quantum and private capacity]\label{thm:boundingquantum}
	Let  $(T_t)_{t\ge0}$ be a transferred QMS.
	Then:
	\begin{align*}
	&	\max_k\log(n_k)\leq\\
		&~~~~~~~~~~\mathcal{P_\leftrightarrow}(T_{t(\epsilon)}),\mathcal{Q_\leftrightarrow}(T_{t(\epsilon)})\leq\max_k\log(n_k)+\epsilon
	\end{align*}
	Moreover, the upper bound is a strong converse bound.
\end{theorem}

\begin{proof}
	We have the following chain of inequalities:
	\begin{align*}
		&E_{\max}(T_t)\\
		&=\sup\limits_{\rho\in\mathcal{D}\left(\mathcal{H}\otimes\mathcal{H}\right)}\inf\limits_{\sigma\in\mathcal{D}\left(\mathcal{H}\otimes\mathcal{H}\right),\sigma\in\text{ SEP}}
		D_{\infty}\left(T_t\otimes\text{id}\left(\rho\right)\|\sigma\right)\\
		&\leq
		\sup\limits_{\rho\in\mathcal{D}\left(\mathcal{H}\otimes\mathcal{H}\right)}\inf\limits_{\substack{\sigma\in \cD(N_{fix})\otimes\mathcal{D}\left(\mathcal{H}\right)\\
				\sigma\in\text{ SEP}}}
		D_{\infty}\left(T_t\otimes\text{id}\left(\rho\right)\|\sigma\right) \\
		&\leq\sup\limits_{\rho\in\mathcal{D}\left(\mathcal{H}\otimes\mathcal{H}\right)}\inf\limits_{\substack{\sigma\in \cD(N_{fix})\otimes\mathcal{D}\left(\mathcal{H}\right)\\
				\sigma\in\text{ SEP}}}D_{\infty}\left(E_{{fix}}\otimes\text{id}\left(\rho\right)\|\sigma\right)+\epsilon\\
		&=\max_k\log(n_k)+\epsilon\,,
	\end{align*}
	where in the first inequality we used the fact that restricting the infimum over separable states such that one half lies in the fixed point algebra can only increase the quantity. In the second inequality we applied~\Cref{compar} to the transferred semigroup corresponding to the representation $U
	_g\otimes I_\cH$. Finally, it remains to show the last equality. The first term on the r.h.s, $\max_k\log(n_k)$, corresponds to the capacity of the conditional expectation, as computed in~\cite{Gao_2018tro}. It remains to show that, for conditional expectations, the infimum in \Cref{eq_upperbound_qC} is attained at points with one half of the state in the fixed point algebra of this conditional expectation. To see that this is indeed the case, note that for any $\rho\in\mathcal{D}\left(\mathcal{H}\otimes\mathcal{H}\right)$ we have
	\begin{align*}
	&	\inf\limits_{\sigma\in \mathcal{D}\left(\mathcal{H}\otimes\mathcal{H}\right),\sigma\in\text{ SEP}}D_{\infty}\left(E_{{fix}}\otimes\text{id}\left(\rho\right)\|\sigma\right)\\
		&~~~\geq
		\inf\limits_{\sigma\in \mathcal{D}\left(\mathcal{H}\otimes\mathcal{H}\right),\sigma\in\text{ SEP}}D_{\infty}\left(E_{{fix}}\otimes\text{id}\left(\rho\right)\|E_{{fix}}\otimes\text{id}\left(\sigma\right)\right)\\
		&~~~=\inf\limits_{\sigma\in \cD(N_{fix})\otimes\mathcal{D}\left(\mathcal{H}\right),\sigma\in\text{ SEP}}D_{\infty}\left(E_{{fix}}\otimes\text{id}\left(\rho\right)\|\sigma\right)
	\end{align*}
	by the data processing inequality and the fact that conditional expectations are projections.
	The lower bound on the capacities follows from the lower bound in Theorem~\ref{compar} and the expression for the quantum and private capacities of conditional expectations.
\end{proof}

In a similar fashion, one can use that the relative entropy of entanglement~\cite{Pirandola2017} of a channel $T$
$$E_R(T)=\sup_{\rho\in\cD(\cH\otimes \cH)}\,\inf_{\sigma\in\operatorname{SEP}}\,D(T\otimes \id(\rho)\|\sigma)$$
is a strong converse bound on the private and quantum capacities of $T$ \cite{wilde2017converse} in order to derive the following
\begin{theorem}[Bounding the one-way quantum and private capacity]\label{thm:boundingquantumone}
	Let  $(T_t)_{t\ge0}$ be a transferred QMS.
	Then:
	\begin{align*}
		\max_k\log(n_k)\leq\mathcal{P}(T_{t(\epsilon)}),\mathcal{Q}(T_{t(\epsilon)})\leq\max_k\log(n_k)+\epsilon
	\end{align*}
	Moreover, the upper bound is a strong converse bound.
\end{theorem}

Although the last theorems provide bounds for all times, in case of primitive QMS we expect that the
semigroup becomes entanglement
breaking at some point and, thus, all the aforementioned capacities become $0$. We recall that an entanglement breaking channel is one whose action on one part of a bipartite entangled state always yields a separable state.
Using our methods we can also estimate these times (see also \cite{francca}). To this end, we define
\begin{definition}[Entanglement breaking time]
	Let $(T_t)_{t\geq0}$ be a primitive QMS. We define its entanglement breaking time, $t_{\operatorname{EB}}$, to be given by
	\begin{align*}
		t_{\operatorname{EB}}((T_t)_{t\geq0})=\inf\{t\geq0\,|\,T_t~\textrm{is entanglement breaking}\}\,.
	\end{align*}
	
\end{definition}

\begin{theorem}
	Let $(T_t)_{t\geq0}$ be a primitive transferred semigroup on $\mathcal{B}\lb \mathbb{C}^d\rb$. Then
	\begin{align*}
		t_{\operatorname{EB}}((T_t)_{t\ge 0})\leq t(d^{-1})\,.
	\end{align*}
\end{theorem}

\begin{proof}
	In~\cite{gurvits2002}, the authors show that all states in the ball with radius $\frac{1}{d}$
	in the Hilbert Schmidt norm around the bipartite
	maximally mixed state are separable. Moreover, it is well-known that a quantum channel is entanglement breaking if and only if its Choi matrix is separable~\cite{ebchannels}.
	It follows from the transference principle that
	\begin{align}\label{equ:transfnorm1to1}
		\|T_{t(\epsilon)}\otimes\id-E_{{fix}}\otimes \id\|_{1\to 1}\leq \epsilon
	\end{align}
	Note that, as we have an erdogic semigroup, $E_{{fix}}(X)=\frac{1}{d}\tr(X)\,{I}_{\CC^d}$ and, thus, $E_{{fix}}\otimes \id$ will yield the maximally mixed state when applied to the maximally entangled state.
	Choosing the maximally entangled state as our input to the channel, it follows from~\eqref{equ:transfnorm1to1}
	that the Choi matrix of $T_t$ is separable for $t(d^{-1})$ and, therefore, the map becomes entanglement breaking for this time.
\end{proof}

It is then straightforward to adapt the various convergence results we have to obtain estimates on the time the transferred QMS becomes entanglement
breaking. Note that the situation is a bit more subtle if the semigroup is not assumed to be primitive. Consider the example
of $(T_t^{\textrm{deph}})_{t\ge 0}$ as defined in \Cref{eq_depo}. We can see that in this case the semigroup is not entanglement breaking for any finite $t\geq0$ but
is entanglement breaking in the limit $t\to\infty$. This shows that the entanglement breaking time might be infinite if we drop the assumption of
primitivity. As a matter of fact, in~\cite{francca}, the authors show that this is always the case for non-primitive semigroups.

\subsection{Classical capacity}
The classical capacity of a quantum channel is the highest rate at
which classical information can be transmitted through a quantum channel with vanishing error probability.
We will denote the classical capacity of a quantum channel by $\mathcal{C}(T)$.
One can show that~\cite{strongconvrenyi}:
\begin{align}\label{equ:inforadiusclassical}
	\mathcal{C}(T)\leq \lim\limits_{n\to\infty}\inf\limits_{\sigma\in\mathcal{D}\lb \mathcal{H}_A^{\otimes n}\rb}\sup\limits_{\rho\in\mathcal{D}\lb \mathcal{H}_A^{\otimes n}\rb}\frac{1}{n}D_p\lb T^{\otimes n}\lb\rho\rb||\sigma\rb
\end{align}
for any $p>1$.
Moreover, this is a strong converse bound. Taking $p=1$ in the bound  above also gives an upper bound on the classical capacity~\cite{Schumacher_2001}, albeit not a strong converse one. Bounding the classical capacity is a notoriously difficult problem because of the nonadditivity of the output entropy. However, for transferred groups we have:

\begin{theorem}\label{theoclasscapatransf}
	Let $(T_t)_{t\ge 0}$ be a transferred QMS. Then:
	\begin{align*}
		\log\lb \sum\limits_{i=1}^mn_i\rb\leq \mathcal{C}(T_{t(\epsilon)})\leq \log\lb \sum\limits_{i=1}^mn_i\rb+\epsilon\,.
	\end{align*}
	Moreover, the upper bound is a strong converse bound.
\end{theorem}
\begin{proof}
	First, note that the semigroup $T_{t}^{\otimes n}$ corresponds to the channel we obtain by transfering the Markov kernel
	\begin{align*}
		k_{n,t}=\bigotimes\limits_{i=1}^nk_t
	\end{align*}
	on $G^n$. Moreover, the conditional expectation related to the semigroup is clearly $E_{{fix}}^{\otimes n}$.
	We may obtain an upper bound to~\Cref{equ:inforadiusclassical} by restricting the infimum to states that are on $N_{fix}^{\otimes n}$.
	Let $\sigma\in \cD(N_{fix}^{\otimes n})$.
	By \Cref{compar}, for any $\rho\in\mathcal{D}\lb \mathcal{H}_A^{\otimes n}\rb$,
	\begin{align*}
		D_p\lb T_{t(\epsilon)}^{\otimes n}\lb\rho\rb||\sigma\rb\leq
		D_p(E_{{fix}}^{\otimes n}\lb\rho\rb||\sigma)+n\epsilon\,.
	\end{align*}
	where the last inequality follows from \ref{compar}, the fact that $k_{n,t(\epsilon)}$ is a product and the elementary inequality $\log(1+x)\leq x$.
	\begin{align*}
		&\sup\limits_{\rho\in\mathcal{D}\lb \mathcal{H}_A^{\otimes n}\rb}D_p\lb T_{t(\epsilon)}^{\otimes n}\lb\rho\rb\|\sigma\rb\\
		&~~~~~~~~~~~~~~~~~~~~\leq
		\sup\limits_{\rho\in\mathcal{D}\lb \mathcal{H}_A^{\otimes n}\rb}D_p(E_{{fix}}^{\otimes n}\lb\rho\rb\|\sigma)+n\epsilon\,.
	\end{align*}
	Moreover, we have
	\begin{align*}
	&	\inf\limits_{\sigma\in \cD(N_{fix}^{\otimes n})}\sup\limits_{\rho\in\mathcal{D}\lb \mathcal{H}_A^{\otimes n}\rb}D_p(E_{{fix}}^{\otimes n}\lb\rho\rb\|\sigma)\\
		&~~~~~~~~~~~~~~~~= \inf\limits_{\sigma\in\mathcal{D}\lb \mathcal{H}_A^{\otimes n}\rb}\sup\limits_{\rho\in\mathcal{D}\lb \mathcal{H}_A^{\otimes n}\rb}D_p(E_{{fix}}^{\otimes n}\lb\rho\rb\|\sigma)\,,
	\end{align*}
	which follows from an application of the data processing inequality and the fact that $E_{{fix}}$ is a projection, as in the proof of \Cref{thm:boundingquantum}.
	The upper bound then follows from taking the infimum over all $\sigma\in \cD(N_{fix}^{\otimes n})$, dividing the expression by $n$, taking the limit $n\to\infty$ and the expression for the classical capacity of a conditional expectation given in~\cite{Gao_2018tro}.  The lower bound follows by an argument similar to that given before for the other capacities.
\end{proof}

\subsection{Entanglement-assisted classical capacity}
The entanglement-assisted classical capacity of a quantum channel is the highest rate at
which classical information can be transmitted through a quantum channel with vanishing error probability
given that the sender and receiver share and potentially consume an unlimited amount of entanglement. We refer to~\cite[Section 8.1.3]{Watrous_2018} for
a precise definition.
We will denote the entanglement-assisted capacity
of a quantum channel $T$ by $\mathcal{C}_{\operatorname{EA}}(T)$ and note that it is an upper bound on the classical capacity of a quantum channel.
In~\cite{gupta2015multiplicativity}, the authors show that the following quantity is an upper bound on the entanglement assisted classical capacity (EAC) of a quantum channel in the strong converse sense\footnote{note that we are just rewriting the mutual information in terms of the relative entropy in~\eqref{equ:radiusentassited}.}:
\begin{align}\label{equ:radiusentassited}
&	\chi_{\operatorname{EA}}(T)	=\inf\limits_{\sigma_A\in\mathcal{D}\lb\mathcal{H}_A\rb}\sup\limits_{\outerp{\psi}{\psi}\in\mathcal{D}(\mathcal{H_A}\otimes\mathcal{H_B})}\\
&~~~~~~~~~~~~~~~~~~~~~~~~~~~~~D\lb T\otimes \id(\outerp{\psi}{\psi})\|\sigma_A\otimes\rho_B\rb\,,\nonumber
\end{align}
where $\rho_B$ is the reduced density matrix of $T\otimes \id(\outerp{\psi}{\psi})$ on system $B$ and where $\outerp{\psi}{\psi}$ stands for the rank-one orthogonal projection on the norm-one vector $\psi\in\Hcal$. A closed formula for the entanglement assisted capacity was obtained in~\cite{Gao_2018tro}. There, they show that for a conditional expectation $E_{{fix}}$ we have
\begin{align*}
	\mathcal{C}_{\operatorname{EA}}(E_{{fix}})= \log\lb\sum\limits_{i=1}^mn^2_i\rb    \,.
\end{align*}
Using similar ideas as before we can estimate the entanglement assisted classical capacity using relative entropy transference.

\begin{theorem}[Bounding the entanglement assisted classical capacity]\label{EAcapbound}
	Let $(T_t)_{t\ge 0}$ be a transferred QMS.
	Then:
	\begin{align*}
		\log\lb\sum\limits_{i=1}^mn^2_i\rb    \leq \mathcal{C}_{\operatorname{EA}}(T_{t(\epsilon)})\leq \log\lb\sum\limits_{i=1}^mn^2_i\rb+\epsilon\,.
	\end{align*}
	Furthermore, if $(S_t)_{t\geq 0}$ satisfies an $\alpha$-MLSI and $G$ is finite, then:
	\begin{align*}
		\mathcal{C}_{\operatorname{EA}}(T_{t})\leq \log\lb\sum\limits_{i=1}^mn^2_i\rb+e^{-2\alpha t}\log(|G|)
	\end{align*}
	Moreover, the upper bounds are in the strong converse sense.
\end{theorem}

\begin{proof}
	The proof is similar to the ones of \Cref{thm:boundingquantum,theoclasscapatransf} and \Cref{theoclasscapatransflogo}, and hence is omitted.
\end{proof}

\subsection{Capacities from a modified logarithmic Sobolev inequality}\label{capmlsi}

Similarly to the derivation of decoherence times, one expects to get tighter bounds on the various capacities by directly applying a quantum functional inequality, when the latter is known. As was the case with decoherence times, the decay of the capacities we obtain does not depend on particular properties of the representation at hand and will in general not be tight.
For capacities, the right functional inequality to consider is the modified logarithmic Sobolev inequality (MLSI). To the best of our knowledge, this connection between a MLSI and capacity bounds cannot be found in the literature beyond primitive semigroups~\cite{M_ller_Hermes_2016}, so we establish it here for more general semigroups.
Here, we still assume that $(T_t)_{t\ge 0}$ is a quantum Markov semigroup on $\cB(\cH)$ that is symmetric with respect to the Hilbert Schmidt inner product. Then, instead of using the entropy comparison theorem (\Cref{compar}), one can simply decompose the relative entropy between $T_t(\rho)$, $\rho\in\cD(\cH)$, and any state $\sigma\in \cD(N_{fix})$ as follows:
\begin{lemma}\label{entropydecomposition}
	For any $\rho\in\cD(\cH)$, and any $\sigma\in \cD(N_{fix})$,
	$$D(T_t(\rho)\|\sigma)=D(E_{fix}(\rho)\|\sigma)+D(T_t(\rho)\|E_{fix}(\rho))\,.$$
\end{lemma}
\begin{proof}
	\begin{align*}
	&	D(T_t(\rho)\|\sigma)\\
		&=\tr(T_t(\rho)\,(\ln(T_t(\rho))-\ln\sigma))\\
		&=\tr(T_t(\rho)\,(\ln(T_t(\rho))-\ln(E_{fix}(\rho))))\\
		&~~~~~~~~~~~~~~~~~~~~~~~~~~~~~+\tr(T_t(\rho)(\ln(E_{fix}(\rho))-\ln\sigma))\\
		&=D(T_t(\rho)\|E_{fix}(\rho))+\tr(E_{fix}(\rho)(\ln (E_{fix}(\rho))-\ln\sigma))\\
		&=D(T_t(\rho)\|E_{fix}(\rho))+D(E_{fix}(\rho)\|\sigma)\,,
	\end{align*}	
	where in the third line we used that $E_{fix}$ is a conditional expectation with respect to the completely mixed state, so that $E_{fix}=E_{fix}^\dagger$ and for any $\sigma\in \cD(N_{fix})\cap \cD(\cH)_+$, $\ln(\sigma)=E_{fix}(\ln(\sigma))$.
\end{proof}
We recall that, given a faithful quantum Markov semigroup $(T_t=\e^{-t\cL})_{t\ge 0}$, its decoherence-free modified logarithmic Sobolev constant $\alpha_{1}(\cL)$ has been defined in~\cite{BarEID17} as follows:
\begin{align*}
	\alpha_1(\cL):=\inf_{\rho\in\cD_+(\cH)}\frac{\tr(\cL(\rho)(\ln\rho-\ln E_{fix}(\rho)))}{D(\rho\|E_{fix}(\rho))}\,
\end{align*}
(we recall our convention that $\Lcal$ is a positive semi-definite operator). It is the largest constant $\alpha>0$ such that the following decay in relative entropy occurs\footnote{the theory of functional inequalities for the exponential decay of R\' enyi entropies is not well-established beyond the primitive case~\cite{MF16,2018arXiv181000906C} and for $p=+\infty$. We leave this to future work.}:
$$D(T_t(\rho)\|E_{fix}(\rho))\le \e^{-\alpha\,t}D(\rho\|E_{fix}(\rho))\,.$$
The classical logarithmic Sobolev inequality is then defined in an analogous way.
An extension of \Cref{compar} then gives the
\begin{theorem}\label{comparisonentropMLSI}
	Let $(T_t)_{t\ge 0}$ be a transferred QMS from a finite group $G$ on $\Bcal(\Hcal)$, $\sigma\in \Dcal(N_{fix})$ and suppose that $S_t$ has an $\operatorname{MLS}$ constant $\alpha_1(L)>0$.	Then, for any state $\rho\in \Dcal(\Hcal)$ and $t\geq0$:
	\begin{align}
		D(E_{{fix}}(\rho)\|\sigma) 
	&	\leq D(T_{t}(\rho)\|\sigma) \label{eq_comparisonentrop1MLSI}\\
		&\leq D(E_{{fix}}(\rho)\|\sigma)+e^{-\alpha_1(L) t}\log\lb |G|\rb\nonumber
	\end{align}
	
\end{theorem}
\begin{proof}
	It follows from~\Cref{compar} that
	\begin{align*}
		D(T_{t}(\rho)\|\sigma) \leq D(E_{{fix}}(\rho)\|\sigma)+\int_{G}k_t\log(k_t)d\mu_G.
	\end{align*}
	Now note that 
	\begin{align*}
		\int_{G}k_t\log(k_t)d\mu_G=D(\mu_t\|\mu_G),
	\end{align*}
	where $g$ is an arbitrary element of the group and where $d\mu_t=k_t*\delta_g\,d\mu_G$. Remark also that $\mu_t=S_t(\delta_g)$. Thus, as we assumed that $S_t$ satisfies a MLSI, we conclude that
	\begin{align}\label{equ:boundrelativenetropygroup}
		D(\mu_t\|\mu_G)\leq e^{-\alpha_1(L) t}\log(|G|),
	\end{align}
	where in the last step we used the fact that $D(\mu_t\|\mu_G)=\log(|G|)-S(\mu_t)\leq \log(|G|)$ for any $\mu_t$. This yields the claim.
\end{proof}
We remark that the proof also generalizes to compact Lie groups, provided a bound on $D(\mu_{t_0}\|\mu_G)$ for some $t_0>0$ to obtain an analogue of~\eqref{equ:boundrelativenetropygroup} and then apply the relative entropy decay estimate. We are not aware of techniques for bounding this entropy in the Lie group case except once again invoking inequalities like~\eqref{equ:boundhormander1toinfty} and then bounding the relative entropy by the max-relative entropy. Thus, it would be interesting to obtain more fine-tuned bounds on $D(\mu_{t_0}\|\mu_G)$.\\

With this tool at hand, the following result can be proved in a very similar fashion as \Cref{thm:boundingquantumone,thm:boundingquantum,theoclasscapatransf,EAcapbound}:
\begin{theorem}\label{theocapmlsi}
	Let $(T_t=\e^{-t\cL})_{t\ge 0}$ be a quantum Markov semigroup on $\cB(\CC^d)$ that is symmetric with respect to the Hilbert Schmidt inner product. Then, for any $t\ge 0$:
	\begin{align*}
		&	\mathcal{Q}(T_t),\mathcal{P}(T_t)\le \max_k\log n_k+2\e^{-\alpha_1(\cL\otimes \id)t}\log(d)\\
		&\mathcal{C}_{\operatorname{EA}}(T_t)\le \log\left(\sum_kn_k^2\right)+2\e^{-\alpha_1(\cL\otimes \id)t}\log(d)\,.
	\end{align*}	
	
\end{theorem}	

\begin{proof}
	The proof proceeds completely analogously to the one of \Cref{thm:boundingquantum}. The only difference lies in the use of \Cref{entropydecomposition} instead of \Cref{compar} to bound the relative entropy.
	In this case, we will obtain a remaining term of the form $D(\rho\|E_{fix}\otimes\id(\rho))$. But another application of Lemma~\ref{entropydecomposition} shows that this term is bounded by $2\log(d)$.	
\end{proof}	
Note that we cannot apply an estimate on $\alpha_1\lb \cL\rb$ directly to obtain capacity bounds for the classical capacity due to the need of regularization. In order to do so, we need to show that the following complete version of the modified logarithmic Sobolev inequality holds (see \cite{BarEID17,gao2018fisher}): define the \textit{complete modified logarithmic Sobolev} (cMLS) \textit{constant}
\begin{align*}
&	\alpha_c(\cL):=\inf_{k\in\mathbb{N}}\inf_{\rho\in\cD_+(\cH\otimes\mathbb{C}^k)}\\
	&\frac{\tr((\cL\otimes\id_{\cB(\CC^k)})(\rho)(\ln\rho-\ln((E_{fix}\otimes\id_{\cB(\CC^k)})(\rho))))}{D(\rho\|(E_{fix}\otimes\id_{\cB(\CC^k)})(\rho))}\,.
\end{align*}
The cMLS constant $\alpha_c$ is known to tensorize \cite{BarEID17,gao2018fisher}: for any two QMS $(T_t=\e^{-\cL t})_{t\ge 0}$ and $(Q_t=\e^{-\cK t})_{t\ge 0}$
\begin{align*}
	\alpha_c(\cL\otimes\id+\id\otimes \cK)\ge \min\{\alpha_c(\cL),\,\alpha_c(\cK)\}\,.
\end{align*}	
By a simple look at the proof of the tensorization of $\alpha_1(\cL)$ for the generalized depolarizing semigroup \cite{beigi2018quantum,capel2018quantum} one can derive the positivity of its cMLS constant. This can be readily extended to the case of a simple semigroup of generator of the form $\cL=\id-E_{fix}$:
\begin{lemma}
	For any subalgebra $N$ of $\cB(\CC^d)$ with conditional expectation $E_N$ associated to the completely mixed state, the simple semigroup $(T^{N}_t)_{t\ge 0}$ of generator $\cL^N=\id-E_{N}$ of fix point algebra $N$ satisfies $$\alpha_c(\cL^N)\ge 1\,.$$
\end{lemma}
\begin{proof}
	For sake of clarity, denote by $\id_k$ the identity map on $\cB(\CC^k)$. Then, for any $\rho\in\cD_+(\cH\otimes\CC^k)$,
	\begin{align*}
	&	\tr((\id-E_N)\otimes\id_{k}(\rho)(\ln\rho-\ln (E_{N}\otimes\id_k)(\rho)))
		\\
		&~~=D(\rho\|(E_N\otimes\id_k)(\rho))+D((E_N\otimes\id_k)(\rho)\|\rho)\\
		&~~\ge D(\rho\|(E_N\otimes \id_k)(\rho))\,.
	\end{align*}	
\end{proof}

The link to the classical capacity of the QMS is made in the following theorem.

\begin{theorem}\label{class}
	Let $(T_t=\e^{-t\cL})_{t\ge 0}$ be a symmetric quantum Markov semigroup on $\cB(\CC^d)$ with positive $\operatorname{cMLS}$ constant $\alpha_c(\cL)$. Then, for any $t\ge 0$:
	$$\mathcal{C}(T_{t})\le \e^{-\alpha_c(\cL)t}\log d+\log\left(  \sum_{i=1}^m\,n_i  \right)\,.$$
\end{theorem}	

\begin{proof}
	We use \Cref{equ:inforadiusclassical} for $p=1$, which is equal to the classical capacity of a quantum channel (see \cite{strongconvrenyi}):
	\begin{align*}
	&	\mathcal{C}(T_t)\\
	&~~~\le\lim_{n\to\infty}\inf_{\sigma\in\cD(\cH_A^{\otimes n})} \sup_{\rho\in\cD(\cH_A^{\otimes n})}\,\,\frac{1}{n}\,D(T_t^{\otimes n}(\rho)\|\sigma)\\
		&~~~\le \lim_{n\to\infty} \inf_{\sigma\in\cD(N_{fix}(\cL^{(n)}))}\sup_{\rho\in\cD(\cH_A^{\otimes n})}\,\frac{1}{n}D(T_t^{\otimes n}(\rho)\|\sigma)\\
		&~~~\le \e^{-\alpha_c(\cL)\,t }  \lim_{n\to\infty}\frac{1}{n}\sup_{\rho\in\cD(\cH_A^{\otimes n})}D(\rho\|E_{fix}^{\otimes n}(\rho))   \\
		&~~~+\lim_{n\to\infty }\frac{1}{n}\inf_{\sigma\in \cD(N_{fix}(\cL^{(n)}))}\sup_{\rho\in\cD(\cH_A^{\otimes n})}D(E_{fix}^{\otimes n}(\rho)\|   \sigma)\,,
	\end{align*}	
	where we used \Cref{entropydecomposition} as well as the definition of the cMLS constant in the last line. The rest of the proof follows the same lines as for the one of \Cref{theoclasscapatransf} and the one of \Cref{theocapmlsi}.
	
\end{proof}

\begin{remark}
	Note that, although we expect that this method will yield tighter bounds for a given semigroup, it does not yield strong converse bounds for the capacities, neither does it provide bounds on the private two-way capacities.
	In order to get a bound on the two-ways private capacities, or a strong converse bound on the classical capacity of symmetric channels, we would need to extend the theory of complete logarithmic Sobolev inequalities to sandwiched R\'{e}nyi divergences. This falls outside the scope of this paper and will be done elsewhere. One exception is the two-way quantum capacity, as the results of~\cite{Berta_2018} show that the relative entropy of entanglement gives a strong converse for this capacity.
	
\end{remark}	

As mentioned at the beginning of the section, another way to upper bound capacities is to combine the upper bounds found in \Cref{compar} with a modified logarithmic Sobolev constant for the classical semigroup of kernel $(k_t)_{t\ge 0}$. This method first has the advantage of generally providing better bounds than the estimates based on $t(\eps)$. Moreover, it is easier to use in practice as compared to the quantum MLSI method, due to the relative lack of maturity of the latter field.
\begin{proposition}\label{theoclasscapatransflogo}
	Let $(T_t=\e^{-t\L})_{t\ge 0}$ be a transferred QMS from a finite group $G$ such that $(S_t)_{t\geq 0}$ with $\operatorname{MLS}$ constant $\alpha_1(L)>0$. Then:
	\begin{align*}
		\log\lb \sum\limits_{i=1}^mn_i\rb\leq \mathcal{C}(T_{t})\leq \log\lb \sum\limits_{i=1}^mn_i\rb+e^{-\alpha_1(L) t}\log(|G|)\,.
	\end{align*}
\end{proposition}
\begin{proof}
	The proof is essentially the same as the one of~\Cref{theoclasscapatransf}, but instead of bounding the relative entropy of $k_t$ using $t(\epsilon)$, we apply \Cref{comparisonentropMLSI} instead. Moreover, note that the MLSI of classical semigroups tensorizes~\cite{Bobkov_2003}.
\end{proof}
\begin{remark}
	The method of \Cref{theoclasscapatransflogo} can be extended to other capacities in a similar way as what we did before from the existence of a MLSI constant directly for the quantum Markov semigroup $(T_t)_{t\ge 0}$. Since the method is identical to those for the case of the classical capacity, we do not pursue this path here.
\end{remark}

\subsection{Resource theories and entropic inequalities}
In a related setting, it is also possible to apply our techniques to obtain estimates for different \emph{relative entropies of a resource}~\cite{Brand_o_2015}. In the framework of resource theories, one is usually given a sequence of sets of free states $\mathcal{F}_n\subset\mathcal{D}\lb \mathcal{H}^{\otimes n}\rb$, which are supposed to model those states that do not provide any resources.
We will usually denote $\mathcal{F}_1$ by $\mathcal{F}$.
One is also given a set of free operations, which are those quantum channels that cannot convert free states into resource states. The relative entropy of a resource $D_\mathcal{F}$ is then defined as
\begin{align*}
	D_\mathcal{F}\lb \rho\rb=\inf\limits_{\sigma \in \mathcal{F}}D\lb \rho\|\sigma\rb.
\end{align*}
One can then show that the regularized version of $D_\mathcal{F}\lb \rho\rb$, i.e. $\lim n^{-1}D_{\mathcal{F}_n}   \lb \rho^{\otimes n}\rb$, quantifies the optimal conversion rate of one resource to another whenever the set of free states and operations satisfies some natural properties.
We refer to~[Theorem 1]\cite{Brand_o_2015} for more details.

The connection to our techniques arises from the fact that for many prominent resource theories, like the resource theory of coherence~\cite{Streltsov_2017} or the resource theory of asymmetry~\cite{Ahmadi_2013}, the set of free states has the additional property of being naturally related to an algebra $N$. For instance, in the case of coherence, the free states are described as those states that are diagonal in a given basis and, thus, are naturally contained in the algebra of diagonal operators in a certain basis. More precisely, whenever $E_{fix}(\mathcal{D}\lb \mathcal{H}\rb)\subset \mathcal{F}$ we can use our techniques to bound the relative entropy of a resource and its regularized version for outputs of a given semigroup:
\begin{proposition}\label{theoresource}
	Let $\mathcal{F}$ be a set of free states for a resource theory and $(T_t)_{t\geq 0}$ be a transferred QMS from a finite or Lie group $G$ on $\mathcal{B}(\mathcal{H})$ such that $E_{fix}^{\otimes n}(\mathcal{D}\lb \mathcal{H^{\otimes n}}\rb)\subset \mathcal{F}_n$ for all $n$, where $\mathcal{F}_n$ denotes the set of free states corresponding to $\cH^{\otimes n}$. Then for any state $\rho\in\mathcal{D}\lb \mathcal{H^{\otimes n}}\rb$ we have
	\begin{align*}
		\lim\limits_{n\to\infty}\frac{1}{n}D_{\mathcal{F}_n}\lb \lb T_t(\rho)\rb^{\otimes n}\rb\leq \int_G k_t\log(k_t)d\mu_G.
	\end{align*}
\end{proposition}
\begin{proof}
	Note that it follows from~\Cref{compar} that for any $n$ and $\sigma\in E_{fix}^{\otimes n}(\mathcal{D}\lb \mathcal{H}^{\otimes  n}\rb)\subset \mathcal{F}_n$:
	\begin{align*}
	&	D((T_t(\rho))^{\otimes n}\|\sigma)\\
		&~~~~\leq D(E_{fix}^{\otimes n}(\rho)^{\otimes n}\|\sigma)+\int_G k_t^{\otimes n}\log(k_t^{\otimes n})d\mu_G.
	\end{align*}
	Thus, by picking $\sigma=E_{fix}^{\otimes n}(\rho^{\otimes n})\in \mathcal{F}_n$ and by the additivity of the relative entropy we conclude that
	\begin{align*}
		\inf\limits_{\sigma \in \mathcal{F}_n}D(T_t(\rho)^{\otimes n}\|\sigma)\leq n\int_G k_t\log(k_t)d\mu_G,
	\end{align*}
	from which the claim follows.
\end{proof}
Thus, we see that, as in the case of the capacities, it is possible to upper bound the regularized maximal relative entropy of any output of a transferred semigroup by computing or estimating the purely classical relative entropy between the kernel and the Haar measure on the group.

The upper bounds obtained in \Cref{compar} can also be similarly used to upper bound the amortized channel relative entropy \cite{berta2018amortized} defined for any two channels $\mathcal{E},\mathcal{F}:\cD(\cH_A)\to \cD(\cH_B)$ by
\begin{align*}
	&D^A(\mathcal{E}\|\mathcal{F})\\
	&~=\sup_{\phi_{RA},\psi_{RA}}D(( \operatorname{id}_R\otimes \mathcal{E})(\phi_{RA})\|( \operatorname{id}_R\otimes \mathcal{F})(\psi_{RA}))\\
	&~~~~~~~~~~~~~~~~~~~~~~~~~~~~~~~~~~~~~~~~~~-D(\phi_{RA}\|\psi_{RA})\,,
\end{align*}
where $R$ denotes a reference system of arbitrary large dimension. It was shown in \cite{fang2019chain} that
\begin{align}\label{eqamortized}
	D^A(\mathcal{E}\|\mathcal{F})=\lim_{n\to\infty}\,\frac{1}{n}\,D(\mathcal{E}^{\otimes n}\|\mathcal{F}^{\otimes n})\,,
\end{align}
where the channel relative entropy is defined as
\begin{align*}
	D(\mathcal{E}\|\mathcal{F}):=\sup_{\phi_{RA}}D(( \operatorname{id}_R\otimes \mathcal{E})(\phi_{RA})\|( \operatorname{id}_R\otimes \mathcal{F})(\phi_{RA}))\,.
\end{align*}
Upper bounding the amortized channel relative entropy in terms of a single-letter entropic expression is not an easy task in general and it has many applications in the context of  channel discrimination~\cite{berta2018amortized}. The particular case of transferred semigroups is tractable, as they fall under the category of environment-parametrized channels of~\cite[
Proposition 33]{berta2018amortized}. There, the authors show a single-letter expression bound the amortized channel relative entropy.
In the next corollary, we reprove their upper bound in the case when $\mathcal{E}:=T_t$ is the semigroup $T_t$ at time $t$, and $\mathcal{F}:=E_{fix}$ is the conditional expectation onto the fixed point algebra of $(T_t)_{t\ge 0}$, while also obtaining a stronger bound for all R\' enyi entropies:
\begin{corollary}[Chain rule for $p$-R\'{e}nyi divergences]\label{cor2}
	Let $(T_t)_{t\ge 0}$ be a transferred QMS from a finite or Lie group $G$ on $\cB(\cH)$, with corresponding kernel $(k_t)_{t\ge 0}$.	Then for all $t\ge 0$, the following holds
	\begin{align}\label{inequ1}
		D^A(T_t\|E_{fix})\le \int\,k_t\ln k_t\,d\mu_G\,.
	\end{align}
	Moreover, for any $p\ge 1$ and any reference system $R$ the following entropic chain rule holds: for all $\rho,\sigma\in\cD(\mathcal{H}\otimes \mathcal{H}_R)$,
	\begin{align}
		D_p((\operatorname{id}_R\otimes T_t)&(\rho)\|(\operatorname{id}_R \otimes E_{fix})(\sigma) )\nonumber\\
		&~~~~~\le D_p(\rho\|\sigma)+D_p(\nu_t\|\mu_G)\,\label{ineq2}
	\end{align}
	where $d\nu_t=k_t d\mu_G$.
\end{corollary}

\begin{proof}
	For each $n\in\mathbb{N}$, we use (\ref{transfer}) for the transferred semigroup $(T_t^{\otimes n}\otimes \operatorname{id}_{R_n})_{t\ge 0}$ and its corresponding conditional expectation $E_{fix}^{\otimes n}\otimes \operatorname{id}_{R_n}$, where $R_n$ is an arbitrary reference system, so that for any $\rho\in \cD(\cH^{\otimes n}\otimes \operatorname{id}_{R_n})$:
	\begin{align*}
	&	D\big( (T_t^{\otimes n}\otimes \operatorname{id}_{R_n})(\rho)\|(E_{fix}^{\otimes n}\otimes \operatorname{id}_{R_n})(\rho))\\
	&	~~~~~~~~~~~\le\int\,k_t^{\otimes n}\ln\,k_t^{\otimes n}d\mu_G=n\,\int\,k_t\ln\,k_t\,d\mu_G\,.
	\end{align*}
	Inequality (\ref{inequ1}) follows from \Cref{eqamortized}. Inequality (\ref{ineq2}) is a simple consequence of (\ref{ineqfundam}) and the data processing inequality for $E_{fix}$. The statement for other values of $p$ can be proved in a similar way.
\end{proof}
\begin{remark}
	In \cite{fang2019chain}, the following chain rule for relative entropies was derived: given two quantum channels $\mathcal{E},\mathcal{F}$, and for any $\rho,\sigma\in\cD(\cH_R\otimes \cH)$
	\begin{align}
		D((\operatorname{id}_R\otimes \mathcal{E})(\rho)\|(\operatorname{id}_R \otimes \mathcal{F})(\sigma) )\le D(\rho\|\sigma)+D^A(\mathcal{E}\|\mathcal{F})\,,
	\end{align}	
	leaving as an open problem whether such an inequality still holds for $p$-divergences, where $D^A_p(\mathcal{E}\|\mathcal{F})\equiv \lim_{n\to\infty}\,\frac{1}{n}\,D_p(\mathcal{E}^{\otimes n}\|\mathcal{F}^{\otimes n})$. In view of the above Theorem, we found that a similar chain rule holds for $p$-divergences in the restricted case when $\mathcal{E}\equiv T_t$ and $\mathcal{F}\equiv E_{fix}$, up to a weakening of the bound, since
	$$D^A_p(T_t\|E_{fix})\le D_p(\nu_t\|\mu_G).$$
	Similar bounds can be derived for the TRO channels defined in \cite{Gao_2018tro}.
\end{remark}

\subsection{Examples}
It is straightforward to translate the estimates in the last sections to obtain estimates for various capacities of transferred semigroups. We will now illustrate these bounds for three noise models of practical relevance: collective decoherence, depolarizing noise and dephasing noise. As far as we know, these are the first estimates available for capacities of collectively decohering quantum channels. On the other hand, the capacities of depolarizing and dephasing channels are widely studied and we use these examples to benchmark the quality of our bounds. We observe that our bounds show the right exponential decay of the capacities for large time, but are not able to capture it for small times.
The depolarizing and dephasing semigroups are of the simple form discussed before, that is, a difference of a conditional expectation and the identity. Thus, as expected, we may obtain better bounds by a direct application of a modified logarithmic Sobolev inequality.
We will also discuss quantum channels stemming from representations of finite groups that are not of simple form, such as random transposition channels. As mentioned before, it is straightforward to turn any mixing time for a reversible chain on a group into a capacity bound and, thus, this list of examples is by no means exhaustive.
\paragraph{Collective decoherence:} It follows from the results in section~\ref{examplessec} that both weak collective decoherent and strong collective decoherent semigroups will reach their limiting capacity up to an additive error $\epsilon>0$ in time $\mathcal{O}\lb \log\lb\epsilon^{-1}\rb\rb$, showing that encoding into the decoherence-free subspaces of these channels is essentially optimal.

\paragraph{Depolarizing noise:} As discussed in section~\ref{examplessec}, the depolarizing channel
\begin{align*}
	T_{t}^{\textrm{dep}}(\rho)=(1-e^{-t})\frac{I_{\CC^n}}{n}+e^{-t}\rho.
\end{align*}
corresponds to transferred channel we obtain by transfering the uniform random walk on $\ZZ_n\times\ZZ_n$ via the projective representation given by the discrete Weyl matrices. It then follows from the estimates in section~\ref{finitegroups} that the following upper bound holds for the two way quantum capacity of the depolarizing channel (similar estimates would also follow from transfering the heat semigroup on the unitary group with the natural representation):
\begin{align*}
	&\mathcal{Q}_{\leftrightarrow}\lb T^{\textrm{dep}}_{t}\rb\\
	&\leq \frac{\textrm{exp}\lb -\frac{n^2-1}{n^2}t+\frac{n^2-1}{2n^2-4}\log(n^2-1)\log\log(n^2)+1\rb}{2}\,.
\end{align*}
Moreover, we may estimate when the depolarizing channel becomes entanglement breaking and, thus, when the quantum capacity becomes $0$. The estimate we obtain is
\begin{align*}
&	t_{\operatorname{EB}}\lb T^{\textrm{dep}}_{t}\rb\\
	&~~~~~\leq \frac{n^2}{n^2-1}\lb1+\log(2n)\rb+\frac{n^2}{2(n^2-2)}\log\log(n^2)\,.
\end{align*}
To the best of our knowledge, the best available bound on the two way quantum capacity of these channels is the one in~\cite{Pirandola2017}, where they show that
\begin{align*}
	\mathcal{Q}_{\leftrightarrow}\lb T_{t}^{\textrm{dep}}\rb\leq\log(n)-H_2(p_{n,t})-p_{n,t}\log(n-1)
\end{align*}
with $p_{n,t}=\frac{n^2-1}{n^2}(1-e^{-t})$ and $H_2$ the binary entropy for $t\leq\log\lb 1-\frac{n}{n+1}\rb$ and 0 else.
\begin{center}
	\begin{minipage}[t]{0.5\textwidth}
		\includegraphics[width=\textwidth]{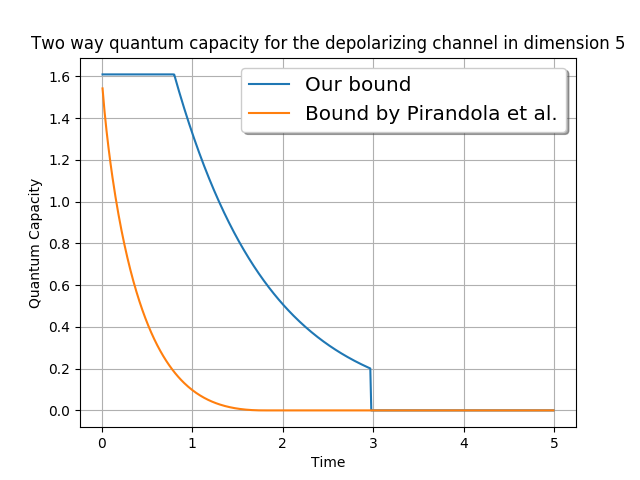}
	\end{minipage}
\end{center}

\paragraph{Dephasing noise:} Again, as discussed in section~\ref{examplessec}, the dephasing semigroup $(T_{t}^{\textrm{deph}})_{t\ge 0}$
corresponds to the transferred channel we obtain by transfering the uniform random walk on $\ZZ_n$. It then follows from the estimates in section~\ref{finitegroups} that the following upper bound holds for the two way quantum capacity:
\begin{align*}
&	\mathcal{Q}_{\leftrightarrow}\lb T_{t}^{\textrm{deph}}\rb\\
&\leq \log(n) \textrm{exp}\lb -t\frac{n-1}{n}-\frac{n-1}{2n-4}\log(n-1)\log\log(n)\rb\,.
\end{align*}
This capacity was computed in~\cite{Pirandola2017}, where they show that
\begin{align*}
	&\mathcal{Q}_{\leftrightarrow}\lb T_{t}^{\textrm{deph}}\rb\\
	&=\log(n)+\lb \frac{1-e^{-t}}{n}+e^{-t}\rb\log\lb \frac{1-e^{-t}}{n}+e^{-t}\rb\\
	&~~~~~~~~~~~~~~~~~~~~~~~~~~~~~~~~~~+\lb1-e^{-t}\rb\log\lb\frac{1-e^{-t}}{n}\rb.
\end{align*}
\begin{center}
	\begin{minipage}[t]{0.5\textwidth}
		\includegraphics[width=\textwidth]{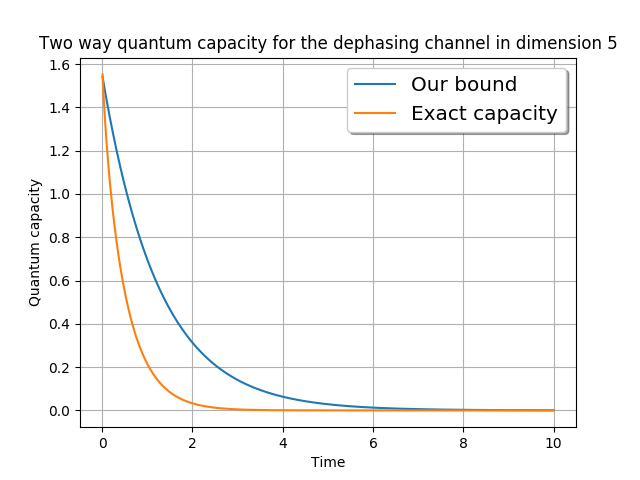}
	\end{minipage}
	
\end{center}
We see that the bounds we obtain using our methods are close to the state of the art, specialized bounds, at large times. Moreover,the slope of the decay is close to the correct one, except at small times.

Again using a MLSI inequality argument, we obtain that the quantum capacity of the dephasing channel is bounded by
\begin{align*}
	\mathcal{Q}\lb T_{t}^{\textrm{deph}}\rb\leq 2e^{-t}\log(n).
\end{align*}

\paragraph{Random SWAP gate:} Consider the natural representation of the permutation semigroup $S_n$
acting on $\lb\mathbb{C}^d\rb^{\otimes n}$ with and the semigroup with on $S_n$ with generator
$L^{RT}(f)(\om)=\frac{1}{n}\sum_{ij} (f(\om)-f(\si^{ij}\om))$.  The transferred semigroup acting on $\lb\mathbb{C}^d\rb^{\otimes n}$ is  then given by the random transposition semigroup
\[ \L^{RT}(\rho) \lel \frac{1}{n} \sum_{ij} (\rho-F_{ij}\rho F_{ij})\pl ,\]
where    $F_{ij}(e\ten f)=f\ten e$  acts between the different registers $i$ and $j$ of the tensor product. It then follows from the results listed in~\Cref{sect_SWAP}
and our results on capacities that for times $t\gl \log n + \frac{1-\ln \eps}{2}$ various capacities of this quantum channel are at most $\eps$  away from their limiting capacity. Similarly, we can estimate decoherence for the nearest neighbour interaction
\[      \L^{NN}(\rho) \lel  \sum_{j=1}^n (\rho-F_{j(j+1)}\rho F_{j(j+1)})\pl .\]
It again follows from the results in~\Cref{sect_SWAP} that for $t\gl \frac{5}{32\pi^2}(2+\ln n-\eps)$ the quantum channel $e^{-t\L^{NN}}$ is $\eps$ away from its various limiting capacities.

It is also possible to bound the capacities through the MLSI constant given in \Cref{transpomodelMLSI} together with \Cref{theoclasscapatransflogo}. This way we get that, for instance, the classical capacity considered here is at most $\e^{-t}\log{n!}\sim \e^{-t}n\log n$ away from the limiting capacity. However, in this case we know \cite{Zhao} that $\al_c(\L^{RT})\gl 1$ and hence \Cref{class} shows that for  $\eps=e^{-t}\log d^n=e^{-t}n\log d$, the channel is $\eps$ away from its limiting classical capacity. This is better than both bounds.

Similarly, for the case of nearest-neighbor swaps, the results listed in~\Cref{sect_SWAP} and for instance \Cref{theoclasscapatransf} show that for times $t\sim cn^2(\log(n)-\ln \eps)$, the capacities considered here  quantum channel are at most $\eps$ away from their limiting capacity.

\section{Functional inequalities and estimates for decoherence times: beyond transference}\label{sect3}

The different applications of transference so far have been based on the decoherence time of the transferred QMS. We recall that for a general QMS (not necessarily selfadjoint), the decoherence time is defined for any $\eps>0$ as
\begin{align}\label{mixtimeTtbis}
	&t_{\operatorname{deco}}(\eps)\\
	&:=\inf\{t\ge 0;\,\|T^\dagger_t(\rho)-E^\dagger_{fix}(\rho)\|_1\le \eps~~\forall\rho\in\mathcal{D}(\mathcal{H})\}\,.\nonumber
\end{align}
We also recall the definition of the \textit{mixing time}\footnote{Remark the difference in the normalization of the norms in both definitions. In the quantum case, density matrices are normalized with respect to the unnormalized trace whereas in the classical case, we look at the evolution of states normalized with respect to the probability distribution $\mu_G$.} of a classical primitive Markov process $(S_t)_{t\ge 0}$
\begin{align}\label{mixtimeStbis}
	&t_{\operatorname{mix}}(\eps) \\
	&= \inf\{ t\geq0\,|\, \|S_t(f)-\EE_{\mu_G}(f)\|_{L_1(\mu_G)}\,\le\, \eps\,\,\nonumber\\
	&~~~~~~~~~~~~~~~~~~~~~~~~~~~~~~~~~~~~~~~\forall f\ge 0,\,\EE_{\mu_G}[f]=1\} \,.\nonumber
\end{align}

In this section, we study some aspects of the theory of decoherence time that goes beyond transference. In this matter, it appears that the most relevant quantity is a regularized form of the decoherence time that we call the \emph{complete decoherence time} (cb-decoherence time) and defined as
\begin{align}\label{cbdecotimeTt}
&	t_{\operatorname{deco}}^{cb}(\eps)\\
	&:=\inf\{t\ge 0;\,\|T^\dagger_t-E^\dagger_{fix}\,:\,L_1(\tr)\to L_1(\tr)\,\|_{cb}\le \eps\}\,.\nonumber
\end{align}
Here $L_1(\Tr)$ denotes the $L_1$ space defines with respect to the unnormalized trace. Remark that the only difference with \Cref{mixtimeTtbis} is the choice of the completely bounded norm instead of the usual one. We thus have the trivial bound:
\[t_{\operatorname{deco}}(\eps)\leq t_{\operatorname{deco}}^{cb}(\eps)\,,\]
which shows that the cb-decoherence time controls the usual decoherence time. Remark that for classical semigroups, both definitions coincide thanks to \Cref{lem_CB_commutative} and therefore one can overlook this notion of decoherence time for transferred QMS.\\

In the application to quantum capacities , it will yet be an other possible definition of the decoherence time that will be useful. We thus define the \emph{ultracontractive (complete) decoherence time} (UC-cb decoherence time for short) as 
\begin{align}\label{def_UC_cb_decotime}
	&t_{1,\infty}^{cb}(\eps):=\inf\{t\ge 0;\,\\
	&~~~~\|T^\dagger_t-E^\dagger_{fix}\,:\,L_1(\tr)\to L_1^\infty(\Nfix\subset\cB(\cH))\,\|_{cb}\le \eps\}\,.\nonumber
\end{align}
The term ``ultracontractive'' comes from the associated functional inequality that we define and study in \Cref{sect33}. It is also closely related to an other functional inequality, namely hypercontractivity, that we study in \Cref{appendix_HC}. Finally it also provides an upper bound on the decoherence time by the ordering of the $L_p$ norms.

One can generalize these definitions to include general amalgamated $L_p$ spaces for $1\leq p\leq q\leq \infty$:
\begin{align}\label{def_pq_cb_decotime}
	&t_{p,q}^{cb}(\eps):=\inf\{t\ge 0;\,\\
	&\|T^\dagger_t-E^\dagger_{fix}\,:\,L_s^p(\Nfix\subset\cB(\cH))\to L_s^q(\Nfix\subset\cB(\cH))\,\|_{cb}\le \eps\}\,.\nonumber
\end{align}
The relationship between all these definitions will be studied in \Cref{sect3interpolation}, following the original approach of Saloff-Coste in the classical case~\cite{saloff1994precise}. We then focus on the application to quantum capacities in \Cref{sect3capacities}

\subsection{Interpolation  method}\label{sect3interpolation}

The choice of the trace norm as the measure of distance to equilibrium in the definition of the decoherence time is justified by its operational interpretation as a measure of distinguishability between two density matrices. One could also ask how this compares with other choices, such as other $L_p$ norms. This makes even more sense in the quantum situation, where different non-commutative norms appear: conditioned and completely bounded. In the classical case, this was discussed by Saloff-Coste in \cite{saloff1994precise} using interpolation theory. We first briefly sketch the main message in the classical setting (see \cite{saloff1994precise}): given a classical primitive Markov semigroup $(S_t)_{t\ge 0}$ on a group $G$, with generator $L$ and Markov kernel $(k_t)_{t\geq0}$, Saloff-Coste proposed to study for all $1\leq p\leq +\infty$ the quantities

\begin{align*}
v^L_{1,p}(t)&:=\underset{x\in G}{\sup}\,\norm{k_t(x,\cdot)-\Ind}_p\\
&=\|S_t-\EE_{\mu_G}:L_1(\mu_G)\to L_p(\mu_G)\| \pl.
\end{align*}
Remark that the mixing-time corresponds to the study of $v_{1,1}^L(t)$. This definition is justified by the ordering of the $L_p$ norms, which implies that for any $1\leq p\leq q$ and any $t\geq0$,
\[v^L_{1,1}(t)\leq v^L_{1,p}(t) \leq v^L_{1,q}(t)\,.\]
More generally, we define
\[ v^L_{p,q}(t)  := \|S_t-\EE_{\mu_G}:L_p(\mu_G)\to L_q(\mu_G)\| \pl.\]
In the quantum case, let $(T_t=e^{-t\L})_{t\geq0}$ be a QMS on $\Bcal(\Hcal)$. We do not assume that $(T_t)_{t\geq0}$ is a transferred QMS or that it is selfadjoint. Mimicking the classical case, we define:
\begin{align}\label{eq_def_decotime}
&	v^\cL_{p,q}(t) :=\\
& \|T_t-E_{fix}:L_{s}^p(\Nfix\,\subset\, \cB(\cH))\to L_{s}^q(\Nfix\,\subset \,\cB(\cH))\|  \,.\nonumber
\end{align}
We denoted by $v^{\cL,\cb}_{p,q}$ the same quantity but defined with respect to the cb norm. We recall that in both cases, the definition of the norm is independent of the choice of the subscript $s$ which is thus arbitrarily taken to be equal to $\infty$ or $1$.\\
The following proposition remains true when replacing $v^\cL_{p,q}$ by $v^{\cL,\cb}_{p,q}$.

\begin{proposition}\label{iii} With the above notations, we have for any $1\leq p\leq q\leq +\infty$
	\[v^\cL_{p,q}(t)\le 2\,v^\cL_{1,\infty}(t)^{\frac{1}{p}-\frac{1}{q}}\,.\]
\end{proposition}

\begin{proof} We first observe that
	\[ \|T_t-E_{fix}:L_1(\cB(\cH))\to L_1(\cB(\cH))\|\le 2\]
	and hence by interpolation for $\frac{1}{q}=\frac{1-\theta}{\infty}+\frac{\theta}{1}$ we have
	\begin{align*}
&	 \|T_t-E_{fix}:L_1(\cB(\cH))\to L_1^q(\Nfix\subset \cB(\cH))\|\\
	&~~~~~~~~~~~~~~~~~~~~~~~~~~~~~~~~~\kl  2^{1/q} v_{1,\infty}^\L(t)^{1-1/q}  \pl .
	\end{align*}
	For the next step we interpolate this inequality with
	\[  \|T_t-E_{fix}:L_q(\cB(\cH))\to L_q^q(\Nfix\subset \cB(\cH))\| \kl 2 \]
	and get ($\frac{1}{p}=\frac{1-\eta}{1}+\frac{\eta}{q}$) that
	\begin{align*}
	&\|T_t-E_{fix}:L_p(\cB(\cH))\to L_p^q(\Nfix\subset \cB(\cH))\|\\
&	\kl 2^{\eta} 2^{(1-\eta)/q}
	v^\cL_{1,\infty}(t)^{(1-\eta)(1-1/q)}
	\kl 2\, v^\cL_{1,\infty}(t)^{1/p-1/q} \pl .
	\end{align*}
	The proof for the cb-norm is identical.
\end{proof}

One can get a finer control of the decoherence time from the simple remark that, for a selfadjoint QMS $(T_t)_{t\geq0}$ (not necessarily transferred), as $T_t\circ E_{fix}=E_{fix}$ for all $t\geq0$:
\begin{align}
	&v^\Lcal_{1,1}(t+s+r)\nonumber\\
	&~~~\leq v^\cL_{1,2}(t+s+r) \nonumber \\
	&~~~\le \norm{T_t:L_\infty^1(\Nfix\,\subset\,\Bcal(\Hcal))\to L_\infty^q(\Nfix\,\subset\,\Bcal(\Hcal))} \nonumber\\
	&~~~~~~~\|T_s:L_\infty^q(\Nfix\,\subset\,\Bcal(\Hcal))\to L_\infty^2(\Bcal(\Hcal))\|\,v^\Lcal_{2,2}(r)\,.  \label{decomposition}
\end{align}
Now in the last term, each individual term can be estimated using particular functional inequalities (resp. ultracontractivity (UC), hypercontractivity (HC) and Poincar\'e inequality (PI)) of the classical semigroup using transference. In practice, since these estimates are independent of the representation of the transferred QMS, we expect that they provide bounds that are less tight than the ones one would get if one had access to the exact UC/HC/PI constants of the QMS.\\

We conclude this section with a property which allows to connect hypercontractive estimates with ultracontractive ones (see \Cref{appendix_HC}).
\begin{proposition}\label{propiii}
	Let $(T_t)_{t\geq0}$ be a selfadjoint QMS. Then
	\begin{equation}\label{eq_iii}
		v^\cL_{1,2}(t)^2 \lel v^\cL_{1,\infty}(2t)
	\end{equation}
\end{proposition}

\begin{proof}
	In the selfadjoint setting, we exploit that
	\begin{align*}
	 &\|T_t-E_{fix}:L_u^1(N_{fix}\subset \cB(\cH))\to L_u^2(N_{fix}\subset \cB(\cH))\|\\
	 &=\|T_t-E_{fix}:L_u^2(N_{fix}\subset \cB(\cH))\to L_u^{\infty}(N_{fix}\subset \cB(\cH))\|   
	 \end{align*}
	together with the fact that $T_{2t}-E_{fix}=(T_t-E_{fix})^2$, in order to get
	\begin{align*}
	  &\|T_{2t}-E_{fix}:L_2^1(N_{fix}\subset \cB(\cH))\to L_2^{\infty}(N_{fix}\subset \cB(\cH))\|\\
	  & ~~~~~~~~~~~~\leq \|T_t-E_{fix}:L_2^1(\Bcal(\Hcal))\to L_2(\cB(\cH)))\|^2 \pl .
	  \end{align*}
	To prove the other inequality, we use that by definition
	\begin{align*}
	&	\|T_{t}-E_{fix}:L_2^1(N_{fix}\subset \cB(\cH))\to L_2(N_{fix}\subset \cB(\cH))\|^2\\
		& =     \underset{\norm{x}_{L_2^1(N_{fix}\subset\Bcal(\Hcal))}\leq1}{\sup}\,\norm{(T_t-E_{fix})(x)}_{L_2(\tau)}^2  \\
		& \leq  \underset{\norm{x}_{L_2^1(N_{fix}\subset\Bcal(\Hcal))}\leq1}{\sup}\,\langle (T_t-E_{fix})(x)\,,\,(T_t-E_{fix})(x)\rangle  \\
		& \leq \underset{\norm{x}_{L_2^1(N_{fix}\subset\Bcal(\Hcal))}\leq1}{\sup}\, \norm{(T_{2t}-E_{fix})(x)}_{L_2^\infty(N_{fix}\subset\Bcal(\Hcal))} \\
		& \leq \|T_{2t}-E_{fix}:L_2^1(N_{fix}\subset \cB(\cH))\to L_2^{\infty}(N_{fix}\subset \cB(\cH))\|  \,,
	\end{align*}
	where in the third line we use again that $T_t$ is selfadjoint together with $T_{2t}-E_{fix}=(T_t-E_{fix})^2$ and the H\"older inequality for the conditioned $L_p$ norms.
\end{proof}

\subsection{Complete ultracontractivity}\label{sect33}

In this section we introduce (complete) ultracontractivity (or, more generally, the Varapoulos dimension) for general selfadjoint QMS, not nessecarily transferred ones. These provide better estimates for small times.

\begin{definition}\label{def_Varapoulos}
	We say that a (not necessarily primitive) QMS has $(p,q)$-\emph{Varapoulos dimension} $\al$ if there exists $c_{p,q}>0$ and $t_0>0$ such that for all $t\le t_0$
	\begin{align}
	&	\|T_t:L_{p}(\cB(\cH))\to L_{p}^{q}(\Nfix\,\subset \,\cB(\cH))\|_{\cb}
	\nonumber\\
	&	\kl c_{p,q}(\al)\, t^{-\frac{\al}{2}(1/p-1/q)}\,.\tag{$R_{p,q}(\al,t_0)$}
	\end{align}
	For $p=1$, $q=+\infty$ and $t_0=1$, we say that $(T_t)_{t\geq0}$ is \emph{ultracontractive}, and we denote it by $\operatorname{UC}(C_{\al},\alpha)$.
\end{definition}
The reason for the factor $1/2$ comes from the behaviour of the heat kernel on $\rz^n$ or $\mathbb{T}^n$. A beautiful extrapolation argument by Raynaud shows that heat kernel estimates for small times essentially do not depend on $p,q$. Again, we give a prove for general selfadjoint QMS, independently of the transference construction.

\begin{lemma}\label{ray}
	For all $1\leq p\leq q$, $R_{p,q}(\al,t_0)$ and $R_{1,\infty}(\al,t_0)$ are equivalent (i.e. they hold equivalently up to the constant $c_{p,q}$).
\end{lemma}

\begin{proof} Let us define $\beta=\frac{\al}{2}$ and
	\begin{align*}
	 r_{p,q} &\lel \sup_{0<t\le t_0} t^{\beta(1/q-1/p)}\\
	 &
	\|T_t:L_{\infty}^p(\Nfix \subset \cB(\cH))\to L_{\infty}^q(\Nfix \subset \cB(\cH))\|_{\cb} \pl .
	\end{align*}
	Since $T_t$ is subunital we known that $\|T_t:L_q(\cB(\cH))\to L_q(\cB(\cH))\|_{\cb}\le 1$. Thus the same interpolation argument as in the proof of \Cref{iii} implies
	\[ r_{p,q}\le r_{1,\infty}^{1/p-1/q} \pl .\]
	Now let $1<p<q$ and $t\le t_0$ and $x\in \cB(\cH)$ such that $\|x\|_1=1$. We find that
	\begin{align*}
		&t^{\beta(1/q-1)} \|T_{t}(x)\|_{L_1^{q}(\Nfix \subset \cB(\cH))}\\
		&\le t^{\beta(1/q-1)} r_{p,q} (t/2)^{\beta(1/p-1/q)} \|T_{t/2}(x)\|_{L_1^p(\Nfix \subset \cB(\cH))} \\
		&\le r_{p,q} \pl 2^{\beta(1-1/q)} (t/2)^{\beta(1/p-1)}\\
		&~~~~~~~~~~~~~~~~~~
		\|T_{t/2}(x)\|_{L_1^q(\Nfix \subset \cB(\cH))}^{1-\theta} \|T_t(x)\|_{L_1(\Bcal(\Hcal))}^{\theta} \\
		& \le r_{p,q} \pl 2^{\beta(1-1/q)} (\sup_{t\le t_0} t^{\beta(1-1/q)}\|T_tx\|_{L_1^q(\Nfix \subset \cB(\cH))})^{1-\theta}
		\pl .\end{align*}
	Here $\frac{1}{p}=\frac{1-\theta}{q}+\frac{\theta}{1}$. For $x\in \Bcal(\Hcal))$ the supremum is finite, and hence
	\[ \sup_{t<t_0} t^{\beta(1-1/q)}\|T_t(x)\|_{L_1^q(\Nfix \subset \cB(\cH))}
	\kl 2^{\frac{\beta(1-1/q)}{1-\theta}} r_{p,q}^{1/1-\theta} \pl .\]
	By approximation the assertion follows for all $x\in L_1(\cB(\cH))$.
\end{proof}
In certain situations, it may happen that $R_{p,q}(\al,t_0)$ holds only for a short time $t_0<1$. However in this article we will only consider example where we can take $t_0=1$. For sake of clarity, we thus present our result only in this case. The general case $t_0>0$ would follow similarly by considering for instance $\hat{T}_t=e^{-t}T_t$.\\

Finally, using ultracontractivity (or more generally the notion of Varapoulos dimension), we obtain control on different decoherence times of a QMS.
\begin{theorem}\label{HH} Let $\L$ be the generator of a transferred QMS with spectral gap $\la_{\min}(\L)$ which satisfies $\operatorname{UC}(C,\alpha)$ and recall the definition:
	\[t_{p,q}^{\cb}(\eps):=\inf\{t\geq0\,;\,v_{p,q}^{cb}(t)\leq\eps\}\,.\]
	Then
	\[ t_{p,q}^{\cb}(\eps)
	\kl 1+ \frac{1}{\la_{\min}(\L)}\left(\ln C_{\al}+\frac{\ln(1/\eps)}{1/p-1/q}\right) \pl .\]
\end{theorem}
\begin{proof} By transference,  we have
	\begin{align*}
	&	\|T_{1+t}-E:L_1(\cB(\cH))\to L_1^2(\Nfix\subset \cB(\cH))\|_{\cb}\\
		&\le \|T_1:L_1(\cB(\cH))\to L_1^2(\Nfix\subset \cB(\cH))\|_{\cb}\quad e^{-\la_{\min}(\Lcal)t} \\
		&\le C_{\al}^{1/2} \,e^{-\la_{\min}(\Lcal)t} \pl .
	\end{align*}
	Thanks to \Cref{iii} this implies
	\[ v_{1,\infty}^{cb}(1+t)\kl C_{\al}\, e^{-\la_{\min}(\Lcal)t} \]
	and
	\[ v_{p,q}^{cb}(1+t)\kl C_{\al}^{1/p-1/q}\,e^{-\la_{\min}(\Lcal)(1/p-1/q)t} \pl .\]
	This implies the assertion after taking logarithms.
\end{proof}

Because of \Cref{lem_CB_commutative}, the application to transferred QMS is straightforward. Again, notice that we can use the \emph{quantum} spectral gap instead of the \emph{classical} one, when a direct evaluation is possible.
\begin{corollary}
	Let $(T_t)_{t\geq0}$ be a transferred QMS, with classical Markov semigroup $(S_t)_{t\geq0}$. Assume that $(S_t)_{t\geq0}$ is primitive with spectral gap $\lambda_{\min}(L)$ and that it is \emph{ultracontractive} with constants $C,\alpha>0$:
	\[ \|S_t:L_1(\mu_G)\to L_{\infty}(\mu_G)\|
	\lel \sup_{g\in G}\, |k_{t}(g)|\leq C\,t^{-\alpha/2}\,.\]
	Then
	\[ t_{p,q}^{\cb}(\eps)
	\kl 1+ \frac{1}{\la_{\min}(\L)}\left(\ln C+\frac{\ln(1/\eps)}{1/p-1/q}\right) \pl .\]
\end{corollary}
We already discussed in \Cref{sect23} how such ultracontractivity holds in general for H\"ormander systems and finite groups.

A classical result from Varapoulos says that ultracontractivity is equivalent to the following \emph{Sobolev inequality} (see \cite{varobook}):
\[\norm{f}_{L_\theta(\mu_G)}^2 \leq C\,\left(-\langle f\,,\,Lf\rangle_{\mu_G}+\norm{f}_{L_2(\mu_G)}^2\right)\,,\]
where $\theta=2\alpha/(\alpha-2)$. This inequality itself implies Nash inequality:
\[\norm{f}_{L_2(\mu_G)}^{2(1+2/\alpha)} \leq C\,\left(-\langle f\,,\,Lf\rangle_{\mu_G}+\norm{f}_{L_2(\mu_G)}^2\right)\norm{f}_{L_1(\mu_G)}^{4/\alpha}\,.\]
In the quantum case, this last inequality was studied in \cite{[KT16]} in the case of primitive selfadjoint QMS. In \cite{Zhao} it was shown that Varopoulos' theorem remains true in the cb-category.

\subsection{Application to quantum capacities}\label{sect3capacities}

The value $t_{1,\infty}^{cb}(\eps)$ is a complete decoherence  time, usually bigger than $t_{mix}$ and $t_{deco}$. It provides a universal bound for convex functions on channels.

\begin{proposition}\label{convex} Let $T_t=e^{-t\L}$ be a selfadjoint semigroup with fixpoint algebra $N_{fix}$ and conditional expectation $E_{fix}$, let $\eps>0$ and $t\gl t_{1,\infty}^{cb}(\eps)$. Let $\al$ be convex function on channels. Then
	\[ \al(E_{fix}) \kl \al(T_t) \kl (1-\eps)\al(E_{fix})+ \eps \sup_{S} \al(S) \]
	where the supremum is take over all channels $S$ with the same input and ouput dimension.
\end{proposition}
We need the following observation due to Li Gao: 
\begin{lemma}(Li Gao)\label{lem_Gao} Let $N\subset \Mz_n$ and
	$T:\Mz_n\to \Mz_n$ be a completely positive $N$-bimodule map and $0<\eps<1$ such that
	\[ \|T-E:L_1(\Mz_n) \to L_1^{\infty}(N\,\ssubset \, \Mz_n)\|_{cb} \kl \eps \pl .\]
	Then there exists a completely positive $N$-bimodule map $\Phi$  such that
	\[ T\lel (1-\eps)E+\eps\Phi \pl .\]
\end{lemma}

\begin{proof}[Proof of \Cref{lem_Gao}] Indeed, we refer to \cite{gao2018fisher} for the fact that an $N$ bimodule map admits a modified Choi matrix $\chi_T$, and that  $\|\chi_{T}-1\|\le \eps$ holds iff $\|T-E_N:L_1(\Mz_n)\to L_1(N\subset \Mz_n)\|\le \eps$. Then we deduce that 
	\[ (1-\eps)1\kl  \chi_T \kl (1+\eps) 1 \pl .\] 
	This in turn is equivalent to 
	\[ (1-\eps)E_N\pl\le_{cp}T \le_{cp} (1+\eps) E_N  \pl, \]
	where $\leq_{cp}$ means that the inequality remains true for all ampliations $E_N\otimes \id_{\Mz_m}$ and $T \otimes \id_{\Mz_m}$ for all positive integers $m\geq0$.
	In particular,
	\[ T \lel T-(1-\eps)E+(1-\eps)E \lel \eps \Phi+(1-\eps)E \pl .\]
	Here $\Phi=\frac{T-(1-\eps)E}{\eps}$ is normalized so that $\Phi^{\dag}(1)=1$, i.e. $\Phi$ is a channel. Obviously, it is also an $N$ bimodule map.
	\qd
	
	\begin{proof}[Proof of \Cref{convex}]
		Since we assume $\al$ to be convex, we find
		\begin{align*}
			\al(T_t) 
		&\lel \al((1-\eps)E_{fix}+\eps \Phi)\\
			&\kl (1-\eps)\al(E_{fix})+\eps \al(\Phi)\\
		&	\kl  (1-\eps)\al(E_{fix})+\eps \sup_{S} \al(S) \pl .
		\end{align*}
		The lower bound follows from convexity, because $E_{fix}$ is obtained as an average. \qd

		\begin{remark}
			Many capacities are either convex functions of the channel, for example the entanglement assisted capacity, or admit a controllable convex roof, for example the side-channel assisted capacity $Q\le Q_{ss}$ from \cite{winterco}. Assuming that the maximal value $\al^*=\sup_{S} \al(S)$ is of order $\log n$, we see that for $t\gl t_{1,\infty}^{cb}(\eps)$, the leading term of $\al(T_t)$ is $\al(E_{fix})$. In that sense,  $t_{1,\infty}^{cb}(\eps)$  is a ``universal'' coherence time. In the transference situation we can get a hold of this constant via commutative methods.
		\end{remark}

\appendices
\section{Entropy Comparison Theorem}\label{normestimates}
Let $G$ a compact group with Haar measure $\mu_G$ and let $g\mapsto u(g)$ be a projective representation of $G$ on some finite dimensional Hilbert space $\Hcal$. For a bounded measurable function $k:G\to\RR^+$, we define the operator $\Phi_k$ on $\Bcal(\Hcal)$ as:
\begin{equation*}
	\Phi_{k}(\rho) := \int k(g) \,u(g)^*\rho \,u(g) \,d\mu_{{G}}(g) \, .
\end{equation*}
We recall that the fixed-point algebra of the map $\Phi_k$ is given by the commutant of $u(G)$:
\[ N_{fix} = \{ \sigma \in \cB(\cH)\,| \,\sigma u(g)=u(g)\sigma\} = u(G)' \]
and that the following bimodule property holds: for any $\sigma_1,\sigma_2\in N_{fix}$,
\[ \Phi_k(\sigma_1\, \rho\, \sigma_2) = \sigma_1 \,\Phi_k(\rho) \,\sigma_2 \, .\]
\begin{theorem}\label{comparison} Let $x \in N_{fix}^+$ and $k:\,G\to\RR^+$ a bounded measurable function such that $\int k \,d\mu_G=1$. Then, for any positive $y\in\cB(\cH)$, and any $p\ge 1$ of H\"{o}lder conjugate $\hat{p}$:
	\[ \|x^{-\frac{1}{2\hat{p}}}\Phi_k(y)x^{-\frac{1}{2\hat{p}}}\|_{p}\,
	\leq\, \|k\|_{p}\, \|x^{-\frac{1}{2\hat{p}}}E_{fix}(y)x^{-\frac{1}{2\hat{p}}}\|_{p} \, .\]
	Moreover, for any states $\rho,\sigma\in\cD_+(\cH)$ such that $\sigma\in \cD(N_{fix})$:
	\begin{align}\label{thatstheoneforchannel}
		D(\Phi_k(\rho)||\sigma) \leq D(E_{fix}(\rho)||\sigma) + \int_G k\, \log k \,d\mu_G  \, .
	\end{align}
\end{theorem}

In order to prove \Cref{comparison}, we need to introduce a few functional analytical notions: given two (possibly infinite dimensional) Hilbert spaces $\cH$ and $\cK$, we denote by $\cT_p(\cH,\cK)$ the Banach space of linear operators $x:\,\cH\to\cK$ with norm
\begin{align}\label{defshatten}
	\|x\|_{\cT_p(\cH,\cK)} := \|xx^*\|_{\cT_{p/2}(\cK)}^{1/2}
	= \|x^*x\|_{\cT_{p/2}(\cH)}^{1/2} \,,
\end{align}
where $\|a\|_{\mathcal{T}_q(\cH)}:=\tr(|a|^q)^{\frac{1}{q}}$. With a slight abuse of notations, we will also denote this norm by $\|x\|_p$. Next, given two Hilbert spaces $\cH$ and $\cK$ and a positive invertible element $\sigma \in \cB(\cH)$ we recall that the Kosaki norms $\|x\|_{\widehat{{L}}_p(\sigma)}: = \|\sigma^{\frac{1}{p}}\,x \|_{\mathcal{T}_p(\cH,\cK)}$ form an interpolation family (see. Theorem 4 of \cite{Gao_2018tro} for more details). We will also use the column and row spaces. For a Hilbert space $\cH$: we denote the column space $\cH^c:=\cB(\CC,\cH)$ and row space $\cH^r:=\cB(\cH,\CC)$. From this one may construct the interpolation spaces $\cH^{c_p} = [\cH^c, \cH^r]_{1 / p}$ and $\cH^{r_p} = [\cH^r,\cH^c]_{1/p}$ (see \cite{Pis03} for an introduction to these spaces and their properties). Importantly, we also recall that if a pair of spaces $X_0$ and $X_1$ form an interpolation scale $[X_0,X_1]_\theta$, and if $Y_0 \subseteq X_0$ and $Y_1 \subseteq X_1$ are given by the \textit{same} projector, then $[Y_0,Y_1]_\theta$ is also an interpolation scale. The Haagerup tensor product will also play an important role in the proof of \Cref{comparison}. It is defined as follows: given two operator spaces  $X\subset \cB(\cH)$ and $Y\subset \cB(\cK)$, the Haagerup tensor norm is defined on $X\otimes Y$ as
\begin{align*}
	&\|z\|_{X\otimes_h Y}:=\\
	&\inf_{z=\sum_{k}x_k\otimes y_k}\,\|\big(  \sum_{k}x_k\,x_k^*  \big)^{1/2}  \|_{\cB(\cH)}\, \|\big(  \sum_{k}y_k^*\,y_k  \big)^{1/2}  \|_{\cB(\cK)}\,.
\end{align*}
Finally, we recall that a \textit{ternary ring of operators} (TRO) is a closed operator subspace $X$ of $\cB(\cH,\cK)$ with the property
\begin{equation}
	x y^* z \in X \text{ for all } x,y,z \in X \pl.
\end{equation}
If $X$ is a TRO and if $\sigma$ is in the left algebra $ \mathcal{L}(X):=\operatorname{span}\{xy^*|\,x,y\in X\}\subset \cB(\cK)$, we may construct an interpolation scale $X_{p, \sigma} = [X_\infty, X_{1,\sigma}]_{1/p}$ as the space $X$ equipped with the norm $\| x \|_{\sigma, p} = \|\sigma^{\frac{1}{p}} x \|_p$. For more information, see Theorem 5.2 in the appendix of \cite{Gao_2018tro}.
Let $\cH'$ be a copy of $\cH$ that corresponds to the output space of the channel $\Phi_k$. Now, given $\psi \in \cH$, let $\eta_\psi : L_2(\mu_G) \rightarrow \cH'$ and its adjoint be given by
\begin{align*}
	&\eta_\psi(\chi) = \int_G  \chi(g)\,u(g)^* \psi\, d \mu_G(g), \pl\,\,\, \pl\\
	 &\eta_\psi^*(\varphi)(g) = \braket{\psi,\, u(g) \varphi}_\cH \pl.
\end{align*}
Denote by $V : \cH \rightarrow \cH'\otimes L_2(\mu_G)$ the Stinespring dilation of the conditional expectation that is given by $V\psi=(g\mapsto u(g)\psi)$, so that
\begin{align*}
	E_{fix}(\rho)&= (\operatorname{id}_{\cB(\cH')}\otimes \mathbb{E}_{\mu_G})(V\rho V^*)\\
	&= \int_G\,u(g)\,\rho\,u(g)^*\,d\mu_G(g)\,.
\end{align*}
Next, identify the space $V(\cH)$ with a subspace of $L_2(\mu_G)^r \otimes_h \cH'^c \cong \cB( L_2(\mu_G),\cH' )$, so that $V(\psi) = \eta_\psi$.
\begin{lemma} \label{lem:vatro}
	$V(\cH)$ is a TRO.
\end{lemma}
\begin{proof}
	Let $\psi,\varphi,\tilde{\psi},\tilde{\varphi} \in \cH$. We compute
	\begin{equation}
		\begin{split}
			\braket{\tilde{\psi}, \eta_\psi \eta_\varphi^*(\tilde{\varphi})}_{\cH}
			& = \int_G \braket{\tilde{\psi} ,\,u(g)^* \psi}\,\braket{\varphi, u(g) \tilde{\varphi}}  d \mu_G(g) \\
			&= \tr(\ket{\tilde{\varphi}}\bra{\tilde{\psi}} E_{fix}(\ket{\psi}\bra{\varphi})) \\
			& = \braket{\tilde{\psi}, E_{fix}(\ket{\psi}\bra{\varphi}) \tilde{\varphi}}_\cH \\
			&\equiv \braket{\tilde{\psi}, \sigma \tilde{\varphi}}_\cH,
		\end{split}
	\end{equation}
	where $\sigma := E_{fix}(\ket{\psi}\bra{\varphi})$. Since this holds for all $\tilde{\psi}\in\cH$, $\eta_\psi \eta_\varphi^*(\tilde{\varphi}) = \sigma \tilde{\varphi}$.
	We also have that, since $\sigma\in u(G)'$, for any $\psi'\in\cH$
	\begin{align*}
		\eta_{\psi'}^*(\sigma \tilde{\varphi})(g) = \braket{\psi', u(g) \sigma \tilde{\varphi}}_\cH &= \braket{\sigma^* \psi', u(g) \tilde{\varphi}}_\cH \\
		&= \eta_{\sigma^* \psi'}^*(\tilde{\varphi})(g) \pl.
	\end{align*}
	Hence $\eta_{{\psi'}}^* \eta_\psi \eta_\varphi^* = \eta_{\sigma^* {\psi'}}^*$. Since $E_{fix}(\ket{\psi}\bra{\varphi})^* \psi' \in \cH$, $\eta_{\sigma^* \psi'}^* \in V(\cH)^*$, so that $V(\cH)$ is a TRO.
\end{proof}
Next, for any $\Omega \in\mathcal{T}_2(\cH)$, define the following dual operators:
\begin{align*}
&	\hat{\eta}_\Omega:\left\{\begin{array}{cc}
		\!\!\!\!\!	\!\!\!\!\!\!\!\!\!\!\!\!\!\!\!\!\!\!\!\!\!\!	\cH\otimes{L}_2(\mu_G)\to\cH'&\\
		\chi\mapsto \int_G\,u(g)^*\,\Omega \,\chi(g)\,d\mu_G(g)&
	\end{array}
	\right.\,,\\
	&	\hat{\eta}_\Omega^*:\left\{
	\begin{array}{cc}
		\!\!	\!\!\!\!\!\!	\cH'\to\cH\otimes{L}_2(\mu_G)&\\
		\psi\mapsto (g\mapsto \Omega^* \,\,u(g)\psi)\,.&
	\end{array}
	\right.
\end{align*}	
By definition, the operators $\hat{\eta}_\Omega$ span the space $V(\cH)\otimes_h\cH^r$.	Since $V(\cH)$ is a TRO, as shown in Lemma \ref{lem:vatro}, $V(\cH)\otimes_h\cH^r$ is also a TRO.

\begin{lemma} \label{lem:interpscale}
	Let $1 \leq p \leq \infty$ of H\"{o}lder conjugate $\hat{p}$, and $L_{k}:L_2(\mu_G)\to L_2(\mu_G)$ denote left multiplication by $k$ for $k :G\to\RR^+ $ bounded measurable. Then, for any $\Omega\in\mathcal{T}_2(\cH)$ and positive operator $\sigma\in  \cL(V(\cH))\equiv N_{fix}$,
	\begin{align} \label{eq:kosakiscale}
		\|\sigma^{-\frac{1}{2\hat{p}}} \,	\hat{\eta}_\Omega(\Id_{\cH} \otimes L_{k}^{\frac{1}{2}})\,
		\, \|_{2p}
		\le \|k^{\frac{1}{2}}\|_{{L}_{2p}(\mu_G)}\,	\|\sigma^{-\frac{1}{2\hat{p}}} \, \hat{\eta}_\Omega\|_{2p}\,.
	\end{align}
\end{lemma}

\begin{proof}
	By Theorem 5.2 of \cite{Gao_2018tro}, for any positive operator $\sigma\in\mathcal{L}(V(\cH)\otimes_h\cH^r)\equiv \mathcal{L}(V(\cH))\equiv N_{fix}$, the spaces $\tilde{X}_{p,\sigma}:=V(\cH)\otimes_h\cH^r$ equipped with the Kosaki norm $\| x\|_{\widehat{L}_p(\sigma)}$ form an interpolation scale. In particular, $\tilde{X}_{2p,\sigma}=[\tilde{X}_\infty,\tilde{X}_{2,\sigma}]_{\frac{1}{p}}$. Next, observe that $\sigma^{-\frac{1}{2}}\hat{\eta}_\Omega\in V(\cH)\otimes_h\cH^r$.	Therefore, assuming that $\| \sigma^{-\frac{1}{2}}{\hat{\eta}}_\Omega \|_{\widehat{L}_p(\sigma)} < 1$, then there exists an analytic function $\xi : \{ z \in \mathbb{C} : 0 \leq \operatorname{Re}(z) \leq 1 \} \rightarrow {\tilde{X}_{\infty}+\tilde{X}_{2,\sigma}} $, of finite dimensional range such that $\xi(1/p) =\sigma^{-\frac{1}{2}}\hat{\eta}_\Omega$ \cite{10.2307/1995285},
	\begin{equation*}
		\| {\xi(it)} \|_{\infty}=\| {\xi(it)} \|_{\widehat{L}_\infty(\sigma)} {<} 1, \pl \text{ and } \pl \| {\xi(1 + i t)} \|_{\widehat{L}_2(\sigma)} {<} 1 \pl.
	\end{equation*}
	Let us define the analytic function $x(z) := \sigma^{z/2}\xi(z)(\Id_{\cH} \otimes L_k^{\frac{p z}{2}}) $. We have that
	\begin{equation*}
		{\|x(it)\|_\infty=\| \sigma^{\frac{it}{2}} {\xi(it)}(\Id_\cH \otimes L_{k}^{\frac{i t p}{2}}) \|_{\infty}} {= \| {\xi(it)} \|_{\infty} \,<\,} 1 \pl.
	\end{equation*}
	Next, since $\xi(1 + it) \in V(\cH)\otimes_h\cH^r$, there is $\Omega'\in \mathcal{T}_2(\cH)$ such that $\xi(1 + it)\equiv \hat{\eta}_{\Omega'}$. Therefore,
	\begin{align*}\label{eq:2norm}
		\|x(1+it)\|_2^2&=\|\sigma^{\frac{1}{2}}\xi(1+it)(\Id_\cH\otimes L_k^{\frac{p}{2}})\|_2^2\\
		&=\|\sigma^{\frac{1}{2}}\hat{\eta}_{\Omega'}(\Id_\cH\otimes L_k^{\frac{p}{2}})\|^2_{2} \\
		&= \int_G |k(g)|^{p} { \|\sigma^{\frac{1}{2}}\,u(g)^* {\Omega}'\|^2_2} {\,d\mu_G(g)}\\
		&= \int_G |k(g)|^{p} \|u(g)^*\,\sigma^{\frac{1}{2}}\,{\Omega}'\|^2_2\,{d\mu_G(g)} \,\\
		&{=}\, \| k^{\frac{1}{2}} \|_{L_{2p}(\mu_G)}^{2p} \,\|\sigma^{\frac{1}{2}}\Omega'\|_2^2 \\
		&{=}\, \| k^{\frac{1}{2}} \|_{L_{2p}(\mu_G)}^{2p} \,\|\sigma^{\frac{1}{2}}\,\xi(1+it)\|_2^2\\
		&< \| k^{\frac{1}{2}} \|_{L_{2p}(\mu_G)}^{2p}\,,
	\end{align*}
	where in the fourth line, we used that $\sigma\in N_{fix}$. Then, using Stein's interpolation theorem:
	\begin{equation*}
		\begin{split}
			\| \sigma^{-\frac{1}{2\hat{p}}} \hat{\eta}_\Omega(\Id_\cH\otimes L_k^{\frac{1}{2}}) \|_{2p}&= \| x({1/p}) \|_{2p}	 \leq \|k^{\frac{1}{2}}\|_{L_{2p}(\mu_G)}\,,
		\end{split}
	\end{equation*}
	and the result follows after rescaling.

\end{proof}
\begin{proof}[Proof of \Cref{comparison}] We use the notation of the theorem. The following holds for $\rho=|\Omega|^2$:
	\begin{align*}
	&	\|\sigma^{-\frac{1}{2\hat{p}}} \,\hat{\eta}_\Omega\,(\Id_{\cH}\otimes L_{k^{\frac{1}{2}}})\|_{\cT_{2p}( \cH\otimes {L}_2(\mu_G),\cH')}^2\\
		&=  \| \sigma^{-\frac{1}{2\hat{p}}}\,\hat{\eta}_\Omega\,(\Id_{\cH}\otimes L_{k})\,\hat{\eta}_\Omega^*\, \sigma^{-\frac{1}{2\hat{p}}}\|_{\cT_p(\cH')}\\
		&= \left\|\int_G k(g)\, \sigma^{-\frac{1}{2\hat{p}}}\,u(g)^*\,\rho\,u(g) \,d\mu_G(g)\sigma^{-\frac{1}{2\hat{p}}} \right\|_{\cT_p(\cH')} \\
		&= \| \sigma^{-\frac{1}{2\hat{p}}} \, \Phi_{k}(\rho)\, \sigma^{-\frac{1}{2\hat{p}}} \|_{\cT_p(\cH')} \, ,
	\end{align*}
	where $L_{k}:L_2(\mu_G)\to L_2(\mu_G)$ denotes the operator of multiplication by $k$. Therefore, for the claim to hold, it suffices to show that for $1 \leq p \leq \infty$,
	\begin{align*}
	&	\|\sigma^{-\frac{1}{2\hat{p}}} \,\hat{\eta}_\Omega\,(\Id_{\cH}\otimes L_{k^{\frac{1}{2}}})\|_{\cT_{2p}( \cH\otimes {L}_2(\mu_G),\cH')}\\
		&\le \|k^{\frac{1}{2}}\|_{{L}_{2p}(\mu_G)}\,	\|\sigma^{-\frac{1}{2\hat{p}}} \,\hat{\eta}_\Omega \|_{\cT_{2p}( \cH\otimes {L}_2(\mu_G),\cH')}\,.
	\end{align*}	
	This follows from Lemma \ref{lem:interpscale}. Differentiation at $p=1$ gives the entropic inequality (\ref{thatstheoneforchannel}).
\end{proof}

\section{Transference of hypercontractivity}\label{appendix_HC}

In the classical theory and compared to ultracontractivity (UC), it is wellknown that an even finer control of the decoherence time can be realized using the \textit{hypercontractivity} property (HC) of the semigroup at intermediate times, using \Cref{decomposition}. Hypercontractivity is concerned with controlling the term 
\begin{align*}
&\|T_t:L_\infty^q(\Nfix\,\subset\,\Bcal(\Hcal))\to L_\infty^2(\Bcal(\Hcal))\|\\
&~~~=\|T_t:L_2(\Nfix\,\subset\,\Bcal(\Hcal))\to L_2^p(\Bcal(\Hcal))\|\,,
\end{align*}
where $p$ is the H\"older conjugate of $1\leq q\leq 2$ and where we assume that $(T_t)_{t\geq0}$ is selfadjoint. We insist that the situation is more tricky than in the UC case, because as mentioned in \Cref{lem_CB_commutative}, the only operator norms that transfer to the cb case are when the image is $L_\infty$.\\

Hypercontractivity has been originaly studied in the quantum case in \cite{OZ99} in the context of spin system, than in \cite{[KT13]} for decoherence time. Up to now, very few examples are known where it is possible to estimate the hypercontractive constant. Transference provides a new way to do so. It must however be noted that when concerned with applications as the ones presented in this article, a direct application of transference lead to better result than first proving HC via transference and then applying transference.\\

Thus, our motivation in this section is simply to show that using transference, one can obtain bounds on the hypercontractive (and thus also log-Sobolev) constants for a large class of QMS: the transferred QMS. Only few such bounds existed previously in the literature and only in the primitive case (see \cite{temme2014hypercontractivity,BR18}) so we thought it meaningful to write this technique explicitily in this article. Our main result is the following one.

\subsection{Main result}

\begin{definition}
	We say that the quantum Markov semigroup $(T_t)_{t\ge 0}$ on $\cB(\cH)$ is \emph{(weakly) hypercontractive} if there exist two non-negative constants $c>0,~d\geq1$ such that for all $t\geq \frac c2 \ln (p-1)$:
	\begin{equation}\tag{$\HC\,(c,d)$}
		\norm{T_t:\,L_2(\Bcal(\Hcal))\to L_2^p\left(\Nfix\subset\Bcal(\Hcal)\right)}\leq d^{\frac 12 - \frac 1p}
	\end{equation} 
\end{definition}

The study of hypercontractivity and its application to the estimation of the decoherence time was the subject of the recent article \cite{BR18}. We can slightly improve their results (Proposition 7.1) using \Cref{decomposition} and transference.
\begin{proposition}\label{prop_HC_decotime}
	Assume that the transferred QMS $(T_t)_{t\geq0}$ on $\cB(\cH)$ is $\HC(c,d)$. Assume furthermore that the group $G$ is finite with $|G|\gl e$. Then for all $\rho\in\cD(\cH)$ and all $\kappa>0$, 
	\begin{align*}
	&\norm{T_t^\dagger(\rho)-E_{fix}^\dagger(\rho)}_1\leq \sqrt{d}\,e^{1-\kappa}\,,\,\\
&	t=\frac c2\ln\ln|G|+\frac{\kappa}{\lambda_{\min}(\Lcal)}\,.
\end{align*}
\end{proposition}

\begin{proof}
	Remark that by transference,
	\begin{align*}
		&\norm{id:\, L_\infty^1(\Nfix\,\subset\,\Bcal(\Hcal))\to L_\infty^q(\Nfix\,\subset\,\Bcal(\Hcal))}\\
		& \leq \norm{id:\, L_\infty^1(\Nfix\,\subset\,\Bcal(\Hcal))\to L_\infty^q(\Nfix\,\subset\,\Bcal(\Hcal))}_{cb} \\
		& \leq \norm{id:\, L_1(\mu_G)\to L_q(\mu_G)} \\
		& \leq |G|^{1/\hat q}\,,
	\end{align*}
	where $\hat q$ is the H\"older conjugate of $q$. We now apply \Cref{decomposition} with $t=0$, $s=\frac c2\ln\ln|G|$ and $\hat q=1+\ln|G|$ to get:
	\begin{align*}
		v^\Lcal_{1,1}(s+r) 
		&\leq |G|^{1/\hat q}\,d^{\frac12-\frac1{\hat q}}\,e^{-r\,\lambda_{\min}(\Lcal)} \\
		& \leq \sqrt d\, e^{1-r\,\lambda_{\min}(\Lcal)}\,.
	\end{align*}
	We have used that, as $|G|\geq e$, then $(\ln|G|)/(1+\ln|G|)\leq1$, which gives $|G|^{1/\hat q}\leq e$. This concludes the proof by taking $r=\kappa/\lambda_{\min}(\Lcal)$ and $t\geq s+r$.
\end{proof}

In practice, even if it is hard to find the tightest constants for which the above functional inequality is satisfied, it is still possible to obtain good bounds based on ultracontractivity. As we can use transference on ultracontractivity, we can obtain in this way bound on the hypercontractive constants. However the bound we would obtain on the decoherence time from such an estimate would not be better from using only ultracontractivity. 

\begin{theorem}\label{theo_HC}
	Let $(S_t)_{t\geq0}$ be a reversible Markov semigroup on a compact Lie or finite group $G$, with right-invariant kernel and assume there exists $t_0\geq0$ and $M>0$ such that
	\begin{align}\label{eq1}
		&\norm{S_t-\Ebb_{\mu_G}:L_1(\mu_G)\to L_2(\mu_G)}\\
		&~~~~~~~~~~~~=\underset{g\in G}{\sup}\,\norm{h\mapsto k_t(g)-1}_{2}\leq M\,.\nonumber
	\end{align}
	Let $g\mapsto u(g)$ be a unitary representation of $G$ on some finite dimensional Hilbert space and let $(T_t)_{t\geq0}$ be the corresponding transferred QMS defined as in \Cref{eq_groupQMS}. Then $(T_t)_{t\ge 0}$ satisfies $\HC(c,\sqrt{2})$ with
	\begin{equation}\label{eq_theo_HC}
		c\leq \frac{1}{\lambda_{\min}(\Lcal)}\left(\lambda_{\min}(\Lcal)\,t_0+\ln M+1\right)\,.
	\end{equation}
\end{theorem}

The proof of \Cref{theo_HC} proceeds by consecutive uses of the transference method of \Cref{Transfer} as well as the following interpolation result:

\begin{lemma}\label{lem_estimation}
	Let $(S_t)_{t\geq0}$ be a reversible Markov semigroup on a compact Lie or finite group $G$, with right-invariant kernel. Assume that
	\begin{equation}\label{eq_lem_estimation1}
		\norm{S_{t_{0}}-E_\mu\,:\,L_2(\mu_G)\to L_{\infty}(G)}\leq M
	\end{equation}
	for some $t_{0}\ge 0$ and $M>0$. Then, the semigroup $(S_t)_{t\ge 0}$ is hypercontractive with respect to the completely bounded norm: for all $ 2\le p$, there exist $c>0$ such that for all $t\ge \frac{c}{2}\ln\,(p-1)$:
	\begin{align}\tag{$\operatorname{cHC}_2(c,\sqrt2)$}\label{cHClab1}
		\|S_t:\,L_2(\mu_G) \to L_p(\mu_G)\|_{\cb}\,\le\, {\sqrt2}^{\frac{1}{2}-\frac{1}{p}}\,,
	\end{align}
	with
	\begin{equation}\label{eq_theo_estimation2}
		c\leq \frac{1}{\lambda_{\min}(L)}\left(\lambda_{\min}(L)\,t_{0}+\ln M+1\right)\,,
	\end{equation}
	where $\lambda_{\min}(L)$ is the spectral gap of the generator $L$ of the semigroup $(S_t)_{t\ge 0}$.
\end{lemma}

\begin{remark}
	A similar statement holds if we replace $\norm{S_{t_{0}}-E_\mu\,:\,L_2(\mu_G)\to L_{\infty}(G)}$ by $\norm{S_{t_{p_0}}-E_\mu\,:\,L_2(\mu_G)\to L_{p_0}(\mu_G)}_{\cb}$ for some $p_0>2$. However, one can apply \Cref{lem_CB_commutative} only for $p_0=+\infty$ so we choose to stick to this case.
\end{remark}

\begin{proof}
	By a similar argument as in the proof of Theorem 4.7 in \cite{BR18} (see also \cite{temme2014hypercontractivity} in the case of the usual Schatten norms, and Theorem 3.9 and 3.10 of \cite{Diaconis1996} in the classical setting), we get from \Cref{eq_lem_estimation1} that a \textit{complete logarithmic Sobolev inequality} $\operatorname{cLSI}_2\,(t_0,M)$ holds (see \Cref{appendix} for the definition). By a simple adaptation of Theorem 4.5 in \cite{BR18}, we get $\operatorname{cLSI}_2\,(c,\sqrt{2})$ with $c$ given by \Cref{eq_theo_estimation2}. Gross' integration Lemma for the $\cb$ norms allows us to conclude (cf.  \cite{[BK16]} in the case of a finite group, and \Cref{grossintegration} in the case of a compact Lie group).
\end{proof}

\begin{proof}[\Cref{theo_HC}]
	The proof starts by a simple application of the transference method: first, by (iii) of \Cref{Transfer}:
	\begin{align}\label{transferencetechniquehc}
	&	\|T_t:\,L_2(\cB(\cH))\to L_2^p(\Nfix\subset \mathcal{B}(\cH))\|\\&	\le	\|T_t:\,L_2(\cB(\cH))\to L_2^p(\Nfix\subset \mathcal{B}(\cH))\|_{\cb}\nonumber	\\
		&\le\|S_t:\,L_2(\mu_G)\to L_p(\mu_G)\|_{\cb}\,.\nonumber
	\end{align}
	Then, notice that
	\begin{align*}
	&	\|S_t-E_\mu:\,L_2(\mu_G)\to L_\infty(G)\|_{\cb}\\
		&~~~~~~~~~~= \|S_t-E_\mu:\,L_2(\mu_G)\to L_\infty(G)\|\\
		& ~~~~~~~~~~=\norm{h\mapsto  k_t(h)-1}_{2}\,, \\
	\end{align*}
	where we used \Cref{lem_CB_commutative} in the first line. The result follows from a direct application of \Cref{lem_estimation}.

\end{proof}	

In Theorem 3.7 of \cite{Diaconis1996a}, the authors showed, conversely to the above theorem, how to obtain bounds of the form of (\ref{eq1}) from estimates on the log-Sobolev constant. Similar bounds were obtained from the Bakry Emery condition via Poincar\'{e}, logarithmic Sobolev and Nash inequalities. We refer to \Cref{examplesfinitegroups} for a more detailed discussion. This allows us to get hypercontractivity estimates for a QMS $(T_t)_{t\ge 0}$ from hypercontractivity of any classical Markov semigroup $(S_t)_{t\ge 0}$ from which $(T_t)_{t\ge 0}$ can be transferred: for example, the following corollary is a direct consequence of \eqref{kernelestimateformarkovchain} and \Cref{theo_HC}:

\begin{corollary}[From classical to quantum hypercontractivity]\label{logsobcq}
	Let $(S_t)_{t\ge 0}$ be a reversible Markov semigroup on a finite group $G$, with right-invariant kernel satisfying $\HC(c,0)$. Then the QMS $(T_t)_{t\ge 0}$ defined in \Cref{eq_groupQMS} satisfies $\HC( c',\sqrt{2})$, with
	\begin{align*}
		c'\le \frac{2}{\lambda_{\min}(S)}+\frac{c}{2}\ln\ln |G|.
	\end{align*}	
\end{corollary}	

\subsection{Completely bounded Gross lemma for classical diffusions}\label{appendix}

In this section, we briefly describe the proof of Gross' integration lemma relating the complete logarithmic Sobolev inequality to the hypercontractivity with respect to the completely bounded normes defined in \Cref{Lpnorm} for classical diffusions. The proof is similar to the ones of \cite{[BK16]} for quantum (and hence classical) Markov semigroups in finite dimensions (see also \cite{cheng2015new} for the case of the modified logarithmic Sobolev inequality). 

Let $(S_t)_{t\ge 0}$ be a semigroup on the algebra $L_\infty(E,\cF,\mu)$ of real bounded measurable functions on the probability space $(E,\cF,\mu)$ that is reversible with respect to $\mu$. The semigroup is described by its kernel $(k_t)_{t\ge 0}$ via \Cref{eq_conv_kernel} that we recall here:
\begin{equation}\label{eq_groupMK1}
	S_t (f)(x)=\int_E\,k_t(x,y)\,f(y)\,d\mu(y)\,.
\end{equation}
When extended to its action on the space $L_2(\mu)$ of square integrable real valued functions on $E$, the semigroup $(S_t)_{t\ge 0}$ is strongly continuous, and we denote by $(L,\,\operatorname{dom}(L))$ its associated generator. Since the domain of $L$ is not usually known in practice, we will work on a dense subspace of it. In fact, it will be convenient to assume for technical reasons that the following hypothesis, already used in \cite{B-E1,BGL14}, holds:
\begin{hypothesis}\label{hypothesisA1}
	There exists an algebra $\mathcal{A}$ of bounded measurable functions, containing all the constants, dense in all the spaces $L_p(\mu)$, $p\ge 1$ as well as in $\operatorname{dom}(L)$, that is stable under composition with multivariate smooth functions. We also assume that for any sequence $\{f_n\}$ of $\mathcal{A}$ that converges to a function $f$ in $L_2(\mu)$, and every smooth bounded function $\Phi:\RR\to\RR$ with bounded derivatives, there exists a subsequence $\{\Phi(f_{n_k})\}$ of $\{\Phi(f_n)\}$ converging towards $\Phi(f)$ in $L_1(\mu)$ and such that $L\Phi(f_{n_k})$ converges to $L\Phi(f)$ in $L_1(\mu)$.
\end{hypothesis}

For any $m\in\NN$, the algebra $\mathbb{M}_m(L_\infty(\mu))$ coincides with the algebra $L_\infty(E,\mathbb{M}_m)$ of bounded measurable functions with values in $\mathbb{M}_m$, with norm $\|f\|_{L_\infty(\mathbb{M}_m)}$ defined as $\sup_{x\in E}\|f(x)\|_{\mathbb{M}_m}$. For sake of simplicity, we denote this norm by $\|.\|_\infty$. Next, for any $f\in L_\infty(\mathbb{M}_m)$, define the following trace on $L_\infty(\mathbb{M}_m)$:
\begin{align*}
	\tau(f):= \int_E\,\frac{1}{m}\,\tr_{\mathbb{M}_m}(f(x))\,\mu(dx)\,.
\end{align*}
We denote the completions of the $L_p$ norms associated to that state $L_p(\mathbb{M}_m)$, $p\ge 1$, and denote the norms associated to it by $\|.\|_{L_p(\tau)}$. To simplify the notations, we will get rid of the indices identifying the underlying spaces and introduce the normalized traces $\ntr:=\frac{1}{m}\Tr$. The main difference to~\cite{[BK16]}  in our setting arises from the possible unboundedness of the generator $L$ of the classical semigroup $(S_t)_{t\ge 0}$, which will not be an issue as long as we carry out our differentiations in the algebras $\cA(\mathbb{M}_m)^{++}:=\{g=(g_{ij}),\,g_{ij}\in \cA\,\forall ij\in \{1,...,m\},\,g>0\}$ of positive matrix valued functions with coefficients in $\cA$, with spectrum uniformly bounded away from $0$. Then, we define the \textit{$L_q$-entropy} and the Dirichlet form as follows: given elements $f,g\in \cA(\mathbb{M}_m)^{++}$, 
\begin{align*}
	&\operatorname{Ent}_{q}(f)=\tau(f^q\log f^q)-\tau(f^q)\,\log\tau(f^q)\,,\\
	&\cE_{q,\,\id_{\mathbb{M}_m\otimes L}}(f):=\tau(f^{q-1}\,(\id_{\mathbb{M}_m}\otimes L)(f))\,.
\end{align*}	
Next, we define the notions of completely bounded hypercontractivity and of complete logarithmic Sobolev inequality: Given $q\ge 1$, the semigroup $(S_t)_{t\ge 0}$ is said to\footnote{These inequalities in particular imply the primitivity of the semigroup. One could easily extend these inequalities to non-primitive classical semigroups, which however play no role in this article.}
\begin{itemize}
	\item[-] be \textit{$q$-completely hypercontractive} if there exist $c>0,\,d\ge 1$ such that for all $ p\ge q$ and all $t\ge \frac{c}{2}\log\,\frac{p-1}{q-1}$:
	\begin{align}\tag{$\operatorname{cHC}_q(c,d)$}\label{cHClab}
		\|S_t:\,L_q(\CC\,I_{M}\subset M)\to L_q^p(\CC\,I_M\subset M)\|_{\cb}\,\le\, d^{\frac{1}{q}-\frac{1}{p}}
	\end{align}
	\item[-] satisfy a \textit{$q$-complete logarithmic Sobolev inequality} if there exist $c\ge 0$, $d\ge 1$ such that for all $m\in\NN$, and all $f\in\cA(\mathbb{M}_m)^{++}$:
	\begin{align}\tag{$\operatorname{cLSI}_q(c,d)$}\label{cLSIlab}
		\operatorname{Ent}_q(f)\le \,c\,\cE_{q,\,\id_{\mathbb{M}_m}\otimes \cL}(f)\,+\,\log (d)\,\|f\|^q_{L_q(\mathbb{M}_m)}\,.
	\end{align}	
\end{itemize}	
The equivalence between \ref{cLSIlab} and \ref{cHClab} was proved in \cite{[BK16]} in the case of a Markov chain defined on a finite sample space (and even for quantum Markov semigroups in finite dimensions). Here, we extend this equivalence to the present abstract setting, which in particular incorporates the case of classical diffusions.
\begin{theorem}\label{grossintegration}
	Let $(S_t)_{t\ge 0}$ be a classical Markov semigroup defined on the algebra $L_\infty(E,\cF,\mu)$ of bounded measurable functions on some measure space $(E,\cF,\mu)$, and assume that $\mu$ is an invariant measure of $(S_t)_{t\ge 0}$ for which $(S_t)_{t\ge 0}$ is reversible. Further assume the existence of a subalgebra $\cA$ satisfying \Cref{hypothesisA1}. Then,
	\begin{enumerate}
		\item[(i)] If $\operatorname{cHC}_{q}(c,d)$ holds, then $\operatorname{cLSI}_{q}(c,d)$ holds.
		\item[(ii)] If $\operatorname{cLSI}_{2}(c,d)$ holds, then $\operatorname{cHC}_{q}(c,d)$ holds for any $q\ge 2$.
	\end{enumerate}
\end{theorem}
We now briefly sketch a proof of \Cref{grossintegration}: As usual, the first step towards establishing a Gross lemma is to provide a formula for the differential at $p=q$ of
\begin{align*}
&	p\mapsto \|S_{t(p)}:\,L_\infty^q(\CC\,I_E\subset L_\infty)  \to L_\infty^p (\CC\,I_E\subset \cL_\infty(E)) \|_{\cb}\\
	&=\sup_m\,\|   \id_{\mathbb{M}_m}\otimes\,S_{t(p)}:\,L_q^q(\mathbb{M}_m\subset L_\infty(\mathbb{M}_m))\to\\
	&~~~~~~~~~~~~~~~~~~~~~~~~~~~~~~~~~~~~~~ L_q^p(\mathbb{M}_m\subset  L_\infty(\mathbb{M}_m))    \|\,,
\end{align*}	
for some increasing, twice differentiable function $t:[1,\infty)\to[0,\infty)$. The proof of the differentiability follows closely the one of Lemma 9 of \cite{[BK16]}: given $f\in\cA(\mathbb{M}_m)^{++}$, the $L_q^p(\mathbb{M}_m\subset L_\infty(\mathbb{M}_m))$ norms take the following useful form: 
\begin{align}\label{optimizationnorms}
&	\|f\|_{L_q^p(\mathbb{M}_m\subset L_\infty(\mathbb{M}_m))}=\\
	&\left\{\begin{aligned}
		&	\inf_{X\in \mathbb{M}_m^{++},\,\ntr(X)=1}\,\tau((X^{-\frac{1}{2r}}\,f\,X^{-\frac{1}{2r}})^p)^{\frac{1}{p}}~~~~~p\ge q\\
		&\sup_{X\in \mathbb{M}_m^{++},\,\ntr(X)=1}\,\tau((X^{\frac{1}{2r}}\,f\,X^{\frac{1}{2r}})^p)^{\frac{1}{p}}~~~~~~~p\le q\,.
	\end{aligned}\right.\nonumber
\end{align}	
where $\frac{1}{r}=\frac{1}{q}-\frac{1}{p}$. Therefore, we can restrict our analysis to the one of the following function
\begin{align*}
	F:\, [1,\infty) \times L_\infty(\mathbb{M}_m)\ni(p,g)\mapsto \tau(g^p)^{\frac{1}{p}}\,,
\end{align*}
by considering the differentiation of $F\circ G$, where
\begin{align*}
	G:\, [1,\infty)  \ni p\mapsto (p,g_p)\,,~g_p:x\mapsto X^{-\frac{1}{2r}}\,f_{t(p)}(x)\,X^{-\frac{1}{2r}}
\end{align*}	
where $f_{t(p)}:=S_{t(p)}(f)$, with $f\in\cA(\mathbb{M})^{++}$ and for some fixed $X\in \mathbb{M}_m^{++}$. Both functions $p\mapsto g_p$ (see Lemma 8 of \cite{Potatov10}) and $p\mapsto \tau(g^p)$ at $g$ fixed are differentiable with continuous derivatives. The latter holds by means of bounded convergence. Therefore, the function $p\mapsto F(p,g_p)$ itself is differentiable and the chain rule holds:
\begin{align*}
	\frac{d}{dp}\,F^p\circ G(p)=\frac{\partial}{\partial p_1}\,\tau(g_p^{p_1})|_{p_1=p}+\frac{\partial}{\partial p_2}\,\tau(p,\tau_{p_2})|_{p_2=p}\,.
\end{align*}	
The first term above simply follows from a bounded convergence theorem:
\begin{align*}
	\frac{\partial}{\partial p_1}\,\tau(g_p^{p_1})|_{p_1=p}=\tau(\,\log(g_p)\,g_p^p\,)\,.
\end{align*}
The second term arises from the differentiability in $L_1(\mathbb{M}_m)$ of the map $p\mapsto g_p$, with:
\begin{align*}
	\frac{\partial}{\partial p_2}\,g_{p_2}&=\frac{r'(p_2)}{2r(p_2)^2}\,X^{-\frac{1}{2r(p_2)}}\,\{\log(X),\,f_{t(p_2)}\}\,X^{-\frac{1}{2r(p_2)}}\\
	&~~~+t'(p_2)\,X^{-\frac{1}{2r(p_2)}}\,(\id_{\mathbb{M}_m}\otimes L)(f)\,X^{-\frac{1}{2r(p_2)}}\,.
\end{align*}
Since the function $[0,\infty)\times[0,\infty)\ni (x,y)\mapsto \frac{x^p-y^p}{x-y}$, $p\ge 1$ admits a double integral form as in \cite{Potatov10}, it follows from Lemma 8 of that same paper that $p_2\mapsto g_{p_2}^{p_1}$ is differentiable in $L_1(\mathbb{M}_m)$, and therefore so is $p_2\mapsto \tau(g_{p_2}^{p_1})$, with derivative:
\begin{align*}
	&\frac{\partial}{\partial p_2}\tau(g_{p_2}^{p})|_{p_2=p}\\
	&=\,p\,\tau\left(g_{p}^{p-1}\,  \left[\frac{1}{2p^2}\,X^{-\frac{1}{2r(p)}}\,\{\log(X),\,f_{t(p)}\}
X^{-\frac{1}{2r(p)}}	\right.\right.\\
	&~~~~~~~~~~~\left.\left.\,+t'(p)\,X^{-\frac{1}{2r(p)}}\,(\id_{\mathbb{M}_m}\otimes L)(f_{t(p)})\,X^{-\frac{1}{2r(p)}}\right]\right)\,.
\end{align*}
Since we restrict the differentiation to operator-valued functions in the algebra $\cA(\mathbb{M}_m)^{++}$, the same argument would further provide that $F\circ G$ is twice continuously differentiable, as long as the function $t$ is. The following lemma hence extends Lemma 9 of \cite{[BK16]} to the case of general classical semigroups:
\begin{lemma}
	Let $t:[1,\infty)\to[0,\infty)$ be an increasing, twice continuously differentiable function. Then, for any $X\in \mathbb{M}_m^{++}$, the function $p\mapsto \tau((X^{-\frac{1}{2r(p)}}\,f_{t(p)}\,X^{-\frac{1}{2r(p)}})^p  ) $ is twice continuously differentiable. Moreover,
	\begin{align}
	&	\frac{d}{dp}\tau((X^{-\frac{1}{2r(p)}}\,f_{t(p)}\,X^{-\frac{1}{2r(p)}})^p  )^{\frac{1}{p}}\nonumber \\
		&	=\frac{1}{p^2\,\tau(g_p^p)^{1-\frac{1}{p}}}\nonumber\\
		&\left[\tau(\log(g_p^p)g_p^p)+\ntr\left( \EE_\mu[g_p^{p}]\log X\right)-\tau(g_p^p)\log \left(   \tau(g_p^p  )\right)\right.\nonumber\\
		&\left.+p^2\,t'(p)\tau\left( g_p^{p-1}(X^{-\frac{1}{2r(p)}}     (\id_{\mathbb{M}_m}\otimes L)(f_{t(p)}) X^{-\frac{1}{2r(p)}})  \right)\right]\,.\label{eqdifferentiationnorms}
	\end{align}
	Moreover, the function $X\mapsto  \tau((X^{-\frac{1}{2r(p)}}\,f_{t(p)}\,X^{-\frac{1}{2r(p)}})^p  )$ is Fr\'{e}chet differentiable on $\mathbb{M}_m^{++}$ for all $p\in [1,\infty)$.
\end{lemma}
\begin{proof}
	The only point that remains to be proven is the Fr\'{e}chet differentiability of $X\mapsto  \tau((X^{-\frac{1}{2r(p)}}\,f_{t(p)}\,X^{-\frac{1}{2r(p)}})^p  )$, which follows from a general argument on the Fr\'{e}chet differentiability of noncommutative $L_p$ spaces, see \cite{Potatov}.
\end{proof}	
Next, we define the marginal state on $\mathbb{M}_m$ as follows:
\begin{align}\label{gammaoperopti}
	\gamma:=\frac{\EE_{\mu}[f^q]}{\tau(f^q)}\,,
\end{align}	
Defining the function $X\mapsto\tilde{G}(X)$ that associates the term in between parentheses on the right hand side of \Cref{eqdifferentiationnorms} to any operator $X\in \mathbb{M}_m^{++}$, the following lemma is a straightforward extension of Lemma 10 of \cite{[BK16]}:
\begin{lemma}
	There exists $\kappa>0$ and $K<\infty$ such that for all $p\le q$ and $X\in \mathbb{M}_m$, $\|X-\gamma\|_{L_1(\tau)}\le \kappa$,
	\begin{align*}
	&	\left|\|g_p\|_{L_p(\tau)}-\|g_q\|_{L_q(\tau)}-(t(p)-t(q))\frac{\tilde{G}(X)}{q^2 \|g_q\|_{L_q(\tau)}^{q-1}} \right|\\
		&~~~~~~~~~~~~~~~~~~~~~~~~~~~~~~~~~~~~~~~~\le K\,(t(p)-t(q))^2\,.
	\end{align*}	
\end{lemma}	
\begin{proof}
	The proof consists in a simple Taylor expansion of the function $p\mapsto \tau(g_p^p)^{\frac{1}{p}}$, and we refer to the proof of Lemma 10 of \cite{[BK16]} for more details.
\end{proof}
From the very definition of the function $\tilde{G}$, one easily derives the following formula:
\begin{align*}
	G(X)-G(\gamma)=\|g_{q}\|_{L_q(\tau)}^q\,D(\gamma\|X)\,,
\end{align*}	
where $D(\rho\|\sigma)$ denotes the (normalized) relative entropy between two densities $\rho,\sigma$:
\begin{align*}
	D(\rho\|\sigma):=\tau (\rho\,(\log\rho-\log\sigma)).
\end{align*}	
The link to the $L_p^q(\mathbb{M}_m\subset L_\infty(\mathbb{M}_m))$ norms is made in the next lemma which establishes the proximity to the density $\gamma$ of the optimizer $X$ in the definition of the norms, for $p$ close to $q$.
\begin{lemma}
	For any $0<\eps\le \kappa$, there exists $\delta>0$ such that for all $p,q\in (1,\infty)$, such that $|p-q|<\delta$, there exists $X\in \mathbb{M}_m^{++}$, $\tau(X)=1$, such that
	\begin{align*}
		&\|\cP_{t(p)}(f)\|_{L_q^p(\mathbb{M}_m\subset L_\infty(\mathbb{M}_m))}=\|g_p\|_{L_p(\tau)}\,,\\
		&\|\gamma-X\|_{L_1(\tau)}\le \eps\,.
	\end{align*}	
\end{lemma}	
\begin{proof}
	The proof is identical to the one of Lemma 11 of \cite{[BK16]} and for this reason is omitted.
\end{proof}

In the case when $p>q$, the optimizer of \Cref{optimizationnorms} can actually be further characterized:
\begin{lemma}
	There exists $\eta>0$ such that for any $q<p<q+\eta$, the function $X\mapsto \tau(g_p^p)^{\frac{1}{p}}$ is strictly convex, and there exists a unique $\tilde{X}\in \mathbb{M}_m^{++}$ such that it is identically equal to $\| f_{t(p)} \|_{L_q^p(M\subset L_\infty(\mathbb{M}_m))}$. The optimizer $\tilde{X}$ satisfies:
	\begin{align*}
		\tilde{X}=\frac{\EE_\mu\left[ X^{-\frac{1}{2r(p)}}\,f_{t(p)}\,X^{-\frac{1}{2r(p)}}  \right]}{\|X^{-\frac{1}{2r(p)}}\,f_{t(p)}\,X^{-\frac{1}{2r(p)}}  \|_{L_p(\tau)}^p}
		\,\,.
	\end{align*}	
\end{lemma}
\begin{proof}
	Once again, the proof is identical to the one of Lemma 12 of \cite{[BK16]}. In particular, it relies on the uniform continuity of Schatten $p$-norms that is known to hold in a general von Neumann algebraic context.
\end{proof}	
Combining the last two lemmas, we conclude that there exists a unique positive definite optimizer of the $L_q^p$ norms, for $q<p$ close enough, and that this minimizer is close to the operator $\gamma$ defined in
\Cref{gammaoperopti}. One can also use these results to show that the function $p\mapsto \| f_{t(p)}\|_{L_q^p(\mathbb{M}_m\subset L_\infty(\mathbb{M}_m))}$ is continuous (see Lemma 13 of \cite{[BK16]} for a proof). The above tools can also be use to prove the following differentiation of the $L_q^p$ norm, the proof of which we also omit since it is identical to the one of Theorem 7 of \cite{[BK16]}:
\begin{align}\label{differentiationofnorm}
	&\left.\frac{d}{dp}\,\|f_{t(p)}\|_{L_q^p(\mathbb{M}_m\subset L_\infty(\mathbb{M}_m))}\right|_{p=q}=\frac{1}{q^2\,\| f_{t(q)} \|_q^{q-1}}\\
	&\left[ \tau\left(f_{t(q)}^q\,\log f_{t(q)}^q\right)  -\ntr\left(\EE_\mu[f_{t(q)}^q]  \log(\EE_\mu[f_{t(q)}^q])\right)\right.\nonumber
	\\&~~~~~~~~~~~~~~~~~~~~~~~~~~~~~\left.+q^2t'(q)\tau\left(f_{t(q)}^{q-1}\,L(f_{t(q)})\right)  \right]\,.\nonumber
\end{align}	
This differentiation is the key tool to prove \Cref{grossintegration}: we first assume the following two results hold:
\begin{lemma}\label{linkHCLSIq}
	\begin{enumerate}
		\item[(i)] If $\operatorname{cHC}_{q}(c,d)$ holds, then $\operatorname{cLSI}_{q}(c,d)$ holds.
		\item[(ii)] If $\operatorname{cLSI}_{p(t)}(c,d)$ holds for all $t\ge 0$, then $\operatorname{cHC}_{q}(c,d)$ holds.
	\end{enumerate}
\end{lemma}
\begin{proof}
	The proof of these implications uses \Cref{differentiationofnorm} and is identical to the one of Theorem 4 of \cite{[BK16]}. The only difference resides in $(ii)$ where one invokes the density of $\cA$ in all the $L_p$ spaces in order to show that hypercontractivity holds for any initial bounded operator valued function $f=f_0$.
\end{proof}
The reduction to $q=2$ follows from a standard Stroock-Varopoulos inequality relating the Dirichlet form $\cE_{\id_{\mathbb{M}_m\otimes L}}(f,f^{q-1})$ to $\cE_{\id_{\mathbb{M}_m\otimes L}}(f^{\frac{q}{2}},\,f^{\frac{q}{2}})$. 
\begin{lemma}\label{regularity}
	For any $f\in\cA(\mathbb{M}_m)^{++}$, and any $q> 1$:
	\begin{align*}
		\cE_{2,\,\id_{\mathbb{M}_m}\otimes \cL}(f)\le \frac{q^2}{4(q-1)}\cE_{q,\,\id_{\mathbb{M}_m}\otimes \cL}(f)\,.
	\end{align*}	
\end{lemma}	

\begin{proof}
	Such an inequality was derived under various conditions in the classical and quantum case (see e.g.~Proposition 3.1 of \cite{bakry1994hypercontractivite}, or \cite{BarEID17}) and readily extends to the finite von Neumann algebraic case.
\end{proof}
The above lemma allows us to show that $\operatorname{cLSI}_2$ implies $\operatorname{cLSI}_q$, so that the proof of \Cref{grossintegration} becomes a simple consequence of \Cref{linkHCLSIq}.

\section{Classical Markov semigroups}\label{examplesfinitegroups}
In this appendix we will briefly review the definitions of some classical Markov semigroups we discussed in the main text. We also list some of the functional inequalities known for these semigroups.
\subsection{Finite groups}\label{finitegroups}
Given a finite group and a discrete time Markov chain of kernel $k(g,h)=k(gh^{-1})$, consider the kernel of the associated continuous time chain $(S_t)_{t\ge 0}$ defined by
\begin{align}\label{ktfromk}
	k_t(x,y)=|\,G\,|\,\operatorname{exp}(-t\,(\,\operatorname{id}-k))(x,y)\,.
\end{align}
By construction, this kernel is right-invariant, and the theory developed in \Cref{sect2} applies. In Theorem 3.7 of \cite{Diaconis1996a}, the authors showed how to obtain bounds of the form of (\ref{eq1}) from estimates on the log-Sobolev constant. Adapting their result to our setting, they showed that for a reversible Markov semigroup $(S_t)_{t\ge 0}$ with associated right-invariant kernel $(k_t)_{t\ge 0}$ on a finite group of cardinality $|G|>3$\footnote{The conventions in \cite{Diaconis1996} are slightly different from the ones that we use in this article: in particular, their constant $\alpha$ is related to our weak log-Sobolev constant $c$ as follows: $2\alpha c=1$.},
\begin{align}\label{kernelestimateformarkovchain}
	\sup_{g\in G}	\| G\ni h\mapsto k_t(gh^{-1})-1 \|_2 \le \operatorname{e}^{1-\gamma},
\end{align}
for
\begin{align*}
	t=\frac{c}{2}\ln\ln |G| +\frac{\gamma}{\lambda(\Delta)},~~\gamma>0,
\end{align*}	
where $\lambda(\Delta)$ denotes the spectral gap of the generator of the chain and $c$ its log-Sobolev constant. Thus, a bound on the spectral gap and log-Sobolev constant are sufficient for our purposes.
We now list some of the constants for some Markov semigroups which are of interest in quantum information theory.
\paragraph{The hypercube:} \label{hypercube}
Let $G=\mathbb{Z}_2^n$ and define the following classical Markov chain: for $i=1,...,n$, let $e_i$ be the vector in $\mathbb{Z}_2^n$ with all coordinates $0$ but in the $i$th coordinate, which is set to be $1$. Next, define a probability mass function $Q$ on $\mathbb{Z}_2^n$ by setting $k(0)=k(e_i)=1/(n+1)$ for $i=1,..,n$, and $k(x)=0$ otherwise. In words, at each time, the discrete-time chain jumps from one vertex to a neighboring one with probability $1/(n+1)$, and stays where it was with same probability.

The strong logarithmic Sobolev constant and the spectral gap for this chain are known \cite{Diaconis1996,diaconis1988group}:
\begin{align*}
	\frac{1}{c\,(S^{\operatorname{Hyp}})}=\lambda(S^{\operatorname{Hyp}})=\frac{1}{n+1}.
\end{align*}
A direct application of \Cref{logsobcq} shows that any associated QMS $(T_t^{\operatorname{Sym}})_{t\ge 0}$ obtained from $(S_t^{\operatorname{Hyp}})_{t\ge 0}$ via \Cref{eq_groupQMS} satisfies $\HC(c,\sqrt{2})$, with
\begin{align*}
	c\le 2\,(n+1)+\frac{(n+1)\log(n\log 2)}{2}.
\end{align*}

\paragraph{The finite circle:}
We now consider the simple random walk on $G=\mathbb{Z}_m$ with $m\ge 4$, of associated kernel $k(x,x\pm1)=1/2$ and uniform stationary measure. The spectral gap of the corresponding continuous time Markov chain $(S^{\operatorname{Cir}}_t)_{t\ge 0}$ is given by the formula \cite{Diaconis1996}
\begin{align*}
	\lambda(S^{\operatorname{Cir}})=1-\cos\frac{2\pi}{m}
\end{align*}
It was shown in \cite{Diaconis1996} that $(S^{\operatorname{Cir}}_t)_{t\ge 0}$ satisfies the following bound:
\begin{align*}
&	\| S^{\operatorname{Cir}}_t-\mu_G:\,L_2(\mathbb{Z}_m)\to L_\infty(\mathbb{Z}_m)\|_{cb}^2\\
	&~~~~~= 	\| S_t^{\operatorname{Cir}}-\mu_G:\,L_2(\mathbb{Z}_m)\to L_\infty(\mathbb{Z}_m)\|^2\\
	&~~~~~\le 2\left( 1+\frac{\sqrt{5}m}{8\sqrt{\pi t}}\right)\operatorname{exp}\left({-\frac{16\pi^2\,t}{5m^2}}    \right)+\frac{m+1}{2}\operatorname{e}^{-2t}.
\end{align*}
In particular, in the case $m\ge 5$, the above expression yields the following simpler bound for $t_\infty=5m^2/16\pi^2$:
\begin{align*}
	\| S_{t_\infty}^{\operatorname{Cir}}-\mu_G:\,L_2(\mathbb{Z}_m)\to L_\infty(\mathbb{Z}_m)\|_{cb}^2\le \operatorname{e}.
\end{align*}

Now, chose the uniform random walk of kernel $k(x,y)=1/m$ for any $x,y\in \ZZ_m$. This is a special case of the Markov chain studied in Theorem A.1 of \cite{Diaconis1996}, for which the strong log-Sobolev constant $c(K)$ was shown to be equal to\footnote{We recall once again that our definition of the strong log-Sobolev constant $c(K)$ corresponds to $1/2\alpha$ in \cite{Diaconis1996}.}
\begin{align}\label{eq5}
	c(K)=\frac{\log\left( m-1  \right)}{2-4/m},~~~~~~~~~ \lambda(S)=1-\frac{1}{m}.
\end{align}

\paragraph{Random Transpositions}
In \Cref{examplessec} we considered two random transposition models on the permutation group $\Si_n$:
\[ L^{RT}(f)(\omega) \lel \frac{1}{n}(\sum_{ij} f(\omega )-f(\si^{ij}\omega) \pl, \]
where $\sigma_{ij}$ is a swap of the $i$th and $j$th subsystems, and
\[ L^{NN}(f)(\omega) \lel \sum_{j=1}^n f(\omega)-f(\si_{j(j+1)}\omega) \pl .\]
We refer to \cite{Yau} for
\[ t_{1,\infty}^{L^{RT}}(\eps) \kl  c(\ln n-\ln \eps) \pl \]
for some constant $c$. In \cite{diaconis1988group}, it was shown that the spectral gap of the corresponding continuous time semigroup $(S_t^{\operatorname{RT}})_{t\ge 0}$ is
\begin{align*}
	\lambda_{\min}(S^{\operatorname{RT}})=\frac{2}{n^2}.
\end{align*}	
More recently, it was also shown that the MLSI constant for this transposition model satisfies the following bounds \cite{gao2003exponential}:
\begin{align}\label{transpomodelMLSI}
	\frac{2}{n^2}\le \alpha_1(S^{\operatorname{RT}})\le \frac{4}{n^2}\,.
\end{align}

\section*{Acknowledgment}

I.B. is partially supported by French A.N.R. grant: ANR-14-CE25-0003 "StoQ". M.J and N.L. are partially supported by DMS NSF 1800872.
D.S.F. acknowledges financial support from VILLUM FONDEN via the QMATH Centre of Excellence (Grant no. 10059), the graduate program TopMath of the Elite Network of Bavaria, the 
TopMath Graduate Center of TUM Graduate School at Technische Universit\"{a}t M\"{u}nchen and by the 
Technische Universit\"at M\"unchen – Institute for Advanced Study,
funded by the German Excellence Initiative and the European Union Seventh Framework
Programme under grant agreement no. 291763. Moreover, D.S.F. acknowledges support from the QuantERA ERA-NET Cofund in Quantum Technologies implemented within the European Union's Horizon 2020 Programme (QuantAlgo project) via the Innovation Fund Denmark. 
C.R. acknowledge financial support from the TUM university Foundation Fellowship. CR acknowledges support by the DFG cluster of excellence 2111 (Munich 
Center for Quantum Science and Technology).

\ifCLASSOPTIONcaptionsoff
  \newpage
\fi

\bibliographystyle{IEEEtran}
\bibliography{biblio}

\end{document}